\documentclass[12pt]{article}
\usepackage{amsmath,amssymb,amsthm,amsfonts,natbib,bm,setspace,graphics,float,graphicx,url,caption,multicol,longtable,lscape,verbatim,enumerate,color}
\usepackage[ruled,lined,boxed]{algorithm2e}
\usepackage[english]{babel}
\usepackage[pagewise,mathlines]{lineno}
\usepackage{chapterbib}
\usepackage[letterpaper, margin=1in]{geometry}

\newtheorem{thm}{Theorem} %[section] If you want theorem numbered
\newtheorem{lemma}{Lemma} %%    with section number.
\newtheorem{cor}{Corollary}%[section]
\newtheorem{definition}{Definition}%[section]
\newtheorem{prop}{Proposition}%[section]
\newtheorem*{prop*}{Proposition}

\newtheorem{assumption}{Assumption}
\newcommand{\margmax}{\mathrm{argmax}}
\newcommand{\margmin}{\mathrm{argmin}}

\newcommand{\bG}{\mathbf G}

\newcommand{\tR}{\mathbb{R}}

\newcommand{\cW}{\mathcal{W}}
\newcommand{\cP}{\mathcal{P}}
\newcommand{\bP}{\mathbf{P}}

\newcommand{\bB}{\mathbf{B}}
\newcommand{\hatbB}{\widehat{\bB}}
\newcommand{\bmb}{\bm{b}}
\newcommand{\hatbmb}{\hat{\bmb}}

\newcommand{\cN}{\mathcal{N}}

\newcommand{\bM}{\mathbf{M}}
\newcommand{\bN}{\mathbf{N}}
\newcommand{\bL}{\mathbf{L}}
\newcommand{\bn}{\bm{n}}

\newcommand{\cL}{\mathcal{L}}
\newcommand{\bS}{\mathbf{S}}
\newcommand{\bV}{\mathbf{V}}

\newcommand{\br}{\bm{r}}

\newcommand{\hatbS}{\widehat{\mathbf{S}}}
\newcommand{\hatbM}{\widehat{\mathbf{M}}}
\newcommand{\bSigma}{\bm{\Sigma}}

\newcommand{\bLambda}{\bm{\Lambda}}
\newcommand{\bepsilon}{\bm{\epsilon}}

\newcommand{\be}{\bm{e}}
\newcommand{\hatbe}{\hat{\be}}
\newcommand{\bE}{\mathbf{E}}
\newcommand{\bI}{\mathbf{I}}

\newcommand{\bA}{\mathbf{A}}
\newcommand{\bC}{\mathbf{C}}
\newcommand{\bc}{\bm{c}}

\newcommand{\bO}{\mathbf{O}}

\newcommand{\bo}{\mathbf{o}}
\newcommand{\cO}{\mathcal{O}}

\newcommand{\bU}{\mathbf{U}}
\newcommand{\hatbU}{\widehat{\bU}}
\newcommand{\bD}{\mathbf{D}}

\newcommand{\rank}{\mathrm{rank}}
\newcommand{\diag}{\mathrm{diag}}

\newcommand{\bX}{\mathbf{X}}

\newcommand{\hatbW}{\widehat{\bW}}
\newcommand{\bY}{\mathbf{Y}}
\newcommand{\by}{\bm{y}}
\newcommand{\bw}{\mathbf{w}}
\newcommand{\bW}{\mathbf{W}}

\newcommand{\hatbL}{\widehat{{\bf L}}}
\newcommand{\hatbw}{\widehat{\bw}}
\newcommand{\bR}{\mathbf{R}}

\newcommand{\bK}{\mathbf{K}}

\newcommand{\bx}{\bm{x}}
\newcommand{\barbx}{\bar{\bx}}

\newcommand{\bz}{\bm{z}}

\newcommand{\bTheta}{\bm{\Theta}}

\newcommand{\hatbSigma}{\widehat{\bm{\Sigma}}}
\newcommand{\bPsi}{\bm{\Psi}}

\newcommand{\bZ}{\mathbf{Z}}

\newcommand{\bs}{\bm{s}}
\newcommand{\hatbs}{\hat{\bs}}
\newcommand{\bzero}{\bm{0}}
\newcommand{\bone}{\bm{1}}

\def\Var{{\rm Var}\,}
\def\En{{\mathbb E}_n\,}
\def\Cov{{\rm Cov}\,}

\def\eqd{\,{\buildrel d \over =}\,} 
 
\def\as{\,{\buildrel a.s. \over \longrightarrow}\,}

\newcommand{\tn}[1]{\textnormal{#1}}
\newcommand{\iv}{{1\hspace{-2.5pt}\tn{l}}}

\def\Var{{\rm Var}\,}
\def\E{{\rm E}\,}

\def\Cov{{\rm Cov}\,}

\newcommand{\cJ}{\mathcal{J}}
\newcommand{\cJn}{\mathcal{J}_n}

\newcommand{\cB}{\mathcal{B}}

\newif\ifblind
\newif\ifunblind

\unblindtrue

\title{Linear Non-Gaussian Component Analysis via Maximum Likelihood}
\ifblind
\author{author here}
\else
\author{\smaller{Benjamin B. Risk$^{1,2,3}$, David S. Matteson$^{1}$, David Ruppert$^{1}$} \\
\smaller $^{1}$Department of Statistical Science, Cornell University \\
\smaller $^{2}$SAMSI, Research Triangle Park, North Carolina and the \\
\smaller Department of Biostatistics, University of North Carolina, Chapel Hill \\
\smaller $^{3}$Current Address: Department of Biostatistics and Bioinformatics, Emory University}
\fi

\date{}

\begin{document}
\maketitle 
%\linenumbers

\onehalfspacing

\begin{abstract}
Independent component analysis (ICA) is popular in many applications, including cognitive neuroscience and signal processing. Due to computational constraints, principal component analysis is used for dimension reduction prior to ICA (PCA+ICA), which could remove important information.  The problem is that interesting independent components  (ICs) could be mixed in several principal components that are discarded and then these ICs cannot be recovered. We formulate a linear non-Gaussian component model with Gaussian noise components. To estimate this model, we propose likelihood component analysis (LCA), in which dimension reduction and latent variable estimation are achieved simultaneously. Our method orders components by their marginal likelihood rather than ordering components by variance  as in PCA. We present a parametric LCA using the logistic density and a semi-parametric LCA using tilted Gaussians with cubic B-splines. Our algorithm is scalable to datasets common in applications (e.g., hundreds of thousands of observations across hundreds of variables with dozens of latent components). In simulations, latent components are recovered that are discarded by PCA+ICA methods. We apply our method to multivariate data and demonstrate that LCA is a useful data visualization and dimension reduction tool that reveals features not apparent from PCA or PCA+ICA. We also apply our method to an fMRI experiment from the Human Connectome Project and identify artifacts missed by PCA+ICA.  We present theoretical results on identifiability of the linear non-Gaussian component model and consistency of LCA. \\

\end{abstract}

\doublespacing

\section{Introduction}
The classic independent component analysis (ICA) model is $\bX = \bM \bS$ where $\bX$ is an observed vector, $\bS$ is a latent vector of independent random variables, and $\bM$ is a square matrix called the mixing matrix.  It is assumed that we have a sample $\{\bx_i\}$, $i=1,\dots,n$, with corresponding latent $\{\bs_i\}$.  The goal is to estimate $\bM$ and $\{\bs_i\}$.  Popular ICA methodology does not directly attempt to find components that are independent but rather components that are as non-Gaussian as possible by maximizing an approximation of negentropy \citep{hyvarinen2000independent}. The principle here is that any sum of ICs will be closer to Gaussian distributed than the ICs themselves. Thus, $\{\bs_i\}$ are correctly recovered if they maximize some measure of non-Gaussianity. Moment or cumulant-based methods \citep{cardoso1993blind,virta2015joint}, kernel methods \citep{bach2003kernel}, maximum likelihood methods \citep{chen2006efficient,samworth2012independent}, and methods that directly minimize a measure of dependence \citep{stogbauer2004least,mattesontsay2012} have also been developed.

Transformations that maximize non-Gaussianity play a prominent role in many applications including signal processing \citep{bell1995information}, estimating brain networks \citep{beckmann2012modelling}, face recognition \citep{bartlett2002face}, and artifact removal \citep{griffanti2014ica}. In practice, dimension reduction using PCA is applied to the observations $\{\bx_i\}$ prior to classic ICA (hereafter, PCA+ICA) to meet the assumption of square mixing and to reduce computational costs \citep{hyvarinen2001independent}. PCA+ICA is commonly used to identify brain ``networks'' in functional magnetic resonance imaging (fMRI)  \citep{beckmann2012modelling}, where here a brain network is a set of locations that exhibit similar temporal behavior. However, PCA preprocessing can discard parts of the brain networks \citep{green2002pca}. PCA+ICA is also used to identify artifacts in single-subject fMRI to improve sensitivity and specificity in subsequent group-level analyses \citep{pruim2015ica}. Even though the results from the two-stage PCA+ICA approach have been useful in the applied sciences, our data applications show that a single analysis that uses non-Gaussianity for both dimension reduction and extracting certain latent components (LCs; see below) improves estimation.

We propose linear non-Gaussian component analysis (LNGCA). Consider a sample $\{\bx_i,\bs_i,\bn_i\}$,  $i=1,\ldots,n$, of the random variable:
\begin{linenomath*}
 \begin{align}\label{eq:LNGCAmodel}
  \bX = \bM_{\bS} \bS + \bM_{\bN} \bN
 \end{align}
 \end{linenomath*}
where $\bX \in \tR^T$; $\bS \in \tR^Q$ is a vector of mutually independent non-Gaussian random variables with $1 \le Q \le T$; $\bM_\bS \in \tR^{T \times Q}$; $\bM_\bN\in \tR^{T \times(T-Q)}$; $\bM = [\bM_\bS, \bM_\bN]$ (the concatenation of $\bM_\bS$ and $\bM_\bN$) is full rank; and $\bN$ is $(T-Q)$-variate normal. Note that in classic (noise-free) ICA, $\bS \in \tR^{T-1}$ or $\tR^T$ and $\bN \in \tR^1$ or equals zero. In LNGCA, the dimension of the image of $\bM_\bN$ is $T-Q$, whereas noisy ICA (discussed below) assumes the dimension is $T$. One observes $\{\bx_i\}$ while $\{ \bs_i\}$ and $\{\bn_i\}$ are latent.  We assume $\E \bS = \bzero$ and $\E \bN = \bzero$ in \eqref{eq:LNGCAmodel}, such that it is without loss of generality that we assume $\E \bX = \bzero$. In practice, data are centered by their sample mean. Our goal is to estimate $\bM_\bS$ and the realizations $\{\bs_i\}$ of $\bS$, which we call latent components (LCs).

\subsection{Motivation for LCA}
We estimate the LNGCA model using a maximum-likelihood framework, which we call likelihood component analysis (LCA). We introduce this new term to emphasize that our method uses a likelihood as the pertinent measure of information to achieve dimension reduction. The components are ordered according to a parametric or semi-parametric likelihood rather than by variance as in PCA. By simultaneously performing dimension reduction and latent variable estimation, we will demonstrate through simulations and two real applications that estimation of the proposed model allows the discovery of non-Gaussian signals discarded by other methods. %Non-Gaussian signals are often discarded by PCA+ICA when they are associated with small variance.
When the motivating scientific problem has a low signal-to-noise ratio, LCA is particularly well-suited to recovering the non-Gaussian signals.

The idea behind LCA is to use the marginal likelihoods rather than marginal variances as the measure of information when defining latent components, since low-variance signal may be removed by PCA. Among the class of absolutely continuous random variables with mean zero and unit variance, the standard Gaussian density has maximum differential entropy \citep{cover2006elements}. Consequently, when the non-Gaussian components in the LNGCA model belong to this class of random variables, the expected values of their marginal likelihoods are larger.
% PCA can be viewed as an approach that fixes the density as Gaussian with population mean and variance fixed at the sample versions, and then the first PC corresponds to the direction that \emph{minimizes} the Gaussian likelihood of the transformed data \citep{bolton1999characterization}.
Our approach is to constrain the latent distributions to have unit variance, which allows both the marginal likelihoods and $\bM_\bS$ to be estimated. Then the latent component with the highest likelihood, i.e., lowest entropy, contains the most information, and the Gaussian components will have the smallest marginal likelihoods.
%\vspace{-0.2in}
\subsection{Relation to other methods}

The special case in which the dimension of $\text{im}(\bM_\bS)$ is $T$ or $T-1$ and $\text{im}(\bM_\bN)$ is zero or one, respectively, is equivalent to the classic ICA model \citep{hyvarinen2000independent}. Note that one Gaussian component is allowed in classic ICA because the last component can be determined from the previous components. We ignore this technicality and for clarity, hereafter define classic ICA under the assumption that $\bM_\bS$ is full rank and $\bM_\bN = \bzero$. %if $\Cov \bX = \bI_T$, where $\bI_T$ is the $T \times T$ identity matrix, then $\bM$ is orthogonal and so its last column can be determined (up to sign) by the previous $T-1$ columns.
The case in which the dimension of $\text{im}(\bM_\bN)$ equals $T$ is the noisy ICA model, which is also called independent factor analysis (IFA) \citep{attias1999independent}. The noisy ICA model often imposes the additional assumption that $\bM_\bN = \sigma^2 \bI_{T}$.

The noisy ICA model can be approximated using a variant of PCA+ICA \citep{beckmann2004probabilistic}, where probabilistic PCA is used to estimate the number of components and achieve dimension reduction \citep{tipping1999probabilistic}. Alternatively, IFA could be used for simultaneous dimension reduction and latent variable estimation wherein the ICs are modeled as Gaussian mixtures \citep{attias1999independent}. It is difficult to apply IFA because an $m^Q$-dimensional integral, where $m$ is the number of Gaussian mixtures, must be approximated at each iteration of the EM algorithm, which quickly becomes computationally intractable. \cite{allassonniere2012stochastic} developed stochastic EM algorithms to estimate the IFA model and proposed parametric methods. \cite{guo2013hierarchical} developed a multi-subject IFA model, and \cite{shi2016modeling} extended it to include covariates and an approximate EM algorithm that linearly scales with the number of components, although their application to fMRI uses PCA. \cite{amato2010noisy} developed non-parametric density estimators of the component densities in the noisy ICA model but assume $\bM_\bS$ is semi-orthogonal, which is not realistic for our application.

Other methods exploring non-Gaussian structure in multivariate data include non-Gaussian component analysis (NGCA) and projection pursuit. NGCA is a more general case of \eqref{eq:LNGCAmodel} that allows non-linear dependence between the non-Gaussian components. However, this comes at the cost that the latent components are not identifiable. The subspace that contains the non-Gaussian signal is estimated using multiple projection pursuit indices or radial basis functions \citep{blanchard2006search,kawanabe2007new}. Since it does not estimate latent components, NGCA does not lend itself to identifying brain networks and/or artifacts. %Projection pursuit typically uses a single, fixed index.
Projection pursuit is a method without a generative model that seeks ``interesting'' directions of information by maximizing projection pursuit indices, such as kurtosis \citep{huber1985projection}. %PCA can be considered a special case in which the projection pursuit index is the sample variance \citep{bolton1999characterization}. In practice, one can not find directions for any arbitrary likelihood because the likelihood is not bounded.
\cite{miettinen2014deflation} used the deflationary FastICA algorithm to adaptively select the projection pursuit index from a family of indices for each non-Gaussian direction for the case where $Q=T$. One approach to estimating the model in \eqref{eq:LNGCAmodel} would be to sequentially estimate projection pursuit directions. However, estimates from deflationary fastICA typically have higher asymptotic variance than symmetric fastICA \citep{miettinen2015squared,miettinen2015fourth}. Overall, the LNGCA model in \eqref{eq:LNGCAmodel} is unique in that it specifies a latent variable model for the non-Gaussian signal (which we show is identifiable) while also defining a subspace containing Gaussian noise, and the LCA estimation procedure is unique because it uses a likelihood to simultaneously estimate the latent components in the presence of Gaussian noise. See Web Supplement \ref{WS:addbackground} for additional discussion of these other methods.

In Section \ref{sec:LNGCA}, we discuss the identifiability of LNGCA and a discrepancy measure to account for unidentifiable signed permutations. In Section \ref{sec:ParametricLCA}, we propose parametric LCA. In Section \ref{sec:Spline-LCA}, we propose Spline-LCA where we also estimate the latent densities. In Section \ref{sec:SimIID}, we investigate simulations when the observations of the latent variables are iid. In Section \ref{sec:SpatioTemporalSims}, we examine model robustness by applying our method to temporally and spatially structured simulated data, and we evaluate the impact of estimating the wrong number of components. In Section \ref{sec:LeafExample}, we use LCA for data visualization and dimension reduction in multivariate data from leaf characteristics. In Section \ref{sec:t-fMRI}, we estimate brain networks and artifacts from high-resolution fMRI data from the Human Connectome Project.  % In Section \ref{sec:Discussion}, we present our conclusions and discuss avenues for future research.
Code implementing our methods and proofs of the theorems appear in the Web Supplement.

\section{\smaller LNGCA}\label{sec:LNGCA}
Throughout this section we assume (for simplicity) all random variables are mean zero. Define the equivalence relation $\bB \cong \bC$ for matrices $\bB$ and $\bC$ if $\bB$ equals $\bC$ up to scaling and permutation of columns. Let ``$\eqd$'' denote equality in distribution. Let $\bS = [S_1,\dots,S_Q]^\top$. We state the assumptions of the LNGCA model below.
\begin{assumption}
 $S_1,\dots,S_Q$ are mutually independent, non-Gaussian random variables with $\E \bS = \bzero$ and $\E \bS \bS^\top = \bI_Q$.
\end{assumption}
\begin{assumption}
$\rank([\bM_{\bS},\bM_\bN])=T$
\end{assumption}
\begin{assumption}
$\bN$ is $(T-Q)$-variate normal with $\E \bN = \bzero$ and $\E \bN \bN^\top$ non-singular.
\end{assumption}
The following theorem can be established using Theorem 10.3.9 in \cite{kagan1973characterization}.
\begin{thm}\label{thm1}
Suppose $\bX$ follows the model in \eqref{eq:LNGCAmodel} with Assumptions 1-3. Then for any other representation $\bX = \bM_\bS^* \bS^* + \bE^*$ where $\bS^* \in \tR^Q$ are independent non-Gaussian components and $\bE^*$ is multivariate normal, we have: $\bM_\bS^* \cong \bM_\bS$; $\bS^* \eqd \bS$ up to scaling and permutations; $\bM_\bS \bS \eqd \bM_\bS^* \bS^*$; and $\bE^* \eqd \bM_\bN \bN$.
\end{thm}
\noindent All proofs appear in Web Supplement \ref{WS:proofs}.

From Theorem \ref{thm1}, the signal, $\bM_\bS \bS$, has a unique decomposition (on the equivalence class of scalings and permutations) into a fixed matrix and independent components. The assumption that $\bM$ is full rank is necessary to ensure the uniqueness of the distributions of the latent components, which in turn is necessary for their identifiability. Note that the noise, $\bM_\bN \bN$, does not have a unique decomposition.% (e.g., if $\bN$ comprises independent normals with equal variance, then $\bM_\bN \bO^\top$ and $\bO \bN$ for orthogonal $\bO$ is another decomposition with independent components).

Without loss of generality, we assume that $\bN$ is standard multivariate normal. Let $\{f_q\}$ be the true densities of the LCs (the signal components), which are also called the source densities.  For the purposes of this paper, we will also assume $\{f_{q}\}$ are absolutely continuous, although identifiability holds more generally.
 Denote the eigenvalue decomposition (EVD) of the covariance matrix of $\bX$ by $\bSigma = \bU \bLambda \bU^\top$. Let ${\bf{L}} = \bU \bLambda^{-1/2} \bU^\top$ be a whitening matrix (the covariance matrix of $\bL \bX$ is $\bI_T$), and define the unmixing matrix $\bW =  \bM^{-1} \bL^{-1}$ where $\bM = [\bM_\bS,\bM_\bN]$. Note that $\bW \in \cO_{T \times T}$, where $\cO_{T \times T}$ is the class of $T \times T$ orthogonal matrices. Let $\bw_q^\top$ denote the $q$th row of $\bW$, and let $\bW_\bS$ denote the first $q$ rows. Let $\phi(x)$ denote the standard normal density. Noting that $|\!\det \bW | = 1$, we have
\begin{linenomath*}
\begin{align}\label{eq:JointDensity}
 f_{\bX}(\bx | \bW,\bL) &=  \det(\bL) \prod_{q=1}^Q\! f_q\!\left({\bw_q^\top \bL \bx}\right)\prod_{k=1}^{T-Q} \phi(\bw_{Q+k}^\top \bL \bx).
\end{align}
\end{linenomath*}
Note that for a density and its corresponding row of the unmixing matrix, $\{ f_q, \bw_q \}$, we can trivially define a density $f_q^*(x) = f_q(-x)$ and vector $\bw_q^* = - \bw_q$ such that $f_q^*({\bw_q^*}^\top \bx) = f_q(\bw_q^\top \bx)$ for all $\bx \in \tR^T$. In this sense, we say the density and vector pair, $\{f_q,\bw_q\}$, is identifiable up to sign. We can now establish the identifiability of the LNGCA model.
\begin{cor}
 Suppose the linear structure model in \eqref{eq:LNGCAmodel} with density defined in \eqref{eq:JointDensity} and suppose that Assumptions 1-3 hold. Then $\{f_1,\bw_1\}, \dots, \{f_Q, \bw_Q\}$ are identifiable up to sign and ordering. Note the rows $\bw_{Q+k}$ for $k=1,\dots,T-Q$ are not identifiable.
\end{cor}

\subsection{Sign- and permutation-invariant discrepancy measure}\label{sec:PMSE}

To accomodate the identifiability limitations, we propose a novel measure of dissimilarity that uses a modification of the Hungarian algorithm to match rows of the unmixing matrix as in \cite{ilmonen2010new} and \cite{risk2014evaluation}. Unlike the Amari or minimum distance \citep{ilmonen2010new} measures, it applies to non-square unmixing matrices. We also generalize the measure to apply to matrices that may have a different number of columns, in which case the measure only compares matching columns. This measure is also used to assess convergence in our algorithms.

Consider $\bM_1 \in \tR^{T \times Q}$ and $\bM_2 \in \tR^{T \times R}$ with $Q \le R$. With slight abuse of notation, we now let $\cP_{\pm}$ be the class of $R \times Q$ signed permutation matrices, so that post-multiplication of $\bM_2$ by $\bP_{\pm} \in \cP_{\pm}$ results in a subset of $Q$ (permuted) columns of $\bM_2$ for $Q < R$. Let $||\cdot||_F$ denote the Frobenius norm. Define the sign- and permutation-invariant mean-squared error:
\begin{equation}\label{eq:PMSE}
PMSE(\bM_1,\bM_2) = \frac{1}{TQ}\;\underset{\bP_{\pm} \in \cP_{\pm}}{\margmin} ||\bM_1 - \bM_2 \bP_{\pm} ||_F^2,
\end{equation}
where $\bP_\pm$ is found using the modified Hungarian algorithm. In practice, we also standardize the columns of $\bM_1$ and $\bM_2$ to have unit norm, and thus the measure is scale invariant. Then \eqref{eq:PMSE} is equivalent to finding $\bP_{\pm}$ such that the sum of the correlations between the columns of $\bM_1$ and $\bM_2 \bP_{\pm}$ is maximized. %Another advantage of this measure is that it can be used to compare independent components directly. If $\bS_1$ is a $n \times Q$ matrix in which each row is a realization of the LC in $\tR^Q$, and if $\bS_2 \in \tR^{n \times R}$, then we define their discrepancy as $PMSE(\bS_1,\bS_2)$.
Also define $PRMSE = \sqrt{PMSE}$, i.e., permutation-invariant root mean squared error.

\section{\smaller Parametric LCA}\label{sec:ParametricLCA}

Now let $\{\bx_i\}$ be an iid sample of $\bX$, and let $\barbx = \frac{1}{n}\sum_{i=1}^n \bx_i$. %Since $\E \bX = \bzero$, we will demean the data so that $\sum_{i=1}^n \bx_i = 0$, and assume such hereafter. %Let $\tR_+^{T \times T}$ denote the set of $T \times T$ positive definite matrices.
Assume $n > T$. Let $\hatbSigma$ be the sample covariance matrix of $\{\bx_i\}$, with divisor $n$, not $n-1$. Consider its eigenvalue decomposition, $\hatbSigma = \hatbU \widehat{\bLambda} \hatbU^\top$. Then define $\hatbL = \hatbU \widehat{\bLambda}^{-1/2} \hatbU^\top$. Let $\bo_q^\top$ be the $q$th row of an orthogonal matrix $\bO$. Note that $\sum_{i=1}^n \bo_q^\top \hatbL (\bx_i - \barbx) = 0$ and $\sum_{i=1}^n \log \phi \left(\bo_q^\top \hatbL (\bx_i - \barbx) \right) = -\frac{n}{2}\left(\log 2\pi + 1 \right)$.
 Let $\cO_{Q \times T}$ be the class of $Q \times T$ semi-orthogonal matrices, which is the Stiefel Manifold. Let $p_q(x)$ denote a density used in the objective function (possibly mis-specified):
\begin{align}\label{eq:ObjectiveFunction}
\cJ_n(\bO_\bS \hatbL (\bx_i - \barbx)) = \frac{1}{n}\sum_{i=1}^n \sum_{q=1}^{Q} \log p_q \! \left({\bo_q^\top \hatbL (\bx_i - \barbx)}\right).
\end{align}
 Let $h(\bs) = \sum_{q=1}^Q \log p_q (s_q)$, where $s_q$ is the $q$th element of $\bs$. Let $\| \bs \|$ denote the Euclidean ($\ell$-2) norm. We make additional assumptions:
\begin{assumption}\label{assumption:4}
 %$\{S_q\}$ has a bounded (absolutely continuous) density $\{f_q\}$, $q=1,\dots,Q$.
 % NOTE: Changed to continuously differentiable
 (i) $p_q(s_q)<\infty$, $q=1,\dots,Q$; (ii) for all $\bs_0$ and $\bs_1 \in \tR^Q$, there exist $M>0$ and $\alpha \ge 0$ such that
\begin{equation}
\|h(\bs_0) - h(\bs_1)\| \le M \|\bs_0-\bs_1\| \,  \Big\{ 1 + \|\bs_0\| ^\alpha + \|\bs_1\|^\alpha \Big\}  \label{eq:Prop4i};
\end{equation} 
and (iii) $\E\|\bS\|^{1+\alpha} < \infty$.
\end{assumption}
\noindent Note that Assumptions \ref{assumption:4} (i) and (ii) define conditions for the densities used in the \emph{objective function} rather than the true densities. A discussion of the densities satisfying these assumptions is in Web Supplement \ref{WS:parametricLCA}. However, we first consider the case when the $Q$ true component densities are known, i.e., $p_q = f_q$:
 \begin{linenomath*}
 \begin{align}\label{eq:genericEstimator2}
  \hatbW_\bS^{Or} = \underset{\bO_\bS \in \cO_{Q \times T}}{\margmax} \;\;\; \sum_{i=1}^n \sum_{q=1}^{Q} \log f_q \! \left({\bo_q^\top \hatbL \left( \bx_i - \barbx \right)}\right),
 \end{align}
 \end{linenomath*}
so that $\hatbW_\bS^{Or}$ is an oracle (Or) estimator that cannot be used in practice. %In our Logis-LCA estimator, $\{f_q\}$ are replaced by the logistic density, while in Spline-LCA we also estimate $\{f_q\}$.

Observe that estimating $\bW_\bS$ is equivalent to estimating the LCs because $\hatbs_i = \hatbW_\bS \hatbL (\bx_i - \barbx)$ for all $v$. Thus we would like a consistent estimator of $\bW_\bS$. 
\begin{thm}\label{thm:2}
Suppose $\bX$ follows the LNGCA model in \eqref{eq:LNGCAmodel} with Assumptions~1-4. %Additionally assume $\E \bX = \bzero$ and $\E \bX \bX^\top = \bI$. % Let $\bW_\bS$ denote the first $Q$ rows of $\bM^{-1}\bL^{-1}$. 
Given an iid sample $\{\bx_i\}$, $\hatbW_\bS^{Or} \as \bW_\bS$ on the equivalence class of signed permutations.
\end{thm}

Consider the special case in which $p_q$ in \eqref{eq:ObjectiveFunction} equals the logistic density for all $q$, hereafter Logis-LCA. The Infomax algorithm can be derived as a gradient ascent algorithm for maximum likelihood ICA in which the source densities are assumed to have logistic densities. Infomax is popular in fMRI analysis, where it outperforms FastICA and JADE \citep{correa2007performance,calhoun2006unmixing}. %Under the constraint of zero mean and unit variance, the logistic density has the form
% \begin{linenomath*}
% \begin{align}\label{eq:LogisticDensity}
%  f(x) = \frac{\exp{(-x/\frac{\sqrt{3}}{\pi}})}{\frac{\sqrt{3}}{\pi} \left\{ 1 + \exp (-x/\frac{\sqrt{3}}{\pi}) \right\}^2}.
% \end{align}
% \end{linenomath*}
We define our estimator for some $Q^* \le T$ such that $Q^*$ may or may not equal $Q$. %If $Q^*=Q$ and the true source densities are logistic, then it follows from Theorem \ref{thm:2} that the estimator is consistent. We will show robustness to misspecification of $Q$ via simulations in Section \ref{sec:SpatioTemporalSims}.
%Applying \eqref{eq:LogisticDensity} to \eqref{eq:genericEstimator2} and the centered data $\{\bx_i\}$,
After simplifications, the Logis-LCA estimator of $\bW_\bS$ is defined
\begin{linenomath*}
\begin{align}
  \hatbW_\bS^{Logis} &= \underset{\bO_\bS \in \cO_{Q \times T}}{\margmax} -\sum_{i=1}^n \sum_{q=1}^{Q^*} \log \left\{ 1 + \exp \! \left( -\bo_q^\top \hatbL (\bx_i - \barbx) \frac{\pi}{\sqrt{3}} \right) \right\}. \label{eq:lcalogisW}
\end{align}
\end{linenomath*}
We maximize \eqref{eq:lcalogisW} using a modification of the symmetric fixed-point ICA algorithm \citep{hyvarinen1999fast} discussed in Web Supplement \ref{WS:fixedpoint}.

Next, we define sufficient conditions that characterize the extent to which the densities used in the estimator can mismatch the true densities while maintaining consistency. Let $r_q(s)$ and $r_q'(s)$ denote the first and second derivatives of $\log p_q(s)$. Note that $\E r_q(S_q) = \int r_q(s) f_q(s) \, ds$  since $f_q$ is the true density. %; that is, the index $q$ denotes integration with respect to the true density $f_q$ of $S_q$.
\begin{assumption}\label{assumption:mismatch}
 For all $q$, (i) $\E r_q'(S_q) - \E S_q \, r_q (S_q) < 0$; (ii) $\E r_q'(S_q)$, $\E S_q \, r_q (S_q)$, and $\E S_q^2 r_q'(S_q)$ are finite; (iii) $\log p_q(s)$ is twice continuously differentiable on the support of $S_q$; (iv) $\frac{\partial}{\partial{o_{qt}}} \E \log p_q (\bo_q^\top\bX) = \E \frac{\partial}{\partial{o_{qt}}} \log p_q (\bo_q^\top \bX)$ and $\frac{\partial}{\partial{o_{qt}}} \E X_{t} r_q (\bo_q^\top \bX) = \E \frac{\partial}{\partial{o_{qt}}} X_{t} r_q (\bo_q^\top \bX)$.
\end{assumption}
\noindent The interesting assumption here is \ref{assumption:mismatch}(i), which defines the mis-match criterion. We can check this assumption for a proposed objective function density and a set of hypothetical source densities to gain insight into the robustness of the proposed estimator, which will be done in Section \ref{sec:SimIID}. %Note assumption \ref{assumption:mismatch}(iv) is satisfied if $\log p_q(s)$ and $r_q(s)$ have bounded derivatives on the support of $S_q$. 
Note the differentiability assumption is for the proposed densities and does not need to hold for the true densities. Now consider compact neighborhoods of $\bW_\bS$ of the form $\cN_\epsilon(\bW_\bS) = \{\bO_\bS \in \cO_{Q \times T}: ||\bO_\bS - \bW_\bS||_F \le \epsilon\}$. Note that in place of $\bW_\bS$, we could define this neighborhood for any other $\bW_\bS^* \cong \bW_\bS$ (here the equivalence class is defined for sign changes and permutations of the rows of $\bW_\bS$), which is useful when using a preliminary estimator described below. %Let $\cP_{\pm}$ be the class of $Q \times Q$ signed permutation matrices, $\bP_\pm$. Equivalently, we could consider the set $\cup_{\bP_{\pm} \in \cP_{\pm}} \cN_\epsilon(\bP_{\pm}\bW_\bS)$, and then assume that all maxima correspond to a matrix on the equivalence class $\cP_{\pm}$, since the objective function behaves the same in any neighborhood $\cN_\epsilon(\bP_{\pm}\bW_\bS)$. Thus it is without loss of generality that we consider a neighborhood of $\bW_\bS$.
%DELETED NOTES ON EQUIVALENCE CLASS
Let $p(\bx) = \prod_{q=1}^Q p_q(x_q)$.
\begin{prop}\label{prop:neighbor}
 Suppose Assumptions 1-5. There exists $\cN_{\epsilon^*}(\bW_\bS)$ such that $\E \! \log p(\bO_\bS \bL \bX)$ constrained to $\bO_\bS \in \cN_{\epsilon^*}(\bW_\bS)$ is maximized at $\bW_\bS$. 
\end{prop}
\noindent Restricting the optimization space is necessary because for many source distributions, the population objective function can contain multiple maxima. In fact, when the wrong density is used in the objective function, the global maximum can correspond to the wrong unmixing matrix \citep{risk2014evaluation}. This notion of localness corresponds to the definition of the theoretical fastICA estimator found in \cite{hyvarinen1999fast}, \cite{Hyvaerinen1998}, and \cite{wei2015convergence}. Formally, define
 \begin{align}\label{eq:LocalEstimator}
  \hatbW_\bS^{Local} = \underset{\bO_\bS \in \cN_{\epsilon^*}(\bW_\bS)}{\margmax} \;\;\; \cJ_n(\bO_\bS \hatbL (\bx_i - \barbx)).
  \end{align}
Then we have consistency even when the density is mis-specified.
\begin{thm}\label{thm:3}
Suppose $\bX$ follows the LNGCA model in \eqref{eq:LNGCAmodel} with Assumptions 1-5. %Additionally assume $\E \bX = \bzero$ and $\E \bX \bX^\top = \bI$.
Given an iid sample $\{\bx_i\}$, $\hatbW_\bS^{Local} \as \bW_\bS$ on the equivalence class of signed permutations.
\end{thm}
Under additional assumptions (Assumption \ref{assumption:rootn} in Web Supplement \ref{WS:additionalasymptotics}) and using the methods from \cite{nordhausen2011deflation}, \cite{miettinen2015fourth}, \cite{miettinen2015squared}, and \cite{virta2016projection}, we derive $\sqrt{n}$-consistency, asymptotic normality, and the asymptotic variances, which appear as Theorem \ref{thm:4} and Corollary \ref{cor:av} in the Web Supplement \ref{WS:additionalasymptotics}. We also conducted simulations validating the asymptotics on finite samples; see Figure \ref{fig:Empirical_vs_Theoretical}.

We can replace the condition that optimization is over $\cN_{\epsilon^*} (\bW_\bS)$ with a two-stage estimator in which in the first stage, we use an estimator that is consistent on $\cO_{Q \times T}$ and in the second stage, we use an estimator that may improve upon the initial consistent estimate. \cite{virta2016projection} propose an estimator for the LNGCA model based on a mixture of squared third and fourth moments in which the global maximum on $\cO_{Q\times T}$ is $\sqrt{n}$-consistent under finite eighth moment assumptions.  Define the Local+Virta estimator, $\hatbW_{\bS}^{LV}$, in which the symmetric estimator from \cite{virta2016projection} is updated with a single iteration of the symmetric fixed point algorithm (Algorithm 2 in Web Supplement \ref{WS:fixedpoint}, which is an approximate Newton iteration, \citealt{hyvarinen1999fast}) defined for the objective function in \eqref{eq:LocalEstimator}. Then under the additional moment assumptions, one can obtain an estimator with the wrong likelihood that is consistent on $\cO_{Q \times T}$.

For any LCA estimator $\hatbW_\bS$ and $\hatbs_i = \hatbW_\bS \hatbL (\bx_i - \barbx)$, we also define an estimator of $\bM_\bS$:
\begin{linenomath*}
\begin{align}\label{eq:hatbMs}
 \hatbM_\bS = \underset{\bA \in \tR^{T \times Q}}{\margmin} \; \sum_{i=1}^n \|\bx_i - \barbx - \bA \hatbs_i\|_2^2.
\end{align}
\end{linenomath*}
This is the OLS solution, which here is equivalent to $\hatbM_\bS = \hatbL^{-1} \hatbW_{\bS}^\top$. Although we assume iid observations in the construction of \eqref{eq:genericEstimator2}, the LNGCA model is capable of recovering many forms of dependent data, as is also the case in ICA. This will be demonstrated in simulations.

There is a natural ordering of the LCs when the component densities are not equal, which can be viewed as ordering components by the information measured by their non-Gaussian likelihood under the constraint of unit variance. Additionally, if the LCs have non-zero finite third moments, we can assume positive skewness and then the LNGCA model is fully identifiable (as in ICA, \citealt{Eloyan2012}). We choose the sign on each row $\hatbw_q$ and corresponding $\{\hat{s}_{iq}\}$ such that $\sum_{i=1}^n \hat{s}_{iq}^3 > 0$ for $q=1,\dots,Q$. We define the LCA criteria for ordering LCs for a sample $\{\bx_i\}$:
\begin{linenomath*}
\begin{align}\label{eq:LCAorder}
\sum_{i=1}^n \log f_1 \left\{ \hatbw_1^\top \hatbL (\bx_i-\barbx) \right\} > \sum_{i=1}^n \log f_2 \left\{ \hatbw_2^\top \hatbL (\bx_i - \barbx) \right\} > \dots > \sum_{i=1}^n \log f_Q \left\{ \hatbw_Q^\top \hatbL (\bx_i - \barbx) \right\}.
\end{align}
\end{linenomath*}
%For identifiability, we force the sample third moments to be positive and order components by their likelihoods. 
If we include the Gaussian noise components, then the population analogue of \eqref{eq:LCAorder}, allowing for potentially equal source densities and assuming continuous source densities, is
\begin{linenomath*}
\begin{align*}
\E \log f_1(\bw_1^\top \bL \bX) \ge \dots \ge \E \log f_Q(\bw_Q^\top \bL \bX) > \E \log \phi(\bw_{Q+1}^\top \bL \bX) = \dots = \E \log \phi(\bw_{T}^\top\bL \bX).
\end{align*}
\end{linenomath*}
This conveniently characterizes the noise components as containing the least information.

\section{\smaller Semi-parametric LCA: Spline-LCA}\label{sec:Spline-LCA}

In this section, we use the flexible family of tilted Gaussian densities to model the LCs. The proposed model is equivalent to ProDenICA \citep{hastie2002independent} when $Q = T$. For $Q<T$, it can be shown that the likelihood extends the semiparametric likelihood in \cite{blanchard2006search} to include an independence model for the LCs (see Proposition \ref{prop:NGCA} of Web Supplement \ref{WS:ppmethods}). The independence assumption is necessary for physically and biologically useful interpretations. We chose tilted Gaussian densities with cubic B-splines because ProDenICA generally outperformed parametric and kernel ICA methods \citep{hastie2009elements,risk2014evaluation} and its algorithmic complexity is $O(n)$, which enables its application to large datasets such as fMRI.

Suppose the LCs have tilted Gaussian distributions of the form $\phi(u)e^{g(u)}$, where $g(u)$ is a twice-differentiable function. Define the log-likelihood for some $\bO \in \cO_{T \times T}$:
\begin{linenomath*}
\begin{align*}
&\ell(\bO,g_1,\dots,g_{Q^*}\mid \hatbL,\barbx,Q^*,\{\bx_i\}) = \\
&\sum_{i=1}^n \left[ \sum_{q=1}^{Q^*}  \left\{ \log \phi \left(\bo_q^\top \hatbL (\bx_i -\barbx)\right) + g_q\left(\bo_q^\top \hatbL (\bx_i - \barbx)\right) \right\} + \sum_{k=1}^{T-{Q^*}} \log \phi \left(\bo_{k+{Q^*}}^\top \hatbL (\bx_i - \barbx)\right) \right].
\end{align*}
\end{linenomath*}
This log-likelihood does not have an upper bound. We define a penalized log-likelihood that includes a roughness penalty and an additional term to ensure the solution is a density:
\begin{linenomath*}
 \begin{align}
 \label{eq:Spline-LCA} \ell_{pen}(\bO,g_1,\dots,g_{Q^*}\mid \hatbL,\barbx,Q^*,\{\bx_i\}) &= - \sum_{q=1}^{Q^*} \left\{ \gamma_q \int \{g_q''(u)\}^2 \, du + \int \phi(u) e^{g_q(u)}\, du \right\} \\
 \nonumber &\;\;\;\;+ \frac{1}{n} \sum_{i=1}^n \sum_{q=1}^{Q^*} \left\{\log \phi \left(\bo_q^\top \hatbL (\bx_i - \barbx) \right) + g_q \left(\bo_q^\top \hatbL (\bx_i - \barbx) \right)\right\},
 \end{align}
 \end{linenomath*}
 where we have dropped the noise components since they are constant for all $\bO$ but retained the Gaussian contributions to the tilted Gaussian densities, which are not constant when the data are binned as described below. %We keep track of the Gaussian portion of the tilted Gaussian because it is not constant when we discretize the integral (see below).
 Then we have the following:
\begin{prop}\label{prop:density}
Let $G$ be the class of all cubic splines $g\!: \tR \rightarrow \tR$. Consider the argmax of \eqref{eq:Spline-LCA} for $g_q \in G$. Then (i) $\int \phi(u)  e^{g_q(u)}\, du = 1$ and (ii) $\int u \phi(u)  e^{g_q(u)} \,du = 0$ for each $q$.
\end{prop}

We adapt the ProDenICA algorithm of \cite{hastie2002independent} to LCA, in which we alternate between estimating $\bW_\bS$ for fixed $\{\hat{f}_q\}$, $q=1,\dots,Q^*$, via the fixed point algorithm and estimating $\{f_q\}$ for fixed $\hatbW_\bS$ using the ``Poisson trick''. Our account largely follows the description in \cite{hastie2009elements} but for semi-orthogonal (rather than orthogonal) matrices.

Suppose $\hatbW_\bS$ is given and define $s_{vq} = \hatbw_q^\top \hatbL (\bx_i - \barbx)$. Let $u_1^*,\dots,u_{L+1}^*$ define a discretization, $[u_1^*,u_2^*),[u_2^*,u_3^*),\dots,[u_{L}^*,u_{L+1}^*)$, of the support of the tilt function of the non-Gaussian densities such that $\Delta = u_\ell^* - u_{\ell-1}^*$ for all $\ell=2,\dots,L+1$. It suffices to take $u_1^* = \min(s_{11},\dots,s_{nd})-0.1 \hat{\sigma}_z$ and $u_{L+1}^* = \max(s_{11},\dots,s_{nd})+0.1 \hat{\sigma}_z$, where $\hat{\sigma}_z$ denotes the sample standard deviation, which here is equal to one. Next, let $u_\ell = \frac{1}{2}(u_\ell^* + u_{\ell+1}^*)$. For each $q \in \{1,\dots ,Q^*\}$ and $\ell \in \{1,\dots,L\}$, define $y_{\ell q} = \sum_{i=1}^n \iv\{s_{vq} \in [u_\ell^*,u_{\ell+1}^*)\}$.

We approximate \eqref{eq:Spline-LCA} by discretizing the first integral and estimating the sum over $n$ as a weighted sum over $L$. Restricting our attention to a single $q$ and dividing by $\Delta$, %, we have
% \begin{linenomath*}
% \begin{align*}
%  -\lambda_q \int \left\{ g_q'' (u)\right\}^2 du + \sum_{\ell=1}^L \left[ \frac{y_{\ell q}}{n} \left\{ g_q(u_\ell) + \log \phi(u_\ell) \right\} - \Delta \phi(u_\ell) e^{g_q(u_\ell)} \right].
% \end{align*}
% \end{linenomath*}
% Dividing by $\Delta$, we have
\begin{linenomath*}
\begin{align}\label{eq3}
-\gamma_q \int \left\{ g_q'' (u)\right\}^2 du + \sum_{\ell=1}^L \left[ \frac{y_{\ell q}}{n \Delta} \left\{ g_q(u_\ell) + \log \phi(u_\ell) \right\} - \phi(u_\ell) e^{g_q(u_\ell)} \right]
\end{align}
\end{linenomath*}
for some penalty $\gamma_q$. This is proportional to a Poisson generalized additive model (GAM), where $\frac{y_{\ell q}}{n\Delta}$ is the response and the expected response is equal to $\phi(u_\ell)e^{g_q(u_\ell)}$. In practice, we use cubic B-splines in the \texttt{gam} package \citep{hastiegam} with \texttt{smooth.spline} and the default knot selection, where $\gamma_q$ is chosen to result in a user-specified number of effective degrees of freedom. We find that $df = 8$ and $L=100$ produce fast and accurate density estimates in simulations for a variety of densities with sample size 1,000. This method also easily scales to tens of thousands of observations. We summarize this procedure in Algorithm 1. Note that step 3 requires the first and second derivatives of the log densities of the LCs, which makes the use of B-splines convenient.
\begin{linenomath*}
\begin{algorithm}\label{algorithm:Spline-LCA}
\SetKwInOut{Input}{Inputs}
\SetKwInOut{Output}{Output}
\SetKwFor{For}{for}{}{endfor}
\SetKwFor{While}{while}{}{endwhile}
\caption{The Spline-LCA algorithm.}
\Input{The whitened $n \times T$ data matrix $\bX_{\textrm{st}}$; initial $\bW_\bS^0$; tolerance $\epsilon$; and desired effective degrees of freedom.}
\KwResult{Estimates of the latent components, $\hatbS$, and their densities, $\{ \hat{f}_q\}$.}
\begin{enumerate}
 \item Let $(m)=0$ where $(m)$ denotes the number of update steps. Define $\bS^{(m)} = \bX_{\textrm{st}} {\bW_\bS^{(m)}}{^\top}$.
  \item Estimate $\{f_q^{(m+1)}\}$ in which the smoothness penalty is chosen to result in the specified effective df.
  \item Update $\bW_\bS^{(m+1)}$ given $f_1^{(m+1)},\dots,f_Q^{(m+1)}$ and $\bS^{(m)}$ with one step of the symmetric fixed-point algorithm (see Algorithm 2 in Web Supplement \ref{WS:fixedpoint}).
  \item Let $\bS^{(m+1)} = \bX_{\textrm{st}} {\bW_\bS^{(m+1)}}{^\top}$.
  \item If $PMSE(\bW_\bS^{(m+1)}{^\top},\bW_\bS^{(m)}{^\top}) < \epsilon$, stop, else increment $(m)$ and repeat (2)-(4).
\end{enumerate}
\end{algorithm}
\end{linenomath*}
%\vspace{-0.3in}
%\section{Simulations}\label{sec:Simulations}
%\vspace{-0.2in}
\section{\smaller Simulations: Distributional \& Noise-rank Assumptions}\label{sec:SimIID}
In this section, we simulate the LNGCA model [given by (\ref{eq:LNGCAmodel}) with $\bM_\bS \in \tR^{T \times Q}$] and the noisy ICA model [again given by (\ref{eq:LNGCAmodel}) with $\bM_\bS \in \tR^{T \times Q}$ but now with $\bM_\bN \bN \sim N(0,\sigma^2 \bI_T)$] under a variety of source distributions in which the components are iid as well as a scenario in which the signals are sparse images. We compare (i) deflationary FastICA with the `tanh' nonlinearity (D-FastICA), where the deflation option estimates components one-by-one such that the algorithm is considered a projection pursuit method \citep{hyvarinen2000independent}; (ii) two-class IFA with isotropic noise (IFA); (iii) PCA followed by Infomax (PCA+Infomax); (iv) PCA followed by ProDenICA (PCA+ProDenICA); (v) Logis-LCA; and (vi) Spline-LCA. We evaluate the robustness of these methods with respect to assumptions on the rank of the noise components, distribution of the latent components, and the signal-to-noise ratio (SNR). We define the SNR as the ratio of the total variance from the mixed non-Gaussian components to the total variance from the noise components. Formally, consider the non-zero eigenvalues $d_1,\dots,d_Q$ from the covariance matrix of $\bM_\bS \bS$. For LCA, let $d_{\epsilon_1},\dots,d_{\epsilon_{T-Q}}$ denote the eigenvalues from the EVD of the covariance matrix of $\bM_\bN \bN$. Then, $SNR = \frac{\sum_{q=1}^Q d_q}{\sum_{k=1}^{T-Q} d_{\epsilon_k}}$.
For the noisy ICA model, we have $T$ non-zero eigenvalues (equal to $\sigma^2$) in the denominator. %Let ${\bf m}_t'$ denote the $t$th row of $\bM$. For centered data, the sample analogue of \eqref{eq:SNR} can be calculated as
% \begin{equation*}
%  snr = \frac{\sum_{t=1}^T \sum_{i=1}^n ({\bf m}_t' \bs_i)^2}{\sum_{t=1}^T \sum_{i=1}^n \epsilon_{vt}^2}.
% \end{equation*}

 We fit D-FastICA using a modification of the fastICA R package \citep{Marchini:2010lr}. We fit PCA+Infomax using our own implementation of the Infomax algorithm. We fit PCA+ProDenICA using the ProDenICA function from the R package of that name \citep{hastieR2010}. Note that these methods can provide an estimate of $\bS$ but not the mixing matrix, which we estimated using \eqref{eq:hatbMs}. We fit the IFA model with two-component mixtures of normals using our own implementation, and the ICs were estimated by their conditional means (see equation (81) in \citealt{attias1999independent}). See Web Supplement \ref{supp:secNoisyICA}.
% \vspace{-0.1in}
% \subsection{Simulation Design}\label{sec:Sim1}
% \vspace{-0.1in}

Data were generated with $T=5$ and $Q=2$ according to a $2^2 \times 6$ full factorial design. The three factors were
\begin{enumerate}

\item[i)] {\bf The model}: the  levels were (a) the LNGCA model with rank-$(T-Q)$ noise and (b) the noisy ICA model with rank-$T$ noise. In both models the signal was $\bM_\bS \bS$ where $\bM_\bS$ is $T \times Q$ with $Q < T$.

\item[ii)] {\bf The signal to noise (SNR) ratio}: the levels were (a) high where  the ratio of the variance from the signal components to the variance from the noise components was 5:1 and (b) low where that ratio was 1:5.

\item[iii)] {\bf Signal distribution}:  the levels were (a) logistic, (b)  t, (c) Gumbel, (d) sub-Gaussian mixture of normals,  (e) super-Gaussian mixture of normals, (f) with values determined by a sparse image, as described below.  The two signal components were each iid and had the same distributions in cases (a)--(e) but differed in the sparse signal case.
\end{enumerate}

Since we generated $Q=2$ signal components for all simulations, there were $T-Q=3$ and $T=5$ noise components for the LNGCA model and noisy ICA model, respectively.   Observations in the noise components were iid isotropic normal except for the sparse image scenario, in which we used the R-package {neuRosim} \citep{welvaert2011a} to generate three-dimensional Gaussian random fields with full width at half maximum (FWHM)  equal to 6 for each noise component.

The signal components had scale parameter equal to $\sqrt{3}/\pi$ for the logistic, three degrees of freedom for the t, and scale parameter equal to $\sqrt{6}/\pi$ for the Gumbel. For the super-Gaussian mixture of normals, we simulated a two-class model with the first centered at $0$ with variance $4/9$ with probability 0.95 and the second centered at $5$ with unit variance (excess kurtosis $\approx 9$), which is motivated by a brain network with 5\% of voxels (volumetric pixels) activated. For the sub-Gaussian mixture of normals, we used the two-class model with the first centered at $-1.7$ with unit variance and probability 0.75 and the second centered at $1.7$ with unit variance and probability equal to 0.25, which is equivalent to distribution `l' from \cite{hastie2002independent} (excess kurtosis $\approx -0.3$). For the sparse image, we used neuRosim to generate two 10 $\times$ 10 $\times$ 10 images: in the first component, activation was represented by a sphere of radius two voxels centered at $(5,5,5)$ with voxel-value equal to one and exponential decay rate equal to 0.5; in the second, the feature was a cube centered at $(7,7,7)$ with width equal to two and exponential decay rate equal to one.

%In estimation, latent and noise components were vectorized in the same manner used in fMRI.
%The exponential decay rate used in \texttt{neuRosim} is equivalent to  a power-exponential semi-variogram with rate equal to $1/\sqrt{2}$ and power equal to two, but note that here, the dependence is fixed rather than random.

We conducted 112 simulations (chosen because we used a cluster with 56 processors) with $n=$1,000 observations in which $\bM_\bS$ and $\bM_\bN$ were randomly generated to have condition number between one and ten for each combination of factors. Since neither the set of orthogonal matrices (PCA+ICA methods) nor semi-orthogonal matrices (LCA methods) is convex, we approximated the $\margmax$ by initializing D-FastICA, PCA+Infomax, Logis-LCA, PCA+ProDenICA, and Spline-LCA from twenty random matrices and selecting the estimate associated with the largest objective function value. For Logis-LCA and Spline-LCA, ten of these twenty initializations were from random matrices in the principal subspace. Let $\hatbU_{1:Q}^\top$ denote the first $Q$ rows from $\hatbU^\top$ in the decomposition $\hatbSigma = \hatbU {\widehat{\bLambda}} \hatbU^\top$. Then $\bW_\bS^0 = \bO \hatbU_{1:Q}^\top$ for $\bO \in \cO_{Q \times Q}$ produces a semiorthogonal matrix in the principal subspace, which may help convergence for large SNRs. For IFA, one must specify initial values for the unmixing matrix, the variance of the isotropic noise, and the parameters of the Gaussian mixtures, and here we had four strategies to find the argmax including initialization from the true $\bW_\bS$ (Web Supplement \ref{supp:secNoisyICA}).
% \vspace{-0.1in}
% \subsection{Results}
% \vspace{-0.1in}

When the LNGCA model was true and there was a high SNR, all methods except IFA generally produced accurate estimates of $\bS$ for the logistic, t, Gumbel, super-Gaussian mixture of normals, and sparse images, but only Spline-LCA was accurate for the sub-Gaussian mixture of normals, and the performance of IFA was more variable than other methods for all distributions (Figure \ref{fig:MSE_simulations}). PCA+Infomax performed poorly for the sub-Gaussian mixtures because the logistic distribution generally fails for sub-Gaussian distributions (see \citealt{lee1999independent}).  Boxplots examining the accuracy of $\hatbM_\bS$ showed patterns similar to those found in Figure \ref{fig:MSE_simulations} and consequently are not presented.

\begin{figure}
 \caption{Boxplots of permutation-invariant root mean squared error ($PRMSE$) for estimated columns of $\bS$ where the rank of the noise was $T-Q$ (LNGCA Model) or $T$ (noisy ICA model) in high SNR (`HI') and low SNR (`LO') scenarios for various latent distributions. `DF' = D-FastICA; `IFA' = independent factor analysis; `PI' = PCA+Infomax; `LL' = Logis-LCA; `PP' = PCA+ProDenICA; `SL' = Spline-LCA.}\label{fig:MSE_simulations}
 \begin{minipage}[b]{0.5\linewidth}
  \includegraphics[width=\textwidth]{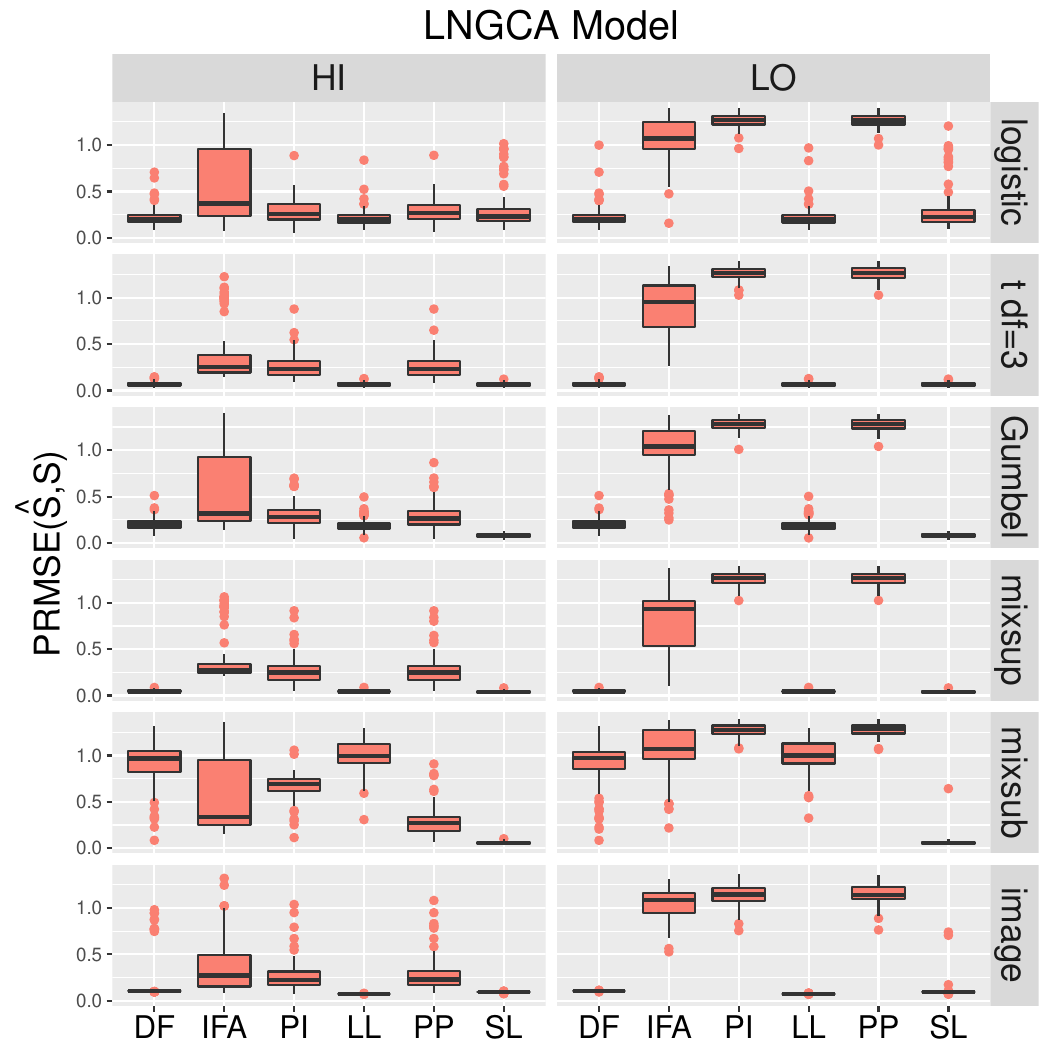}
 \end{minipage}
\begin{minipage}[b]{0.5\linewidth}
 \includegraphics[width=\textwidth]{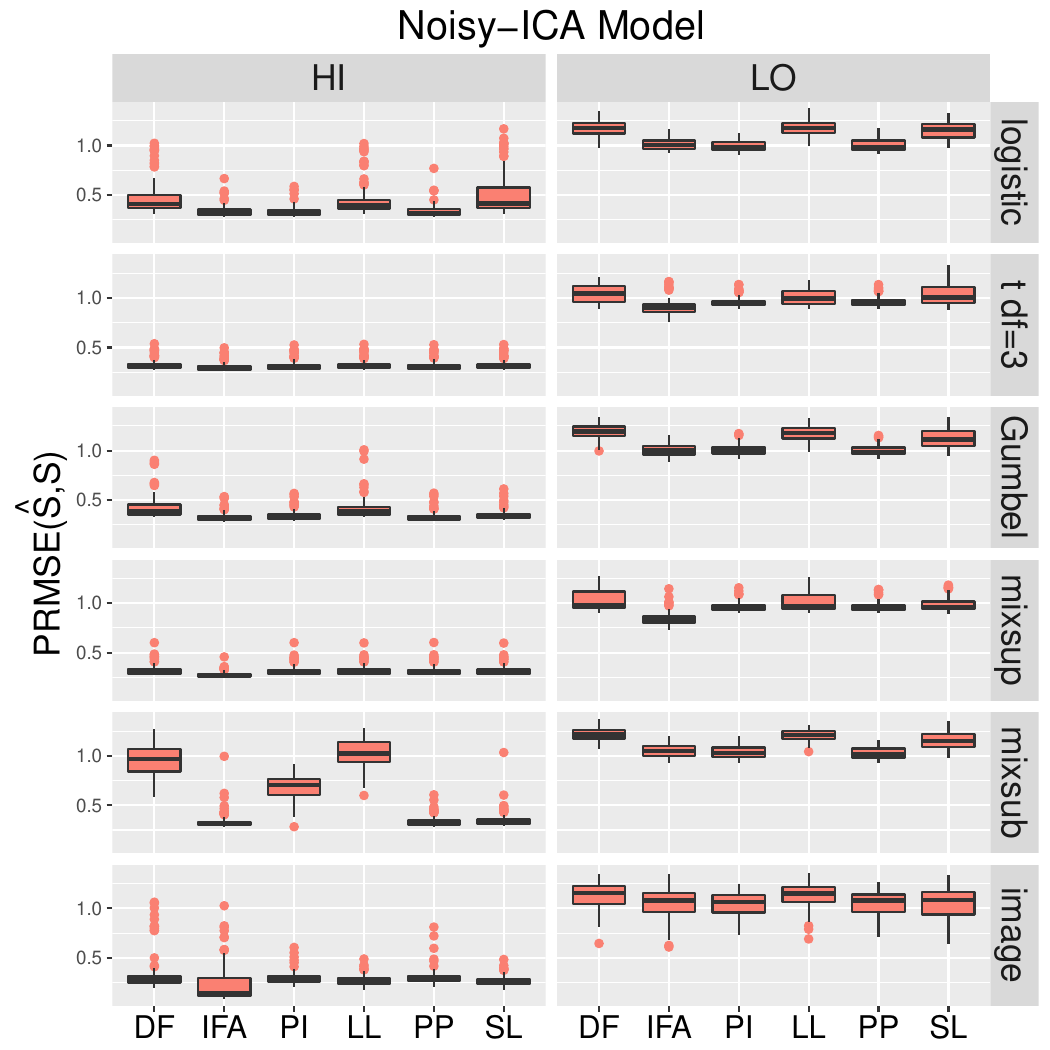}
\end{minipage}
\end{figure}

When the LNGCA model was true and there was a low SNR, Spline-LCA generally outperformed other methods, while IFA, PCA+Infomax, and PCA+ProDenICA failed to recover the LCs for all distributions, and D-FastICA and Logis-LCA recovered all distributions except for the sub-Gaussian mixture of normals. Thus for low SNR, PCA+ICA methods discarded the non-Gaussian signal. This was true even when the correct source density was modeled, as in PCA+Infomax and the logistic density simulation. Spline-LCA was the method most robust to distributional assumptions and was the only method that recovered the sub-Gaussian mixture. We numerically evaluated the condition in Assumption \ref{assumption:mismatch}(i) for Logis-LCA and all values were negative except for the sub-Gaussian mixture of normals; thus the results for $n=1000$ are in general agreement with Theorem 3. We also evaluated the mis-match criterion between densities (a)-(e) and Spline-LCA densities estimated from a sample from the true densities, and all values were negative. 

When the noisy ICA model was true and there was a high SNR, all methods generally produced reasonably accurate estimates for the logistic, t, Gumbel, super-Gaussian, and sparse image. IFA and Spline-LCA were the only methods that recovered ICs with sub-Gaussian distributions. When the noisy ICA model was true and there was a low SNR, all methods performed poorly, although IFA, PCA+Infomax, and PCA+ProDenICA outperformed LCA algorithms for some distributions. Note that in PCA+ICA methods, PCA decomposes the data into a subspace with the signal and some noise, and a subspace with noise only; see Web Supplement \ref{supp:secNoisyICA}. When the SNR is high, this is an effective strategy because the amount of error that corrupts the ICs is negligible. When there is a low SNR, the components estimated with ICA are highly contaminated with noise.

%For the logistic, t, Gumbel, sub-Gaussian mixture, and image, PCA+Infomax and PCA+ProDenICA were slightly more accurate than IFA, although IFA was more accurate for the super-Gaussian mixture of normals.

Overall, LCA methods were robust to the SNR for rank-$(T-Q)$ noise, and performed well in the high SNR scenario for rank-$T$ noise. Additionally, Spline-LCA was most robust to distributional assumptions. In contrast, IFA, PCA+Infomax, and PCA+ProDenICA performed poorly in the low SNR scenario for both the rank-$(T-Q)$ and rank-$T$ noise. %Thus, there is little evidence that IFA, PCA+Infomax, or PCA+ProDenICA should be used.

\section{\smaller Simulations: Spatio-temporal Signals and $Q^* \ne Q$}\label{sec:SpatioTemporalSims}
Next, we examine the ability of D-FastICA, PCA+Infomax, Logis-LCA, and Spline-LCA to recover simulated spatially structured signals (i.e., sources) whose loadings vary deterministically with time in the presence of spatially and temporally correlated noise. Each spatial source is similar to a brain ``network'' (or a resting-state network) in fMRI, in which a set of active locations share the same temporal behavior and each location corresponds to the index $i$. We also examine the effect of using $Q^* \neq  Q$ on source recovery. We did not include IFA in these simulations because it was difficult to estimate when $T$ was relatively large (e.g., $T=50$). Additionally, IFA, PCA+Infomax, and PCA+ProDenICA produced similar results for super-Gaussian distributions in the previous simulations.  %, and PCA+Infomax and PCA+ProDenICA were more accurate than IFA in the high SNR LNGCA scenario. %Hence, our previous simulations suggest there would be little insight gained from including IFA.
%\vspace{-0.1in}
%\subsection{Simulation Design}
%\vspace{-0.1in}
We simulated three sources mixed across fifty time units. The sources were  33$\times$33 images corresponding to $n=1089$. ``Active'' pixels were in the shape of a ``1'', ``2 2'', or ``3 3 3'' with values between 0.5 and 1 and ``inactive'' pixels were mean zero iid normal with variance equal to 0.0001 (see Figure \ref{fig:simFMRI123_LCA}).  Let ${\bf m}_q$ denote the $q$th column of $\bM_\bS$. To simulate the temporal activation patterns of brain networks, we used {neuRosim} \citep{welvaert2011a} to convolve the canonical hemodynamic response function (HRF) with a block-design with a pair of onsets at $\{1, 20.6\}$, $\{10.8, 40.2\}$, and $\{10.8,30.4\}$ for ${\bf m}_1$, ${\bf m}_2$, and ${\bf m}_3$, respectively, and duration equal to 5 time units.

In the LNGCA scenario, noise components were generated as forty-seven independent 33$\times$33 Gaussian random fields with FWHM=6. Each column of $\bM_\bN$ corresponded to an AR(1) process simulated for fifty time units with AR coefficient equal to 0.47 and unit variance, where the AR coefficient was chosen based on a preliminary analysis of the fMRI data analyzed in Section \ref{sec:t-fMRI}. Additionally, noise components were scaled such that the SNR was 0.4, which approximately equals the SNR estimated in Section \ref{sec:t-fMRI}. In the noisy ICA scenario, a 33$\times$33 Gaussian random field with FWHM=6 was simulated for $t=1$. Then noise components were defined recursively for $t=2,\dots,50$ to be equal to 0.47 times the noise at time $t-1$ plus a realization from an independent Gaussian random field with FWHM=6.

We conducted 111 simulations with $Q^* = 2,3$ or $4$ (with fixed $Q=3$) and initialized all algorithms from twenty random mixing matrices for each simulation and each $Q^*$. For Logis-LCA and Spline-LCA, ten of the twenty initializations were from random matrices in the principal subspace, as in Section \ref{sec:SimIID}.
% \vspace{-0.1in}
% \subsection{Results}
% \vspace{-0.1in}
%For rank-$(T-Q)$ noise, Logis-LCA and Spline-LCA accurately estimated components of $\bS$ for all $Q^*$, while D-FastICA was highly variable for $Q^*=2$ and $Q^*=4$, and PCA+Infomax and PCA+ProDenICA performed poorly for all $Q^*$. For rank-$T$ noise, Logis-LCA was the most accurate, whereas Spline-LCA performed worse than PCA+Infomax and PCA+ProDenICA for $Q^* = 3$ and $Q^* = 4$. For rank-$T$ noise, D-FastICA performed well for $Q^*=2$ but poorly for $Q^*=3$ and was highly variable for $Q^*=4$.
\begin{figure}[H]
 \caption{Spatial source recovery from the LNGCA scenario with $Q=3$ for $Q^*=2$, 3, or 4. Images depict LCs and time series depict the loadings ($\widehat{\mathbf{m}}_1,\dots,\widehat{\mathbf{m}}_{Q^*}$) corresponding to the median $PRMSE(\hatbS,\bS)$. In the last column, ``Truth'' corresponds to an arbitrary noise component whereas the algorithms attempted to estimate a fourth LC.}\label{fig:simFMRI123_LCA}

  $\hspace{0.7in} \bf Q^* = 2 \hspace{1.5in} \bf Q^* = 3 \hspace{1.9in} \bf Q^* = 4$

  \includegraphics[width=0.975\textwidth]{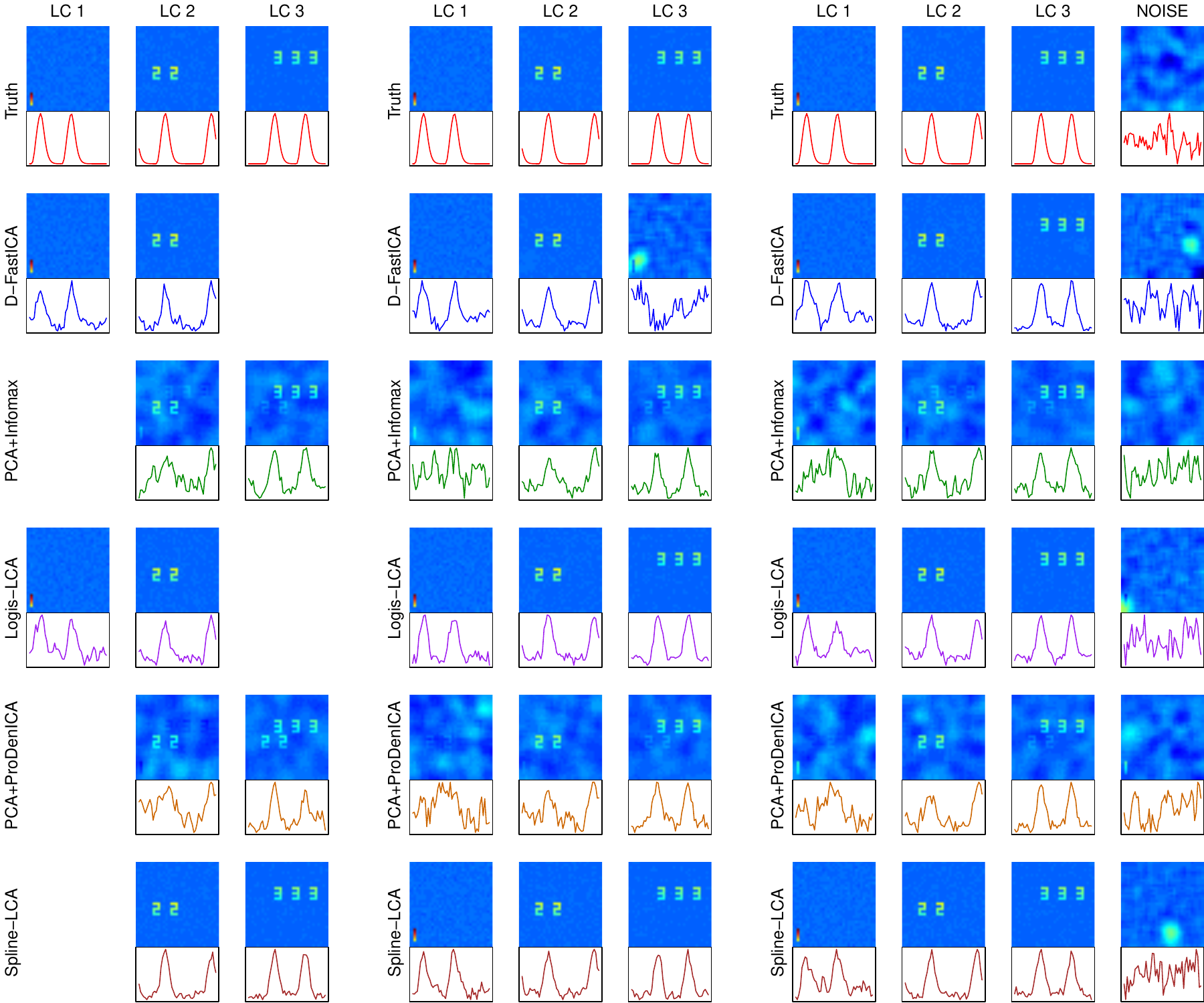}
\end{figure}

By inspecting the images and loadings associated with the median $PRMSE(\hat{\bS},\bS)$ for each method in the LNGCA scenario, we see that D-FastICA recovers a spurious component when $Q^*=3$; PCA+Infomax and PCA+ProDenICA generally fail to unmix features; and Logis-LCA and Spline-LCA are highly accurate (Figure \ref{fig:simFMRI123_LCA}). Boxplots for D-FastICA indicate higher $PRMSE$ than Logis-LCA or Spline-LCA for $Q^*=3$ and $Q^*=4$ (Figure \ref{fig:MSE_simFMRI}), and the third component was typically not recovered for $Q^*=3$ (Figure \ref{fig:simFMRI123_LCA}). This suggests a deflationary approach to estimating LNGCA may be inaccurate. In contrast, Logis-LCA and Spline-LCA recovered the components in all simulations (Figure \ref{fig:MSE_simFMRI}). It is notable that estimates from PCA+Infomax, PCA+ProDenICA, and D-FastICA were sensitive to the choice of $Q^*$, whereas Logis-LCA and Spline-LCA were robust (Figures \ref{fig:simFMRI123_LCA}, \ref{fig:MSE_simFMRI}).
%Note that estimates from PCA+Infomax and PCA+ProDenICA were sensitive to the choice of $Q^*$, as when $Q^*<Q$, an estimated latent component resembled a union of components two and three. In PCA+ProDenICA, the loadings for the estimated component were highly correlated with component three ($r=0.75$), which mistakenly suggests components two and three are functionally connected. For $Q^*=3$, the features in the estimated component one are faintly visible in PCA+Infomax whereas component one was not recovered by PCA+ProDenICA. In contrast, Logis-LCA and Spline-LCA clearly separated components for all $Q^*$, such that when $Q^*<3$, the recovered components were accurate estimates of a subset of the true ($Q=3$) components (Figure 2; Figure S.1).

For the noisy-ICA scenario, the features recovered by Logis-LCA most closely resembled the truth (Figure \ref{fig:simFMRI123_NoisyICA}) and Logis-LCA generally outperformed other methods (Figure \ref{fig:MSE_simFMRI}). Features from component two were again faintly visible in component three for $Q^*=2$ in both PCA+Infomax and PCA+ProDenICA, again indicating inadequate unmixing of the sources. As seen in the LNGCA scenario, D-FastICA recovered a spurious component for $Q^*=3$, but accurately estimated component three in the majority of simulations when $Q^*=4$. Spline-LCA typically failed to recover component one for $Q^*=3$, although it was quite accurate for components two and three. Spatial correlations in the noise can result in spurious disk-like features, which were estimated in D-FastICA for both scenarios and by Spline-LCA in the noisy-ICA scenario. For the simulation associated with the median error, an accurate estimate of component one was associated with a local maxima in Spline-LCA, but the spurious component had a higher likelihood. The true component was recovered in some simulations (Figure \ref{fig:MSE_simFMRI}).

\section{ \smaller Data Visualization and Dimension Reduction}\label{sec:LeafExample}
We used Logis-LCA and Spline-LCA for data visualization and dimension reduction in multivariate data comprising measurements from independent leaf samples \citep{silva2013}. Fourteen variables were generated from eight to sixteen images of leaves from each of thirty  species (Figure~\ref{fig:PlantSpecies}). Many of the covariates are highly correlated (Figure~\ref{fig:LeafCorr}). We plotted the first two PCs, ICs from PCA+Infomax and PCA+ProDenICA, and LCs from Logis-LCA and Spline-LCA. Two-dimensional PCA does not reveal clear features (Figure \ref{fig:LeafCorrPaper}). Since we are examining two dimensions, the effect of ICA is apparent as a rotation of the X- and Y-axes. Rotating the axes does not reveal any additional insight (Figure \ref{fig:LeafCorrPaper}, Figure~\ref{fig:leafLogis}). In contrast, Spline-LCA clearly reveals three clusters, where the green dots correspond to two plant species that have very thin leaves (species 31 and 34 in Figure \ref{fig:PlantSpecies}), the blue category corresponds to a species with leaves that are thinner than most species but less than those comprising the green dots (species 8), and the red category corresponds to all other species. Logis-LCA also reveals structure (Figure~\ref{fig:leafLogis}), although the separation is less than in Spline-LCA. A referee remarked that the NGCA model is not true here, in which the data are a mixture of thirty fourteen-variate distributions corresponding to the thirty species. We agree, but the goal here is to identify useful features. We find the model useful to that end.
 \begin{figure}
   \caption{Data visualization and dimension reduction for the leaf dataset. The original dataset comprises 14 variables, many of which are highly correlated. The green dots correspond to \emph{Podocarpus sp}.\ and \emph{Pseudosasa japonica}; the blue dots to \emph{Neurium oleander}; the red dots to all other species.}\label{fig:LeafCorrPaper}
  \begin{minipage}[b]{0.32\linewidth}
  \includegraphics[width=\textwidth]{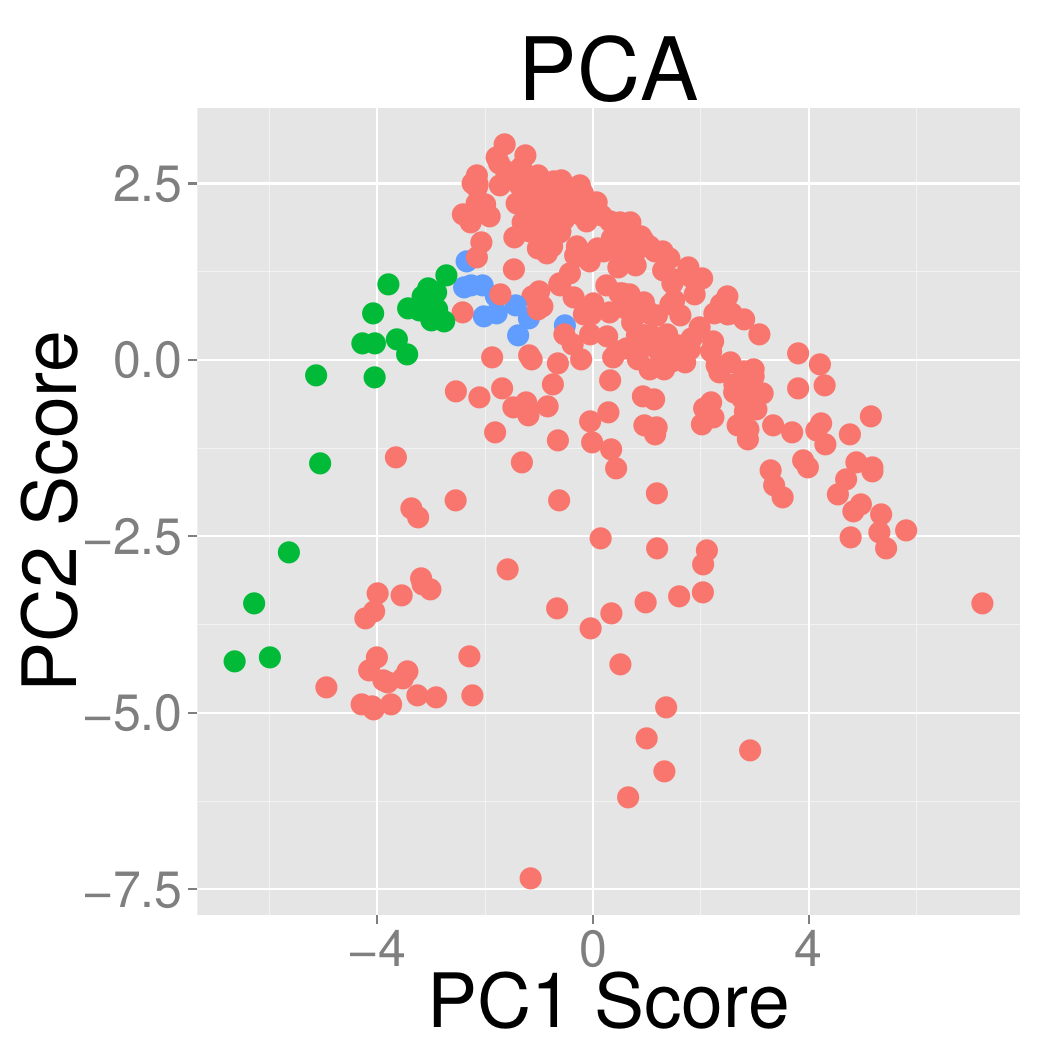}
  \end{minipage}
   \begin{minipage}[b]{0.32\linewidth}
    \includegraphics[width=\textwidth]{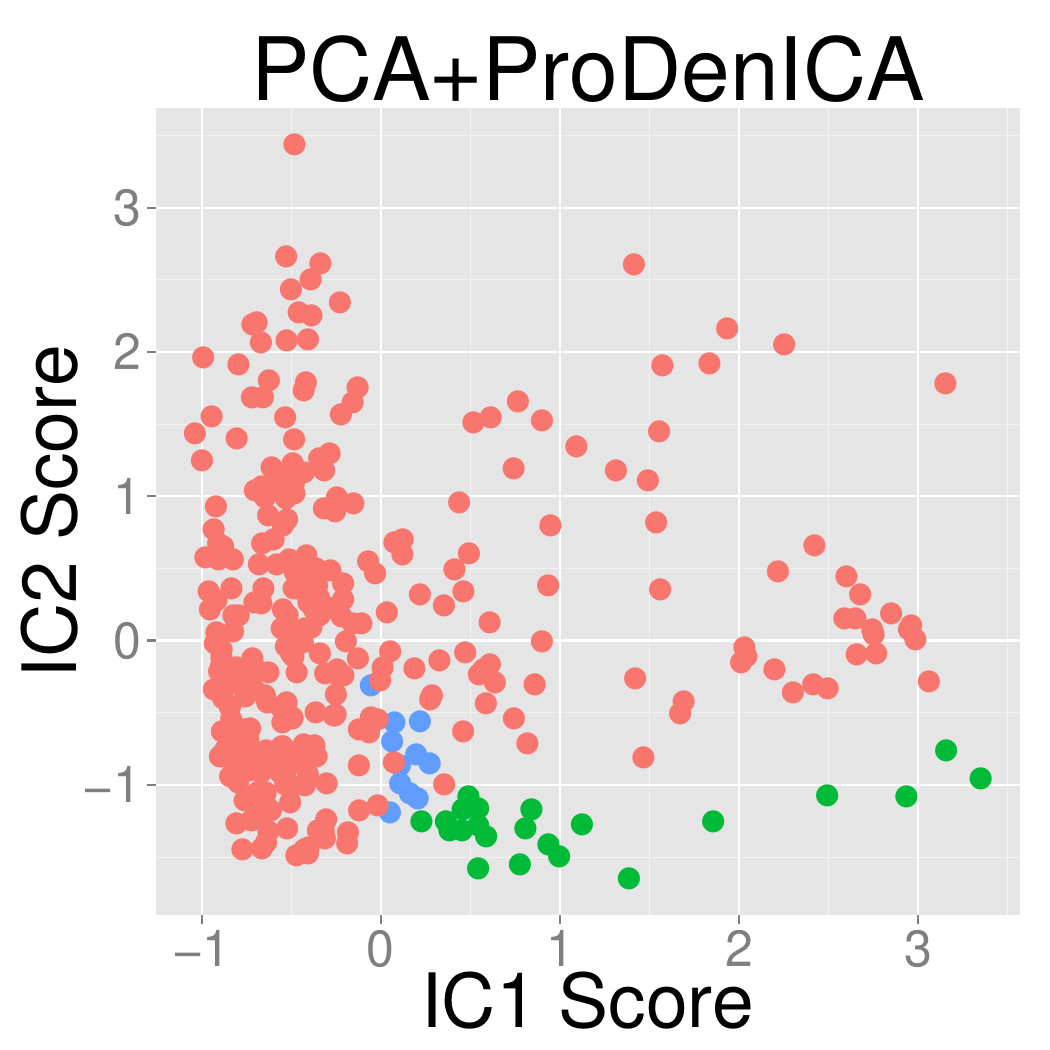}
  \end{minipage}
 \begin{minipage}[b]{0.32\linewidth}
  \includegraphics[width=\textwidth]{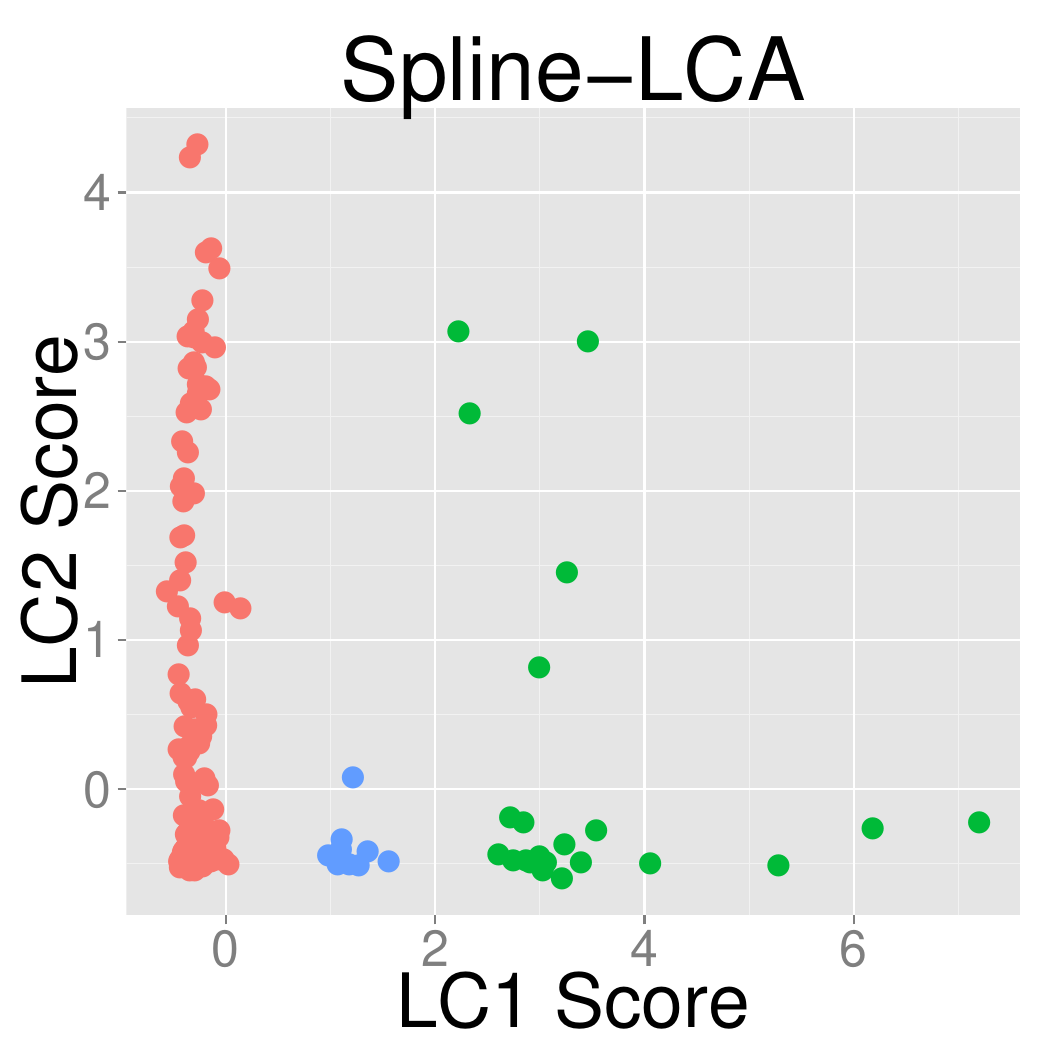}
 \end{minipage}
 \end{figure}

PCA+ICA methods were sensitive to the number of components estimated whereas the highest ranked components were very similar for different $Q^*$ in the LCA methods. In PCA+Infomax and PCA+ProDenICA, the first two (matched) ICs for $Q^*=5$ differed from the ICs estimated using two components, demonstrating the sensitivity of PCA+ICA methods to the number of principal components (Figures \ref{fig:leafLogis} and \ref{fig:leafSpline}). In contrast, the two highest-ranked LCs extracted from Logis-LCA and Spline-LCA when five components were estimated were very similar to the LCs estimated using two components.

\section{Application to fMRI}\label{sec:t-fMRI}

We applied Spline-LCA to eleven subjects from the Social Cognition / Theory of Mind experiment of the WU-Minn Human Connectome Project (HCP); additional information is in Web Supplement \ref{supp:sec:fMRI}. Single-subject ICA is an important technique for identifying artifacts in fMRI due to physiology (heart rate, breathing), subject-specific motion, and/or scanner instabilities, and accounting for these artifacts can decrease false positives and increase sensitivity \citep{pruim2015ica}. We used the minimally preprocessed data from the \emph{fMRIVolume} pipeline \citep{glasser2013minimal}. The preprocessing pipeline includes rigid-body motion correction of all volumes to a subject's reference image. Note that even if perfect alignment were possible, motion artifacts may still be present due to spin history effects and/or spatial variation in the coil sensitivities \citep{friston1996movement}. The \emph{fMRIVolume} pipeline does not include any spatial smoothing. Three-dimensional volume data were vectorized and non-brain tissue excluded using the mask provided from the HCP. This resulted in a 230,459 $\times$ 272 data matrix. Each voxel was treated as a replicate with $i=1,\dots,n$ for $n=\;$230,459, which is analogous to `spatial' ICA of fMRI \citep{calhoun2006unmixing}. We mean centered and variance normalized each voxel's time course prior to conducting LCA, as suggested for ICA \citep{beckmann2004probabilistic}.

We used the ICA software MELODIC (FSL) to determine the number of components that would be used in an analogous ICA of this dataset, which chose thirty components for subject 103414. Thirty components were then estimated for all other subjects. We initiated the algorithm from fifty-six matrices as described in Web Supplement \ref{supp:sec:fMRI}. Initiating the algorithm from fifty-six matrices resulted in multiple initializations converging to the same estimate of the argmax (Figure \ref{fig:InitValues_MDS}). We also completed an analogous PCA+ProDenICA with thirty components using the R package ProDenICA \citep{hastieR2010}.

In all subjects, a component highly correlated with the task was found in both Spline-LCA and PCA+ProDenICA, but a number of other components were only detected in Spline-LCA. We discuss the biological interpretation of the task-related network in the Web Supplement \ref{supp:sec:fMRI}, and here focus on the application to artifact detection. Overall, a median of eight components were found in Spline-LCA but not PCA+ProDenICA, as defined by the matched component having a correlation less than 0.5.  % (minimum number: 4; maximum number: 9). 
In one example from subject 103414, LC 25 exhibited activation at the edges of the brain, which is typical of motion artifacts \citep{salimi2014automatic}. This artifact was not evident in the matched component from PCA+ProDenICA (Figure \ref{fig:LC25}). Additionally, component two was not correlated with any of the components in PCA+ProDenICA. It exhibited activation in the brainstem and near the edges of the brain, and may correspond to other sources of motion and noise (Figure \ref{fig:HCP_networks}; Web Supplement \ref{supp:sec:fMRI}). There were also artifacts that exhibited alternating patterns of positive and negative activation (Figure~\ref{fig:LC14}), which may be due to scanner acquisition and/or air-tissue boundaries (e.g., Figure 6 in \citealt{salimi2014automatic}), and these components were not found in PCA+ProDenICA. %Similar artifacts were also found by Spline-LCA but not PCA+ProDenICA in subject 103414 (LC 7, LC 15).
Our results suggests that LCA may improve artifact detection.

%Component 25 is a clear example from subject 103414 of a motion artifact identified by Spline-LCA but not PCA+ProDenICA (Figure \ref{fig:LC25}). Here, the correlation between Spline-LCA and the matched component in PCA+ProDenICA was 0.38. Voxels near the edge of the cortex had high positive values on one side of the brain and negative values on the opposite side, which is typical of motion artifacts.

\begin{figure}
\centering
\caption{Motion artifact (component 25) identified using Spline-LCA (top) and the matched component from PCA+ProDenICA (bottom; correlation = 0.38) in subject 103414. Note the component exhibited activation near the edge of the brain in the LC but not the IC. Thresholded at $|s_{v,25}|>2$; yellow indicates $s_{v,25}>2$ and blue indicates $s_{v,25}<-2$.}\label{fig:LC25}
 \includegraphics[width=0.7\textwidth]{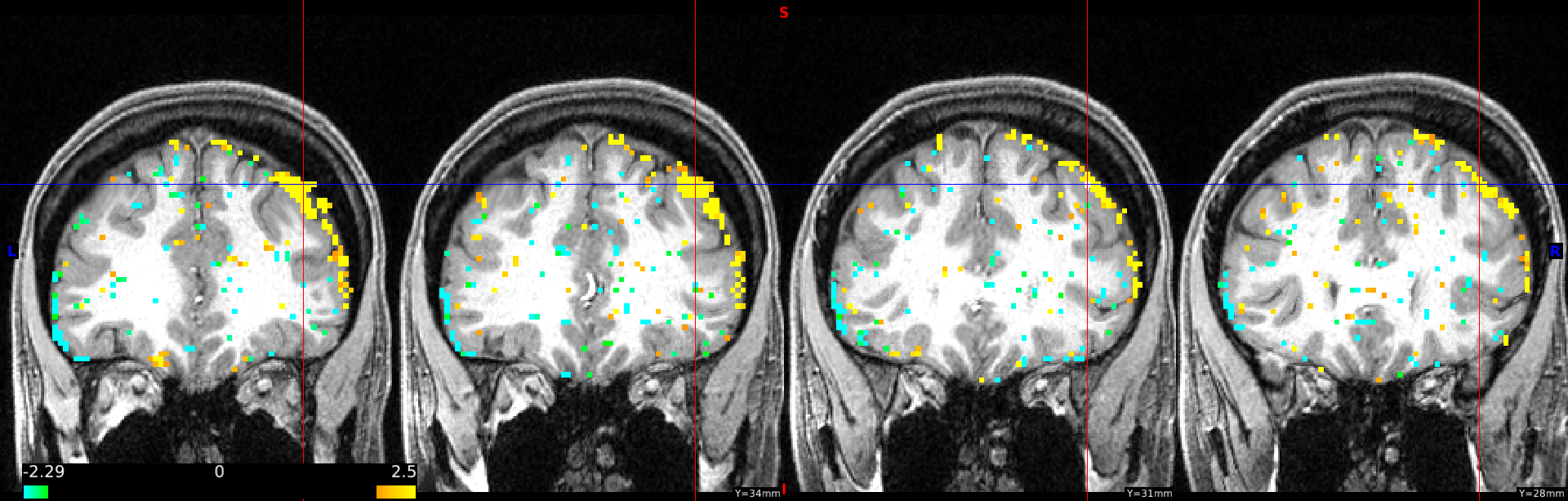}

 \includegraphics[width=0.7\textwidth]{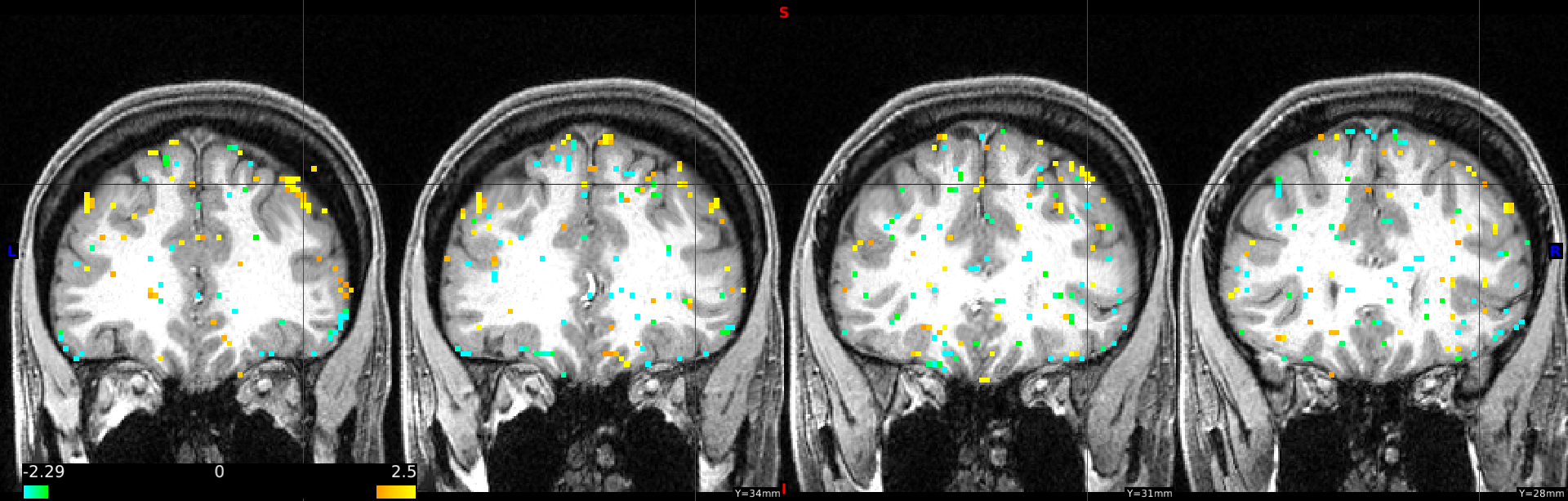}
\end{figure}

\section{Discussion}\label{sec:Discussion}
We propose a new model, LNGCA, and estimation framework, LCA, for non-Gaussian latent components in the presence of Gaussian noise that have many applications including dimension reduction, signal processing, and artifact detection. We presented two applications: data visualization and dimension reduction, and identifying brain networks and artifacts from neuroimagery.  Our first simulation study indicates that our methods perform well when the LNGCA model is true, even for low SNR, and our methods provide a reasonable approximation to noisy ICA when the SNR is high. Additionally, we found that the popular approach to approximating the noisy ICA model, PCA+ICA, does not approximate the LNGCA model under low SNR, and performs similarly to LCA for the noisy ICA model. In the second simulation study, we examined performance when data contained spatiotemporal dependence and a  moderately low SNR. Logis-LCA and Spline-LCA outperformed competing methods for the LNGCA model, and Logis-LCA outperformed all other methods for the noisy ICA model. These results suggest that LCA can be used to reveal structure for a large class of non-Gaussian observations. In the leaf example with correlated multivariate data, Spline-LCA revealed biologically meaningful clusters not apparent from PCA+ProDenICA. In our fMRI application, we simultaneously achieved dimension reduction and latent variable extraction for large image data ($T=272$ and $n=$230,459) and identified artifacts not extracted by PCA+ICA. %Using the LCA criteria in which LCs are ordered by their estimated marginal likelihoods, the component containing the most information coincided with the component most highly correlated with the task for all subjects.

% %, possibly because temporal dependencies can be captured in the rows of the mixing matrix.
% The presence of local maxima in LCA can increase computational expenses, and more initializations are required for larger values of $T$. Since the set of orthogonal matrices is non-convex, local optima are also a problem in PCA+ICA (e.g., \citealt{risk2014evaluation}). %Our approach is to initialize the LCA algorithms from random semi-orthogonal matrices generated from the eigenvectors of matrices with standard normal entries.
% For fMRI data, fifty initializations appeared to be adequate when estimating thirty components with nearly three hundred time points (Figure S.7). In general, we found that Logis-LCA was less sensitive to initialization than Spline-LCA (results not shown). %It appears that the additional flexibility of Spline-LCA comes at the expense of increased detection of local maxima.
% However, we favor Spline-LCA because it can more accurately model source densities. %However, sub-Gaussian components appear to be uncommon in fMRI (sparse images are super-Gaussian). Future research should examine whether Spline-LCA offers advantages over Logis-LCA in fMRI. Additionally, developing algorithms to more efficiently address local optima is an avenue for future research.

LCA offers a computationally tractable alternative to one of the most common applications of ICA to fMRI: artifact detection. Currently, PCA+ICA is used as a pre-processing step to reveal biologically implausible loadings and/or loadings resembling physiological artifacts that can be used to de-noise data for subsequent analyses \citep{beckmann2012modelling}. In LCA, these artifacts appear as LCs since they have non-Gaussian distributions. Our improved detection of artifacts (Figure \ref{fig:LC25}, Figures \ref{fig:HCP_networks} and \ref{fig:LC14}) suggests LCA could be used for more powerful denoising methods over traditional PCA+ICA. %Artifacts may increase and/or become more problematic when using state-of-the-art data with high-resolution, as smaller voxels are associated with smaller signals, suggesting artifact removal is increasingly important \citep{griffanti2014ica}. %The HCP data represent the highest resolution and fastest acquisition times currently available in fMRI, and thus LCA offers a promising alternative to ICA for artifact detection.

An important advantage of LCA over existing frameworks is its robustness to misspecification of the number of estimated components, and future research should examine methods to select $Q^*$. In contrast to LCA,  noisy ICA is sensitive to the choice of $Q^*$ (Section \ref{sec:SpatioTemporalSims}, see also \citealt{allassonniere2012stochastic}). \cite{beckmann2004probabilistic} explored the use of probabilistic PCA to estimate the number of brain networks prior to ICA in order to avoid model over-fitting, which addresses the concern that over-fitting may separate a single brain network into multiple brain networks. However, our simulations suggest that using too few components leads to inappropriately aggregated sources in PCA+ICA methods (Figures \ref{fig:simFMRI123_LCA} and \ref{fig:simFMRI123_NoisyICA}). In contrast, the components recovered for $Q^* \ne Q$ in Logis-LCA across model scenarios and Spline-LCA for the LNGCA scenario accurately represent the spatial features. Moreover, in the leaf data example, the first two components were nearly identical for $Q^*=2$ and $Q^*=5$ for LCA but differed for PCA+ICA (Figures \ref{fig:leafLogis} and \ref{fig:leafSpline}). To determine $Q^*$ in LNGCA, \cite{virta2016projection} suggest the sequential use of the Jarque-Bera test of normality. \cite{nordhausen2016asymptotic} develop asymptotic and bootstrap tests of dimensionality using first-order blind identification (FOBI). The use of these criteria in fMRI and other applications is a direction for future research.
\vspace{-0.2in}
\ifblind
\section{Acknowledgments}
  hidden
\else
\section{Acknowledgments}
We thank Dr.\ Nathan Spreng, Department of Human Development, Cornell University, for scientific guidance and assistance with the HCP data. DSM was supported by a Xerox PARC Faculty Research Award, NSF grant DMS-1455172, and Cornell University Atkinson's Center for a Sustainable Future Award AVF-2017. BBR was partially supported by the NSF grant DMS-1127914 to the Statistical and Applied Mathematical Sciences Institute. Data were provided (in part) by the Human Connectome Project, WU-Minn Consortium (Principal Investigators: David Van Essen and Kamil Ugurbil; 1U54MH091657). %funded by the 16 NIH Institutes and Centers that support the NIH Blueprint for Neuroscience Research; and by the McDonnell Center for Systems Neuroscience at Washington University.
\fi
\singlespace
\bibliographystyle{apalike}
\vspace{-0.2in}
%\bibliography{../../MyBibtex/MasterBibliography}
\bibliography{LCApaper_v2.bbl}

%  \documentclass[12pt]{article}
%\usepackage{amsmath,amssymb,amsthm,amsfonts,natbib,bm,setspace,graphics,graphicx,url,caption,multicol,l
%ongtable,lscape,verbatim,multicol,color,float}
%\usepackage[english]{babel}
%\usepackage[ruled,lined,boxed]{algorithm2e}
%\usepackage[letterpaper, margin=1in]{geometry}
%\usepackage[pagewise,mathlines]{lineno}
%\usepackage[toc]{appendix}
%%\numberwithin{equation}{section}
%%\numberwithin{figure}{section}
%
\setcounter{figure}{0}
\renewcommand{\thefigure}{S.\arabic{figure}}

\setcounter{equation}{0}
\renewcommand{\theequation}{S.\arabic{equation}}

\setcounter{table}{0}
\renewcommand{\thetable}{S.\arabic{table}}

\setcounter{page}{1}

\setcounter{thm}{0}

\setcounter{lemma}{0}

\setcounter{cor}{0}

%\newif\ifblind
%\newif\ifunblind
%%\blindtrue
%\unblindtrue

\begin{center}{\bf\Large Supplement to ``Linear Non-Gaussian Component Analysis via Maximum Likelihood''}

\bigskip

\ifblind
{author here}
\else
{Benjamin B. Risk, David S. Matteson, David Ruppert}
\fi
\end{center}

\maketitle

\doublespacing

%\linenumbers
\appendix

\section{Proofs}\label{WS:proofs}

\subsection{Proofs for Section \ref{sec:LNGCA}}

 We assume all random variables are mean zero. In \cite{kagan1973characterization}, a random variable $\bX \in \tR^T$ is said to have a \emph{linear structure} if it can be represented as $\bX = \bB \bY$ where the elements of $\bY$ are mutually independent random variables and no two columns of $\bB$ are proportional. We say a linear-structure random vector $\bX$ has \emph{essentially unique structure} if for any two representations $\bX = \bB \bY$ and $\bX = {\bf C} \bZ$, we have $\bB$ equals $\bC$ up to scaling and permutation of the columns, which we denote as $\bB \cong \bC$. A random variable $\bX$ is non-unique if there exist representations $\bX = \bB \bY = \bC \bZ$ but $\bB \ncong \bC$. Let $\eqd$ denote equal in distribution. First consider the theorem on uniqueness of decomposition.
\vspace{0.1in}

\noindent {\bf Theorem 10.3.9 from \cite{kagan1973characterization}.}
\emph{ Let $\bX = \bA \bY$ be a structural representation of $\bX$ and let the columns of $\bA$ be linearly independent. Then $\bX$ can be expressed as $\bX = \bX_1 + \bX_2$, where $\bX_1$ and $\bX_2$ are independent, $\bX_1$ has essentially unique structure, and $\bX_2$ is multivariate normal with a non-unique structure. Moreover, this decomposition is unique in the sense that if $\bX = \bZ_1 + \bZ_2$ is another decomposition, where $\bZ_1$ has essentially unique structure, $\bZ_2$ is multivariate normal, and $\bZ_1$ is independent of $\bZ_2$, then $\bZ_1 \eqd \bX_1$ and $\bZ_2 \eqd \bX_2$ up to scaling and permutations.}

For a proof see \cite{kagan1973characterization}.
\vspace{0.1in}

Before proving Theorem 1, we consider the following lemma. %This lemma extends Theorem 10.3.9 to show that if two random variables are equal in distribution and have essentially unique structure, then their decompositions are equivalent, in the sense described below.
\begin{lemma}
 Suppose $\bZ$ and $\bX$ each have essentially unique structure and $\bZ \eqd \bX$. Consider their structural representations: $\bZ = \bM_\bS \bS$ and $\bX = \bM_\bS^* \bS^*$ where $\bM_\bS \in \tR^{T \times Q}$ and $\bM_\bS^* \in \tR^{T \times Q}$ for $Q \le T$, and ${\rm rank}(\bM_\bS) = {\rm rank}(\bM_\bS^*) = Q$. Then $\bM_\bS \cong \bM_\bS^*$ and $\bS \eqd \bS^*$ up to scaling and permutations.
\end{lemma}
\begin{proof}
We have $\bM_\bS \bS \eqd \bM_\bS^* \bS^*$. Then,
\begin{linenomath*}
\begin{align*}
({\bM_\bS}^\top \bM_\bS)^{-1} {\bM_\bS}^\top\bM_\bS \bS =  ({\bM_\bS}^\top \bM_\bS)^{-1} {\bM_\bS}^\top \bM_\bS^* \bS^*.
\end{align*}
\end{linenomath*}
Letting $\bB = ({\bM_\bS}^\top \bM_\bS)^{-1} {\bM_\bS}^\top\bM_\bS^*$, we have $\bS \eqd \bB \bS^*$. Note by assumption $\bS \in \tR^Q$ and $\bS^* \in \tR^Q$. Now $\bS$ has non-Gaussian independent components and thus has essentially unique structure for the given number of components $Q$ (Theorem 10.3.5 in \citealt{kagan1973characterization}); in particular, $\bS = \bI \bS$. We can define a random variable $\bR = \bB^{-1} \bS$, and note that $\bR \eqd \bS^*$, and $\bS^*$ has independent components, which implies $\bR$ has independent components, which implies $\bB \bR$ is a structural representation of $\bS$. Since $\bS$ has essentially unique structure, $\bB \cong \bI$. It follows that $\bS^* \eqd \bS$ up to scaling and permutations.

Now consider the scaling and permutation such that $\bS^* \eqd \bS$. Then we have $\bB = \bI$, so $({\bM_\bS}^\top \bM_\bS)^{-1} {\bM_\bS}^\top \bM_\bS^* = \bI$. Now since $({\bM_\bS}^\top \bM_\bS)^{-1} {\bM_\bS}^\top$ is full row rank, it has a unique right inverse equal to the Moore-Penrose pseudoinverse, which is equal to $\bM_\bS$, which implies $\bM_\bS = \bM_\bS^*$. For $\bB \cong \bI$, it follows that $\bM_\bS^* \cong \bM_\bS$.
%Towards a contradiction, suppose $\bM_\bS^* \ncong \bM_\bS$. But then $\bM_\bS^* \bS^* \overset{d}{\ne} \bM_\bS \bS$.
\end{proof}

We now prove Theorem 1.
\begin{thm}
Suppose $\bX$ follows the model in \eqref{eq:LNGCAmodel} with Assumptions 1-3. Then for any other representation $\bX = \bM_\bS^* \bS^* + \bE^*$ where $\bS^* \in \tR^Q$ are independent non-Gaussian components and $\bE^*$ is multivariate normal, we have: $\bM_\bS^* \cong \bM_\bS$; $\bS^* \eqd \bS$ up to scaling and permutations; $\bM_\bS \bS \eqd \bM_\bS^* \bS^*$; and $\bE^* \eqd \bM_\bN \bN$.
\end{thm}

\begin{proof}
Since $\bX$ has a unique decomposition in the sense of Theorem 10.3.9, we have $\bM_\bS \bS \eqd \bM_\bS^* \bS^*$ and $\bM_\bN \bN \eqd \bE^*$. Moreover, $\bM_\bS \bS$ and $\bM_\bS^* \bS^*$ have essentially unique structure (Theorem 10.3.5 in \citealt{kagan1973characterization}). Applying Lemma 1, we obtain the desired result.
\end{proof}

\begin{cor}
 Suppose the linear structure model in \eqref{eq:LNGCAmodel} of the main manuscript with density defined in \eqref{eq:JointDensity} and suppose that Assumptions 1-3 hold. Then $\{f_1,\bw_1\},\dots,\{f_Q,\bw_Q\}$ are identifiable up to sign and ordering. Note the rows $\bw_{Q+k}$ for $k=1,\dots,T-Q$ are not identifiable.
\end{cor}

\begin{proof}
For identifiability, we need to show that if there exist densities $g_1,\dots,g_T$ and a matrix $\bC$ such that
\begin{linenomath*}
\begin{align}\label{eq:s1}
 |\det(\bL)|\prod_{q=1}^Q f_q\left({\bw_q^\top \bL \bx}\right)\prod_{k=1}^{T-Q} \phi(\bw_{Q+k}^\top\bL \bx) = |\det(\bC)|\prod_{\ell=1}^T g_\ell (\bc_\ell^\top \bx)
\end{align}
\end{linenomath*}
then $Q$ of the marginal densities $g_1,\dots,g_T$ are equivalent up to sign to $f_1,\dots,f_Q$, where densities $g(x)$ and $f(x)$ are equivalent up to sign if they are equal or if $g(x) = f(-x)$ for all $x$ on $\tR$, and that each of the corresponding $Q$ rows of $\bC$ equal $\bw_1^\top \bL,\dots,\bw_Q^\top \bL$. Using a change of variable $\bZ = \bL \bX$, we consider the model $\bZ = \bA_\bS \bS + \bA_\bN \bN$, such that $[\bw_1^\top;\dots;\bw_Q^\top] = \bA_\bS^\top$ (where $[\bw_1^\top;\dots;\bw_Q^\top]$ indicates stacked row vectors) and $[\bw_{Q+1}^\top;\dots;\bw_{T}^\top] = \bA_\bN^\top$. Then \eqref{eq:s1} is equivalent to
\begin{linenomath*}
\begin{align*}
 \prod_{q=1}^Q f_q\left({\bw_q^\top \bz}\right)\prod_{k=1}^{T-Q} \phi(\bw_{Q+k}^\top\bz) = |\det(\bC)||\det(\bL)|^{-1}\prod_{\ell=1}^T g_\ell (\bc_\ell^\top \bL^{-1} \bz).
\end{align*}
\end{linenomath*}
We define $\bR = \bC \bL^{-1}$ such that we have
\begin{linenomath*}
\begin{align}\label{eq:2}
 \prod_{q=1}^Q f_q\left({\bw_q^\top \bz}\right)\prod_{k=1}^{T-Q} \phi(\bw_{Q+k}^\top\bz) = |\det(\bR)|\prod_{\ell=1}^T g_\ell (\br_\ell^\top \bz).
\end{align}
\end{linenomath*}
We have demonstrated identifiability up to signed permutations if we can show that $Q$ of the marginal densities $g_1,\dots,g_T$ are equivalent to $f_1,\dots,f_Q$; that each of the corresponding $Q$ rows of $\bR$ equal $\pm \bw_1,\dots,\pm \bw_Q$; and that $|\det(\bR)| =1$.

Define $\bK = \bR^{-1}$. Given the relationship in \eqref{eq:2}, then there exists another \emph{linear structure} representation of $\bZ$ such that $\bZ = \bK \bY$. Without loss of generality, we have $\E \bY \bY^\top = \bI$ (there is no loss of generality because we can scale $\bK$ such that $\E \bY \bY^\top = \bI$). From Theorem 10.3.3 in \cite{kagan1973characterization}, $\bZ$ has the decomposition $\bZ = \bK_1 \bY_1 + \bK_2 \bY_2$ in which $\bY_1$ are independent non-Gaussian and $\bY_2$ are Gaussian. Then from Theorem~1 and the assumption of unit variance, we have that $\bY_1 \eqd \bS$ (up to ordering), and it follows that there exists a subset of $g_1,\dots,g_T$ equal to $f_1,\dots,f_Q$. Also from Theorem~1, we have $\bK_1 \cong \bA_\bS$. Note that $\bK \in \cO_{T\times T}$ since $\E \bY \bY^\top = \bI$ and $\E \bZ \bZ^\top = \bI$, and hence $| \det (\bR) | = 1$. Then the scaling of $\bK_1$ is also identifiable such that there exists a signed permutation matrix, $\bP_{\pm}$, such that $\bK_1 \bP_{\pm} = \bA_\bS$. Note that $\bW_\bS = \bA_{\bS}^\top$. Define $\bR_\bS = \bK_1^\top$. Then $\bP_{\pm}^\top \bR_\bS  = \bW_\bS$.
\end{proof}

\subsection{Proofs for Section \ref{sec:ParametricLCA}}\label{WS:parametricLCA}

To simplify notation, we assume $\E \bX = \bzero$ but include the estimate of the mean $\barbx$ in our analysis so this assumption is without loss of generality. Let $f_\bS$ denote the joint density of the LCs, and similarly define $p_\bS(\bs) = \prod_{q=1}^Q p_q(s_q)$ for the densities used in \eqref{eq:ObjectiveFunction}. Let $\|\bA \|$ denote the Frobenius norm for $\bA \in \tR^{Q \times T}$. 

Next we discuss Assumption \ref{assumption:4} (ii) and inequality \eqref{eq:Prop4i}. The value of $\alpha$ will depend on the tail behavior of $\frac{d}{dx} \log\{p_q(x)\}$, $q=1,\ldots,Q$.  
For insight into this assumption, consider $Q=1$ such that $h(x) = \log p_1(x)$. By the mean value theorem,
\[
\|h(x_1) - h(x_0)\|  = \|h'(x^*)\|  \, \|x_1-x_0\|
\]
with $x^*$ between $x_0$ and $x_1$. Then if $h'$ is monotonic,
\begin{equation}
\|h(x_1) - h(x_0)\|  \le \left\{ \|h'(x_1)\| + \|h'(x_0)\|\right\}  \|x_1-x_0\|. \label{eq:Prop4ii}
\end{equation}
Therefore, if $\|h'(x)\|$ grows like $\|x\|^\alpha$ as $\|x\| \to \infty$, then \eqref{eq:Prop4i} will hold.

For example, for the exponential power density centered at 0, which is
\[
p_q(x) = \frac{\beta}{2\sigma \Gamma(1/\beta)} \exp\left\{-\left(\frac{|x|}{\sigma}\right)^\beta\right\},
\]
we have 
\[
\frac{d}{dx} \log\{p_q(x)\} = -\beta \, {\textrm{sign}}(x)\frac{|x|^{\beta-1}}{\sigma^\beta}, \ x \ne 0,
\]
%which is bounded in the heavier-tailed case of $\beta \le 1$. For the lighter-tailed case of $\beta > 1$, 
which is bounded for $\beta = 1$. For $\beta > 1$, we can take $\alpha = \beta-1$.  For $\beta<1$, the exponential power density has an unbounded score function at zero, but similar densities can be constructed with exponential power law tails such that one can take $\alpha = 0$. The student-$t$ distributions and the logistic distribution are other examples where $\frac{d}{dx} \log\{p_q(x)\} $ is bounded, so $\alpha = 0$.  % and, in fact, converges to 0 as $|x| \to \infty$, so we can use $\alpha=0$.  
At least in these examples, lighter tails require large values of $\alpha$, but, fortunately, make it easier for $E(\|\bS\|^{1+\alpha}) < \infty$ to hold.

Equation \eqref{eq:Prop4ii} shows that \eqref{eq:Prop4i} cannot be replaced by something like
\[
\|h(x) - h(x')\| \le M \|x-x'\| \,  \Big\{ 1 +\|x-x'\|^\alpha  \Big\}.
\]

The following two propositions are used to prove consistency with pre-whitening. Recall that $\cJ_n$ is defined in \eqref{eq:ObjectiveFunction} of the main manuscript. 
\begin{prop}\label{prop:prewhiten} $\cJ_n(\bO_\bS \hatbL(\bx_i - \bar{\bx})) \as \cJ_n(\bO_\bS \bL \bx_i)$
\end{prop}
\begin{proof}
First note that
\[
\cJn(\bO_\bS \hatbL(\bx_i-\barbx)) = \cJn(\bO_{\bS} \bL \bx_i) + R_n
\]
where
\begin{eqnarray*}
\|R_n\| &=& \left\| \cJn( \bO_{\bS} \hatbL(\bx_i-\barbx) ) - \cJn( \bO_{\bS} \bL \bx_i) \right\| \\
&\le& \frac{1}{n} \sum_{i=1}^n  \left\| h( \bO_{\bS} \hatbL(\bx_i-\bx)  ) - h( \bO_{\bS} \bL \bx_i) \right\|. 
\end{eqnarray*}
Using \eqref{eq:Prop4i},
\begin{eqnarray}
 & &
 \frac{1}{n} \sum_{i=1}^n  \left\| h( \bO_{\bS} \hatbL(\bx_i-\barbx))   - h( \bO_{\bS} \bL \bx_i)  \right\|  \le  
 M \bigg\{ \frac{1}{n} \sum_{i=1}^n \|\bO_\bS (\hatbL(\bx_i - \barbx) - \bL \bx_i) \| 
 \label{eq:Prop4ugly}
\\ & + &
 \frac{1}{n} \sum_{i=1}^n \|\bO_\bS (\hatbL(\bx_i - \barbx) -  \bL\bx_i ) \|  \, \|  \bO_\bS \hatbL(\bx_i - \barbx)\|^\alpha  +
  \frac{1}{n} \sum_{i=1}^n \|\bO_\bS (\hatbL(\bx_i - \barbx) - \bL\bx_i) \| \, \|  \bO_\bS \bL \bx_i \|^\alpha 
 \bigg\}  \nonumber 
 \end{eqnarray}
 Then since $\bO_\bS$ is semi-orthogonal, the right-hand side of \eqref{eq:Prop4ugly} is at most
\begin{eqnarray}
  & &
 M \bigg\{ \frac{1}{n} \sum_{i=1}^n \|\hatbL(\bx_i - \barbx) - \bL \bx_i \| + 
 \frac{1}{n} \sum_{i=1}^n \|\hatbL(\bx_i - \barbx) - \bL\bx_i \|  \, \|  \hatbL(\bx_i - \barbx)\|^\alpha \nonumber \\
 &  + &  \frac{1}{n} \sum_{i=1}^n \|\hatbL(\bx_i - \barbx) - \bL\bx_i\| \, \|  \bL \bx_i \|^\alpha 
 \bigg\}.
 \label{eq:Prop4iii}
\end{eqnarray}
Note that $\{\E \|\bx_i\|^{1+\alpha}\}^{1/(1+\alpha)} = \{ \E \| \bM \bz_i \|^{1+\alpha} \}^{1/(1+\alpha)} \le  \|\bM \| \left\{ \E (\| \bs_i \|+\|\bn_i \|)^{1+\alpha}\right\}^{1/(1+\alpha)} \le \|\bM \| (\E \| \bs_i \|^{1+\alpha})^{1/(1+\alpha)}+\| \bM \| (\E \|\bn_i \|^{1+\alpha})^{1/(1+\alpha)} < \infty$, where the last inequality uses Assumption \ref{assumption:4} (iii) and properties of the normal distribution. For the first term on the right-hand side of \eqref{eq:Prop4iii}
\[
\frac{1}{n} \sum_{i=1}^n \|\hatbL(\bx_i - \barbx) - \bL \bx_i\| \le  
\left\{ \| \hatbL - \bL\|  \, \frac{1}{n} \sum_{i=1}^n \|\bx_i \| + \|\hatbL\| \, \|\barbx\| \right\} \as 0,
\]
since $\hatbL \as \bL$, $\barbx \as 0$, and $\bx_1,\ldots,\bx_n$ are iid so we can apply the strong law of large numbers: 
\[
\frac{1}{n} \sum_{i=1}^n \|\bx_i \| \as \E(\|\bx_i\| ) \le \{\E(\|\bx_i\|^{1+\alpha})\}^{1/(1+\alpha)} < \infty.
\]
For the second term on the right hand side of \eqref{eq:Prop4iii},
\begin{eqnarray*}
& &
\frac{1}{n} \sum_{i=1}^n \|\hatbL(\bx_i - \barbx) - \bL\bx_i\|  \, \|  \hatbL(\bx_i - \barbx)\|^\alpha   \\
 &\le&   
 \frac{1}{n} \sum_{i=1}^n \left(\|\hatbL - \bL\| \, \|\bx_i\| + \|\hatbL\| \, \| \barbx\|   \right) \|\hatbL(\bx_i - \barbx)\|^\alpha \\
 &\le&
 \| \hatbL - \bL\|  \frac{1}{n} \sum_{i=1}^n \|\bx_i\| \, \|\hatbL(\bx_i-\barbx)\|^\alpha
 + \|\hatbL\| \, \| \barbx\|  \frac{1}{n} \sum_{i=1}^n \|\hatbL(\bx_i - \barbx)\|^\alpha. \\
\end{eqnarray*}
To prove this converges to zero, we need to show the means are finite, but we can not directly apply a law of large numbers because the summands are not independent due to prewhitening. First, note that 
\begin{align*}
 \frac{1}{n} \sum_{i=1}^n \|\bx_i\| \, \|\hatbL(\bx_i-\barbx)\|^\alpha \le \frac{1}{n}\sum_{i=1}^n  \left\{ \|\bx_i\|^{1+\alpha} + \|\hatbL(\bx_i-\barbx)\|^{1+\alpha} \right\}.
 \end{align*}
Then it remains to be shown that $\lim \frac{1}{n} \sum_{i=1}^n \|\hatbL(\bx_i-\barbx)\|^{1+\alpha} < \infty$. We have
\begin{align}
 \frac{1}{n} \sum_{i=1}^n \|\hatbL(\bx_i-\barbx)\|^{1+\alpha} &\le \frac{1}{n}\sum_{i=1}^n \left\{ \|\hatbL \| \|\bx_i-\barbx \|\right\}^{1+\alpha} \nonumber \\
 &\le \| \hatbL \| ^{1+\alpha} \, \frac{1}{n} \sum_{i=1}^n \| \bx_i - \barbx \|^{1+\alpha} \label{eq:Prop4iv}
\end{align}
Now consider 
\begin{align}
\frac{1}{n} \sum_{i=1}^n \|\bx_i-\barbx\|^{1+\alpha} &\le \frac{1}{n} \sum_{i=1}^n (\|\bx_i\| + \|\barbx\|)^{1+\alpha} \nonumber \\
&\le \frac{1}{n} \sum_{i=1}^n (2 \|\bx_i \|)^{1+\alpha} + (2\|\barbx\|)^{1+\alpha} \label{eq:Prop4v}
\end{align}
Since $\E \| \bx_i\|^{1+\alpha}<\infty$, we apply the law of large numbers to conclude that \eqref{eq:Prop4v} $<\infty$, and we conclude that \eqref{eq:Prop4iv} $< \infty$. Then \eqref{eq:Prop4iii}$\as 0$  because  $\| \hatbL - \bL\| \as 0$ and $\| \barbx\| \as 0$.

The third term on the right-hand side of \eqref{eq:Prop4iii} can be handled similarly.
\end{proof}

\begin{prop}\label{prop:prewhiten2}
Let $B \subseteq \cO_{Q \times T}$. Then
\[ 
\underset{\bO_\bS \in B}{\sup} \cJ_n(\bO_\bS \hatbL (\bx_i - \barbx)) \le \underset{\bO_\bS \in B}{\sup} \cJ_n (\bO_\bS \bL \bx_i) + o(1)\quad\quad a.s.
\]
\end{prop}
 \begin{proof}
\begin{align*}
\underset{\bO_\bS \in B}{\sup} \cJ_n(\bO_\bS \hatbL (\bx_i - \barbx)) &\le \underset{\bO_\bS \in B}{\sup} \cJ_n (\bO_\bS \bL \bx_i) + \sup_{\bO_\bS \in B} \frac{1}{n} \sum_{i=1}^n  \left\| h( \bO_{\bS} \hatbL(\bx_i-\bx)  ) - h( \bO_{\bS} \bL \bx_i) \right\|
\end{align*}
Note that
\begin{align*}
\sup_{\bO_\bS \in B}  \frac{1}{n} \sum_{i=1}^n \|\bO_\bS (\hatbL(\bx_i - \barbx)) -  \bL \bx_i ) \| & \le \sup_{\bO_\bS \in B}  \| \bO_\bS \| \frac{1}{n} \sum_{i=1}^n \|\hatbL(\bx_i - \barbx) -  \bL\bx_i  \| \\
&\le \sqrt{Q} \frac{1}{n} \sum_{i=1}^n \|\hatbL(\bx_i - \barbx) -  \bL\bx_i  \|.
\end{align*}
Using the inequality in \eqref{eq:Prop4ugly} and the previous argument, we have
\begin{eqnarray}
& & \sup_{\bO_\bS \in B} \frac{1}{n} \sum_{i=1}^n  \left\| h( \bO_{\bS} \hatbL(\bx_i-\barbx)  ) - h( \bO_{\bS} \bL \bx_i) \right\| \le MQ \bigg\{ \frac{1}{n} \sum_{i=1}^n \|\hatbL(\bx_i - \barbx) - \bL \bx_i \| \nonumber \\
& + & \frac{1}{n} \sum_{i=1}^n \|\hatbL(\bx_i - \barbx) - \bL\bx_i \|  \, \|  \hatbL(\bx_i - \barbx)\|^\alpha  +   \frac{1}{n} \sum_{i=1}^n \|\hatbL(\bx_i - \barbx) - \bL\bx_i\| \, \|  \bL \bx_i \|^\alpha \bigg\}. \label{Prop5:i}
\end{eqnarray}
Using the same arguments as in Proposition \ref{prop:prewhiten} to analyze the inequality in \eqref{eq:Prop4iii}, we have \eqref{Prop5:i} $\as 0$. 
\end{proof}

The next proposition is used in the proof of Theorem \ref{thm:2}.
\begin{prop}\label{prop:2}
Consider a random vector $\bY \in \tR^T$ with density $f_\bY$ such that $\E \bY = \bzero$ and $\E \bY \bY^\top = \bI_T$. Then for any $\bo$ and $\bw$ such that $\bo^\top \bo = \bw^\top \bw = 1$, we have
\[
 \E \log \phi (\bo^\top \bY) = \E \log \phi (\bw^\top \bY).
\]
\end{prop}
\begin{proof}
We can ignore the normalizing constants of $\phi(x)$ and consider the quadratic term of the Gaussian kernel.  Then we have
$\E (\bo^\top \bY)^2 = \bo^\top \E(\bY \bY^\top) \bo = \bo^\top \bI \bo = \bo^\top \bo = 1$ and similarly for $\E (\bw^\top \bY)^2$.

\end{proof}

Next we prove consistency when the density used in the objective function equals the true density.
\begin{thm}
Suppose $\bX$ follows the LNGCA model in \eqref{eq:LNGCAmodel} with Assumptions~1-4. %Additionally assume $\E \bX = \bzero$ and $\E \bX \bX^\top = \bI$. % Let $\bW_\bS$ denote the first $Q$ rows of $\bM^{-1}\bL^{-1}$. 
Given an iid sample $\{\bx_i\}$, $\hatbW_\bS^{Or} \as \bW_\bS$ on the equivalence class of signed permutations.
\end{thm}
\begin{proof}
We will include the effects of centering with $\bar{\bx}$ in the discussion that follows such that it is without loss of generality that we assume $\E \bX = \bzero$. Then $\bX \sim (0,\bSigma)$ and let $\bSigma^{-1/2} = \bL$. 

We will show the assumptions in Wald's consistency proof as recast in Theorem 5.14 in \cite{van2000asymptotic} hold; a similar proof is in \cite{asymptopia}. Note that this theory applies to a set of maxima of the population objective function, and thus is convenient for the set defined by the equivalence class of signed permutations of $\bW_\bS$. For clarity, we use $o_p(1)$ notation to correspond to van der Vaart, but note that Propositions \ref{prop:prewhiten} and \ref{prop:prewhiten2} hold almost surely and the proof ultimately demonstrates strong consistency as in \cite{wald1949note} and \cite{asymptopia}. Recall $f_\bS$ denotes the joint density of the LCs. The conditions are not all numbered in \cite{van2000asymptotic}, so for ease of reference we now state them.  
\begin{enumerate}[(i.)]
\item The parameter space is compact. This is stated in \cite{asymptopia}, where as van der Vaart proves consistency for all compact subsets, $K$, of the parameter space. %Consistency holds for the entire parameter space when we can set $K$ equal to $\cO_{Q\times T}$, which is the case for $\hatbW_\bS^{Or}$.
\item $\log f_\bS(\bO_\bS \bL \bx)$ is upper-semicontinuous for almost all $\bx$; in van der Vaart, this corresponds to (5.12).
\item For every sufficiently small ball $U \subset \cO_{Q \times T}$, the function $\bO_\bS \mapsto \sup_{\bO_\bS \in U} \log f_\bS(\bO_\bS \bL \bx_i)$ is measurable and satisfies $\E \sup_{\bO_\bS \in U } \log f_\bS(\bO_\bS \bL \bX) <\infty$; in van der Vaart, this corresponds to (5.13).
\item $\E \log f_\bS (\bO_\bS \bL \bX) \le \E \log f_\bS(\bW_\bS \bL \bX)$ for any $\bO_\bS \in \cO_{Q \times T}$ with equality if and only if $\bO_\bS \cong \bW_\bS$; this assumption is part of the definition of $\Theta_0$ following assumption (5.13) in van der Vaart and is assumption (i) in \cite{asymptopia}.%the condition in Wald's approach that replaces the well-separated criteria in the continuous mapping argmax theorem of van der Vaart Theorem 5.7.
\item The estimator satisfies:
\[
 \cJ_n(\hatbW_\bS \hatbL (\bx_i - \barbx)) \ge \cJ_n(\bW_\bS \bL \bx_i) - o_p(1);
\]
in van der Vaart's notation, this corresponds to $M_n(\hat{\theta}_n) \ge M_n(\theta_0) - o_p(1)$. %Note here this condition can be intepreted as the estimator with pre-whitening is an approximate maximum of the objective function without pre-whitening.
\end{enumerate}
In addition to these conditions, we will outline van der Vaart's proof and provide additional justification to apply the law of large numbers, which is required because the observations are not iid due to pre-whitening.

First, $\cO_{Q \times T}$ is compact, and (i) is satisfied. Next, we assume continuous densities which implies upper semicontinuity (condition ii). From Assumption \ref{assumption:4} (i), the densities are bounded, say by some constant $A$, and we have $\E \sup_{\bO_\bS \in U} \log f_\bS(\bO_\bS \bL \bX) \le \E \log A < \infty$ and hence satisfy condition (iii).

We next show condition (iv) is satisfied. Let $\bW_\bN$ denote rows $Q+1$ to $T$ of $\bW$. Note that the fact that $\E \log f_\bS (\bO_\bS \bL \bX) \le \E \log f_\bS(\bW_\bS \bL \bX)$ does not hold trivially can be seen by the following argument:
\begin{linenomath*}
\begin{align*}
 \E \log \frac{f_{\bS} (\bO_\bS \bL \bX)}{f_\bS(\bW_\bS \bL \bX)} &= (\det \bL)\int \log \left\{\frac{f_{\bS} (\bO_\bS \bL \bx)}{f_\bS(\bW_\bS \bL \bx)} \right\} \left\{ f_\bS (\bW_\bS \bL \bx) \phi(\bW_\bN \bL \bx) \right\} d \bx \\
 &\le (\det \bL) \log \int \left\{\frac{f_{\bS} (\bO_\bS \bL \bx)}{f_\bS(\bW_\bS \bL \bx)} \right\} \left\{ f_\bS (\bW_\bS \bL \bx) \phi(\bW_\bN \bL \bx) \right\} d \bx \\
 &= (\det \bL) \log \int f_{\bS} (\bO_\bS \bL \bx) \phi(\bW_\bN \bL \bx) \; d \bx.
 \end{align*}
 \end{linenomath*}
We would like the last quantity to be equal to zero, in which case we would obtain the desired bound. Let $\bW^*$ be the $T \times T$ matrix formed by stacking $\bO_\bS$ and $\bW_\bN$. The term $f_{\bS} (\bO_\bS \bx) \phi(\bW_\bN \bx) $ is a density if and only if $|\!\det(\bW^*)|=1$, which is not true in general because $\bO_\bS$ may not be orthogonal to $\bW_\bN$. Consequently, this quantity could integrate to greater than one, in which case we would have $\E \log f_\bS (\bO_\bS \bL \bX) \le \E \log f_\bS(\bW_\bS \bL \bX) + \alpha$ for some $\alpha>0$, and the bound is not tight enough.

Then define an orthogonal matrix in $\cO_{T \times T}$ such that rows 1 to $Q$ are equal to $\bO_\bS$ and the other rows are arbitrary. Then
\begin{linenomath*}
\begin{align*}
\E \log \frac{f_\bS(\bO_\bS \bL \bX)}{f_\bS(\bW_\bS \bL \bX)} &= \E \log \frac{f_\bS(\bO_\bS \bL \bX) \phi (\bO_\bN \bL \bX)}{f_\bS(\bW_\bS \bL \bX)\phi(\bO_\bN \bL \bX)} \\
&= \E \log \frac{f_\bS(\bO_\bS \bL \bX) \phi (\bO_\bN \bL \bX)}{f_\bS(\bW_\bS \bL \bX)\phi(\bW_\bN \bL \bX)},
\end{align*}
\end{linenomath*}
where the second line follows from Proposition \ref{prop:2}. Then applying Jensen's inequality, we have
\begin{linenomath*}
\begin{align*}
\E \log \frac{f_\bS(\bO_\bS \bL \bX) \phi (\bO_\bN \bL \bX)}{f_\bS(\bW_\bS \bL \bX)\phi(\bW_\bN \bL \bX)} & \le (\det \bL) \log \int \left( \frac{f_\bS(\bO_\bS \bL \bX) \phi (\bO_\bN \bL \bX)}{f_\bS(\bW_\bS \bL \bX)\phi(\bW_\bN \bL \bX)} \right) {f_\bS(\bW_\bS \bL \bX)\phi(\bW_\bN \bL \bX)} d\, \bx\\
  &= (\det \bL) \log \int f_{\bS} (\bO_\bS \bL \bx) \phi(\bO_\bN \bL \bx) d\, \bx \\
  &= 0,
\end{align*}
\end{linenomath*}
which holds with equality if and only if $f_\bS(\bO_\bS \bL \bx) \phi (\bO_\bN \bL \bx) = f_\bS(\bW_\bS \bL \bx)\phi(\bW_\bN \bL \bx)$, where the only if direction is a consequence of absolute continuity. Now suppose equality holds for the matrix $\bO_\bS^*$. Define $\bO_+ = [{\bO_{\bS}^*}^\top,{\bO_{\bN}^*}^\top]^\top$ such that $\bO_+ \in \cO_{T \times T}$. Let $\bY$ be a random variable with density $f_\bS(\bO_\bS^* \by) \phi (\bO_\bN^* \by) = f_\bS(\bW_\bS \by)\phi(\bW_\bN \by)$. Then there exist random variables $\bR_+$ and $\bR$ such that $\bY = \bO_+ \bR_+$ and $\bY = \bW \bR$. Applying Theorem 1, we have $\bO_\bS^* \cong \bW_\bS$. It follows that
\[
\E \log f_\bS (\bO_\bS \bL \bX) < \E \log f_\bS (\bW_\bS \bL \bX)
\]
for all $\bO_\bS \not \cong \bW_\bS$. 

To show condition (v) is satisfied, %from Proposition \ref{prop:prewhiten}, we have $\cJ_n(\bW_\bS \hatbL( \bx_i - \bar{\bx})) = \cJ_n(\bW_\bS \bL \bx_i) - o_p(1)$. Then 
\begin{align*}
\cJ_n(\hatbW_\bS \hatbL (\bx_i - \bar{\bx}))  & \ge \cJ_n(\bW_\bS \hatbL( \bx_i - \bar{\bx})) &  \text{(by definition)}\\
&= \cJ_n(\bW_\bS \bL \bx_i ) - o_p(1). & \text{(Proposition \ref{prop:prewhiten})}
\end{align*}
In other words, our estimator $\hatbW_\bS$ with $\cJ_n$ defined using the sequence $\{\hatbL,\bar{\bx}\}$ is an approximate maximum of the exact maximum of the function $\cJ_n(\bO_\bS \bL \bx_i )$.

In this paragraph, we recount the first half of the proof of \cite{van2000asymptotic} 5.14. Let $\cW_\bS$ be the set of signed permutations of $\bW_\bS$. Fix some $\bO_\bS^\dagger \notin \cW_\bS$ with $\bO_\bS^\dagger \in \cO_{Q\times T}$, and let $U_\ell$ be a decreasing sequence of open balls around $\bO_\bS^\dagger$ with diameter converging to zero. Define the function: $m_{U_\ell}(\bx_i) = \sup_{\bO_\bS \in U_\ell} \log f_\bS(\bO_\bS \bL \bx_i)$. Then using (ii) we have $m_{U_\ell}(\bx_i) \downarrow \log f_\bS (\bO_\bS^\dagger \bL \bx_i)$ and from (iii) we can apply the monotone convergence theorem to obtain $\E m_{U_\ell}(\bx_i) \downarrow \E \log f_\bS (\bO_\bS^\dagger \bL \bx_i)$. From (iv), we have $\E \log f_\bS (\bO_\bS^\dagger \bL \bX) < \E \log f_\bS (\bW_\bS \bL \bX)$. Then with the previous argument, for any $\bO_k \in \cO_{Q \times T} \setminus \cW_\bS$, we can define a set $U_{\bO_k}$ such that $\E m_{U_{\bO_k}} (\bx_i) < \E \log f_\bS (\bW_\bS \bL \bX)$. Now let $\epsilon$ be given and consider the set $B = \{\bO_\bS \in \cO_{Q\times T}: \underset{\bW_\bS^* \in \cW_\bS}{\cap} ||\bO_\bS - \bW_\bS^*|| \ge \epsilon\}$, which is compact. This set is covered by the balls $U_{\bO_k}$. Then there exists a finite subcover $U_1,\dots,U_p$. 

Next, we detail the second half of the proof of \cite{van2000asymptotic} 5.14, where we incorporate Proposition \ref{prop:prewhiten2} to account for pre-whitening. In the argument that follows, note that if $\E m_{U_k}(X) = - \infty$ for some $k$, then we can discard the set $U_k$, and since we have $\E m_{U_j}(X) < \infty$ from (iii), we have $\E |m_{U_j}(X)|<\infty$ for all remaining sets, and $\frac{1}{n} \sum_{i=1}^n m_{U_j} (\bx_i) \as \E m_{U_j}$ from the law of large numbers. 
% and keep track of the error $o_p(1)$ in lines \eqref{eq:a}, \eqref{eq:b}, and \eqref{eq:c} are unchanged: %\begin{align}
% \nonumber & \sup_{\bO_\bS \in \cO_{Q \times T}} \frac{1}{n} \sum_{i=1}^n \log f_\bS (\bO_\bS \hatbL (\bx_i - \bar{\bx})) = \\
% & \sup_{\bO_\bS \in \cO_{Q \times T}} \frac{1}{n} \sum_{i=1}^n \log f_\bS (\bO_\bS \bL \bx_i) + \frac{1}{n} \sum_{i=1}^n  \log f_\bS(\bO_\bS \hatbL (\bx_i -\barbx)) - \frac{1}{n} \sum_{i=1}^n \log f_\bS (\bO_\bS \bL \bx_i ) \le \\
% &\sup_{\bO_\bS \in \cO_{Q \times T}} \frac{1}{n} \sum_{i=1}^n \log f_\bS (\bO_\bS \bL \bx_i) + o(1)
% \end{align}
% where we have applied Proposition \ref{prop:prewhiten2}. 
% 
% Then we have
\begin{align}
 \sup_{\bO_\bS \in B} \cJ_n (\bO_\bS \hatbL (\bx_i - \bar{\bx})) &\le \; \sup_{\bO_\bS \in B} \cJ_n (\bO_\bS \bL \bx_i ) + o_p(1) & \text{(from Proposition } \ref{prop:prewhiten2}\text{)} \label{eq:a}\\
 &\le \sup_{j=1,\dots,p} \sup_{\bO_\bS \in U_j} \cJ_n (\bO_\bS \bL \bx_i) + o_p(1) & \nonumber \\
 &\le \;\;\sup_{j=1,\dots,p} \frac{1}{n} \sum_{i=1}^n m_{U_j} (\bx_i) + o_p(1) & \nonumber \\
 &\rightarrow \sup_{j=1,\dots,p} \E m_{U_j} (\bX) & \text{(law of large numbers)} \nonumber \\
 &< \;\;\;\E \log f_\bS (\bW_\bS \bL \bX).& \label{eq:b}
\end{align}
Now if $\hatbW_\bS \in B$, then we have 
\begin{align*}
%\sup_{\bO_\bS \in B} \cJ_n (\bO_\bS \hatbL (\bx_i - \barbx)) &\ge \cJ_n(\hatbW_\bS \hatbL (\bx_i - \barbx)) - o_p(1) & \text{(from condition (v.))} \\
%&= \cJ_n(\hatbW_\bS \bL \bx_i) - o_p(1) & \text{(from Proposition }\ref{prop:prewhiten}\text{)}\\
\sup_{\bO_\bS \in B} \cJ_n (\bO_\bS \hatbL (\bx_i - \barbx)) &\ge \cJ_n(\bW_\bS \bL \bx_i) - o_p(1) & \text{(from condition (v.))} \\
&= \E \log f_\bS (\bW_\bS \bL \bX) - o_p(1), & \text{(from LLN)}
\end{align*}
which would imply the following relationship between events:  	
\begin{align}
 \left\{ \hatbW_\bS \in B \right\} \subset \left\{ \sup_{\bO_\bS \in B} \cJ_n ( \bO_\bS \hatbL (\bx_i - \barbx)) \ge \E \log f_\bS (\bW_\bS \bL \bX) - o_p(1) \right\}.\label{eq:c}
\end{align}
In view of \eqref{eq:a} and \eqref{eq:b}, the probability of the event on the right-hand side of \eqref{eq:c} converges to zero as $n \rightarrow \infty$. Note the $o_p(1)$ inequalities hold almost surely from Propositions \ref{prop:prewhiten} and \ref{prop:prewhiten2}. Then 
\[
P \left( \lim_{n \to \infty} \underset{\bW_\bS^* \in \cW_\bS}{\bigcap} \left\{||\hatbW_\bS - \bW_\bS^*|| \ge \epsilon \right\} \right) \rightarrow 0.
\]
\end{proof}

Next we describe conditions for consistency when the density used in the objective function may not be equal to the density of the LCs. We first present a result that is contained in the proof of Theorem 1 in \cite{Hyvaerinen1998}, where here the nonlinearity is equal to the log of the density used in the objective function. 

Recall that $r_q(\cdot)$ denotes the score function of $\log f_q (\cdot)$ and  $r_q'(\cdot)$ denotes the derivative of the score function. Additionally, define $\bZ = [\bS^\top,\bN^\top]^\top$.
\begin{lemma}\label{lemma:2}
Let $\be_1 = [1,0,\dots,0]^\top$ and let $\bepsilon$ be given such that $||\be_1 + \bepsilon|| = 1$. Then
\begin{linenomath*}
 \begin{align*}
  \E \log p_1\left[ (\be_1 + \bepsilon)^\top \bZ \right] = \E \log p_1(S_1) + \frac{1}{2}\left[\E r_1'(S_1) - \E S_1 r_1(S_1) \right] \sum_{q=2}^T \epsilon_q^2 + o(||\bepsilon||^2).
 \end{align*}
 \end{linenomath*}
\end{lemma}
\begin{proof}
 Calculating the gradient with respect to $\bo$,
 \[
  \nabla \E \log p_1(\bo^\top \bZ) = \E \bZ r_1(\bo^\top \bZ),
 \]
where we have applied Assumption 5(iv) to interchange differentiation and integration. Evaluating this at $\be_1$, and using the fact that $\E S_q = \E N_k = 0$, $q=1,\dots,Q$, $k=1,\dots,T-Q$, and the fact that $\bS_1$ is independent of $\bS_q$, $q>1$, and $\bN_k$,
%  \[
%   \nabla \E \log p_1(\bo^\top\bZ)\bigg|_{\be_1} = \be_1 \E S_1 r_1(S_1).
%  \]
\[
  \nabla \E \log p_1(\be_1^\top\bZ) = \be_1 \E S_1 r_1(S_1).
 \]
We also have
\[
% \nabla^2 \E \log p_1(\bo^\top \bZ) \bigg|_{\be_1} = \diag \left[ \E S_1^2 r_1'(S_1), \E r_1'(S_1),\dots, \E r_1'(S_1) \right]
  \nabla^2 \E \log p_1(\be_1^\top \bZ) = \diag \left[ \E S_1^2 r_1'(S_1), \E r_1'(S_1),\dots, \E r_1'(S_1) \right]
\]
where as before we have interchanged integration and differentiation using Assumption 5(iv) and applied independence and the fact that $\E S_q^2 = \E N_k^2 = 1$.

Now for some small $\bepsilon$ with $||\be_1 + \bepsilon|| = 1$, we have
\begin{linenomath*}
\begin{align*}
 &\E \log p_1[(\be_1 + \bepsilon)^\top \bZ] = \\
 &\E \log p_1(S_1) + \bepsilon^\top\be_1 \E S_1 r_1(S_1) + \frac{1}{2} \bepsilon^\top  \diag \left[ \E S_1^2 r_1'(S_1), \E r_1'(S_1),\dots, \E r_1'(S_1) \right] \bepsilon + o(||\bepsilon||^2)=\\
& \E \log p_1(S_1) + \epsilon_1 \E S_1 r_1(S_1) + \frac{1}{2} \epsilon_1^2 \E S_1^2 r_1'(S_1) + \frac{1}{2}\E r_1'(S_1)\sum_{q>1} \epsilon_q^2 + o(||\bepsilon||^2).
\end{align*}
\end{linenomath*}
Note that $\epsilon_1 = \sqrt{1 - \sum_{q>1} \epsilon_q^2} - 1$. Now we consider the first-order Taylor series expansion of $\sqrt{1 - \gamma}$ about 0 which is $1 - \gamma/2 + o(||\gamma||)$, so $\epsilon_1 = - \frac{1}{2} \sum_{q>1} \epsilon_q^2 + o(\sum_{q>1} \epsilon_q^2)$. By Assumption 5(ii),  $|\E S_1^2 r_1'(S_1)|< \infty$. Then we can write
\begin{linenomath*}
\begin{align*}
 \E \log p_1\left[ (\be_1 + \bepsilon)' \bZ\right]  &= \E \log p_1(S_1) + \frac{1}{2}\left[ \E r_1'(S_1) - \E S_1 r_1(S_1) \right] \sum_{q>1} \epsilon_q^2 + o(||\bepsilon||^2).
\end{align*}
\end{linenomath*}
\end{proof}

\begin{prop*}
 (Proposition \ref{prop:neighbor} in the main manuscript.) Suppose Assumptions 1-3 and 5. There exists $\cN_{\epsilon^*}(\bW_\bS)$ such that $\E \! \log p(\bO_\bS \bL \bX)$ constrained to $\bO_\bS \in \cN_{\epsilon^*}(\bW_\bS)$ is maximized at $\bW_\bS$. 
\end{prop*}
\begin{proof}
We consider a perturbation of $\bW_\bS$. Using the change of variables $\bZ = \bW\bL\bX = [\bS^\top, \bN^\top]^\top$, it suffices to consider the case where $\bw_q = \be_q$, where $\be_{qt} = 1$ for $q=t$ and 0 otherwise. For $q=1$, consider a perturbation $\bepsilon_1 \in \tR^T$ with $||\be_1 + \bepsilon_1|| = 1$. From Lemma \ref{lemma:2}, we have
 \begin{linenomath*}
 \begin{align*}
  \E \log p_1[ (\be_1 + \bepsilon_1)^\top \bZ] &= \E \log p_1(S_1) + \frac{1}{2} \E \left[ r_1'(S_1) - S_1 r_1(S_1) \right] \sum_{q>1}\epsilon_{1q}^2 + o(||\bepsilon_1||^2).
 \end{align*}
 \end{linenomath*}
By Assumption 5(i), which states $\E r_q'(S_q) - \E S_q \, r_q (S_q) < 0$, and for sufficiently small $\bepsilon_1$, we have
\[
\frac{1}{2} \E \left[ r_1'(S_1) - S_1 r_1(S_1) \right] \sum_{q>1}\epsilon_{1q}^2 + o(||\bepsilon_1||^2)< 0,
\]
which makes $\be_1$ a local maximum for $\E \log p_1(\bo^\top \bZ)$. Since this also true for $\E \log p_q(\bo^\top \bZ)$, $q=2,\dots,Q$, we have that $\bI_{Q \times T}$ (the $Q \times Q$ identity matrix padded with zeros) is a local maximum on the set $\mathcal{G}_{Q \times T} = \left\{ \bG \in \tR^{Q \times T}: \diag\;{\bG \bG^\top} = \bone_Q\right\}$. Since $\cO_{Q \times T} \subset \mathcal{G}_{Q \times T}$ and $\bI_{Q\times T} \in \cO_{Q\times T}$, $\bI_{Q \times T}$ is also a local maximum on $\cO_{Q\times T}$. (For a similar argument in ICA, see \citealt{wei2015convergence}). Then for the perturbations $\bepsilon_1,\dots,\bepsilon_Q$, it suffices to let $\epsilon^* = \min_{q=1}^Q \min_{t=1}^T \epsilon_{qt}$, and define $\cN_{\epsilon^*}(\bW_\bS)$.  
\end{proof}

\begin{thm}
Suppose $\bX$ follows the LNGCA model in \eqref{eq:LNGCAmodel} with Assumptions 1-5. %Additionally assume $\E \bX = \bzero$ and $\E \bX \bX^\top = \bI$.
Given an iid sample $\{\bx_i\}$, $\hatbW_\bS^{Local} \as \bW_\bS$ on the equivalence class of signed permutations.
\end{thm}
\begin{proof}
 We restrict the parameter space to $\cN_{\epsilon^*}(\bW_\bS)$. Wald's method for consistency of the MLE can be applied to the more general setting in which the wrong likelihood is used if the supremum of the population objective function corresponds to the set of true parameters (condition (iv) in Theorem \ref{thm:2}), which was proven in Proposition \ref{prop:neighbor} for the restricted parameter space $\cN_{\epsilon^*}(\bW_\bS)$. The other conditions are satisfied using the previous arguments in the proof of Theorem \ref{thm:2}. 
\end{proof}

\subsection{Proofs for Section \ref{sec:Spline-LCA}}

Next we show that the solution to the Spline-LCA objective function corresponds to a mean-zero density.
\begin{prop*} 
(Proposition \ref{prop:density} in the main manuscript.) Let $G$ be the class of all cubic splines $g: \tR \rightarrow \tR$. Consider the argmax of (11) of the main manuscript for $g_q \in G$ with $g_q$ denoting the tilt function for the $q$th component. Then (i) $\int \phi(u)  e^{g_q(u)}\,du = 1$ and (ii) $\int u \phi(u)  e^{g_q(u)}\,du = 0$ for each $q$.
\end{prop*}

\begin{proof}
It suffices to consider the case $Q^*=1$. Let $\bo_1$ be given. Let $G$ be the set of cubic splines and note that for any $g \subset G$, we can write $g(u) = \theta_0 + \theta_1 u + j(u)$ with $\theta_0 \in \tR$, $\theta_1 \in \tR$, and $j(u)$ does not depend on $\theta_0$ or $\theta_1$. Noting that $\partial (\int \phi(u)e^{g(u)} du)/\partial \theta_0 = \partial(e^{\theta_0} \int \phi(u) e^{\theta_1 u + j(u)}du)/\partial \theta_0 = \int \phi(u)e^{g(u)} du$, we have
\begin{linenomath*}
\begin{align*}
\frac{ \partial \ell_{pen}}{\partial \theta_0} &= 1 - \int \phi(u) e^{g(u)}\,du,
\end{align*}
\end{linenomath*}
from which it follows that at the optimum $g^*$, $\phi(u) e^{g^*(u)}$ is a density. Next, note that $\partial (\phi(u) e^{\theta_0+\theta_1 u + j(u)} / \partial \theta_1 = u \phi(u) e^{g(u)}$. Then,
\begin{linenomath*}
\begin{align*}
\frac{ \partial \ell_{pen}}{\partial \theta_1} &=  \frac{1}{n}\sum_{i=1}^n \bo_1^\top \hatbL (\bx_i - \barbx) - \int u \phi(u) e^{g(u)}\,du,
\end{align*}
\end{linenomath*}
where we have assumed $\int |u| \phi(u) e^{g(u)} du < \infty$ to interchange integration and differentiation. Then it follows that $\E U = 0$ for $U$ with density $\phi(u)e^{g^*(u)}$.
\end{proof}

\section{Additional Asymptotics for Section \ref{sec:ParametricLCA}}\label{WS:additionalasymptotics}

%\subsection{Notation}
In this section, we examine $\sqrt{n}$-consistency, asymptotic normality, and the asymptotic variances of the parametric LCA estimators. 

Recall that $r_q(\cdot)$  is the score function of $\log f_q (\cdot)$ and  $r_q'(\cdot)$ is the derivative of the score function. Define the following quantities:
 \begin{align*}
  \beta_q &= \E S_q^4 \\
%   \nu_q &= \E \log f_q(S_q) \\
   \eta_q &= \E r(S_q) \\
   \xi_q &= \E r(S_q)^2 - \eta_q^2 \\
   \lambda_q &= \E r(S_q)S_q \\
  \delta_q &= \E r'(S_q) \\
 %  \tau_q &= \E r'(S_q) S_q.
 \end{align*}
Also define the empirical expectation: $\En f(x_i) = \frac{1}{n}\sum_{i=1}^n f(x_i)$. Recall that $\be_q \in \tR^T$ such that $\be_{qq'} = 0$ for $q' \ne q$ and 1 for $q' = q$. 

We apply the approach used in \cite{virta2016projection} to derive asymptotic variances based on rewriting the objective function using Lagrange multipliers. \cite{virta2016projection} find non-Gaussian components using a modified version of symmetric fastICA but with the measure of non-Gaussianity equal to a convex combination of squared skewness and kurtosis. We adapt their approach to log likelihoods. For an arbitrary consistent estimator of the LNGCA model, $\hatbW_\bS$, define $\hatbB_\bS = \hatbW_\bS \hatbL$. Let $\bB_\bS$ be the first $Q$ rows of $\bM^{-1}$. Consistency of $\hatbB_\bS$ follows from Slutsky's theorem. Throughout the remainder of this section, we focus on $\hatbB_\bS$ rather than $\hatbW_\bS$. 

First, consider:
  \begin{align}
  \cL(\bC_\bS,\bTheta) &=  \sum_{q=1}^{Q} \En \left\{ \log p_q (\bc_q^\top (\bx_i - \barbx)) \right\} - \sum_{q=1}^Q \frac{\theta_{qq}}{2} (\bc_q^\top \hatbSigma \bc_q - 1) - \sum_{q=1}^{Q-1} \sum_{q'=q+1}^Q \theta_{q q'} \bc_q^\top \hatbSigma \bc_{q'}. \label{eq:lagrange}
 \end{align}
Consider the substitution $\bo_q^\top \hatbL = \bc_q$. Then we rewrite \eqref{eq:lagrange}:
  \begin{align}
  \cL(\bO_\bS,\bTheta) &=  \sum_{q=1}^{Q} \En \left\{ \log p_q (\bo_q^\top \hatbL (\bx_i - \barbx)) \right\} - \sum_{q=1}^Q \frac{\theta_{qq}}{2} (\bo_q^\top \bo_q - 1) - \sum_{q=1}^{Q-1} \sum_{q'=q+1}^Q \theta_{q q'} \bo_q^\top \bo_{q'} \label{eq:lagrangewhitening}. 
 \end{align}
Then the partial derivatives of \eqref{eq:lagrange} at $\hatbB_\bS$ equal zero.

In the special case where $\bM = \bI$, let $\hatbe_q$ be the estimate of the $q$th row of the true unmixing matrix $\bI_{Q \times T}$.

Next we define the conditions for $\sqrt{n}$-consistency and asymptotic normality.
\begin{assumption}\label{assumption:rootn}
For all $q$, the following expectations are finite: (i) $\E S_q^4$; (ii) $\E r_q^2(S_q)$; (iii) $\E r_q' (S_q)$; (iv) $\E r_q(S_q) S_q$; and (v) $ \E r_q'(S_q) S_q$. 
 \end{assumption}
 
\begin{lemma}
 Suppose $\E \bX = 0$, $\bM = \bI$, and Assumptions 1-6. Consider a consistent estimator, $\widehat{\bE}_\bS$, of the first $Q$ rows of $\bM^{-1}$ with the rows permuted and signs specified such that $\hatbe_q \rightarrow \be_q$. Let $\hatbe_{q q'}$ be the $q'$th element of $\hatbe_q$. Then
 \begin{align}
  \sqrt{n} (\hatbe_{q q'}) &= \sqrt{n} \frac{\En \left\{ (r_q(s_{iq}) - \eta_q) s_{iq'} - (r_{q'}(s_{iq'}) - \eta_{q'}) s_{iq} - (\delta_{q'} - \lambda_q) s_{iq}s_{ir}\right\}}{\delta_q - \lambda_q + \delta_{q'} - \lambda_{q'}} + o_p(1),\; q,q' \le Q \label{eq:lemmaeq1} \\
    \sqrt{n} (\hatbe_{qq} -1) &= -\sqrt{n} \frac{1}{2} \En (s_{iq}^2 - 1) + o_p(1),\; q \le Q \label{eq:lemmaeq2} \\
  \sqrt{n} (\hatbe_{q r}) &= \sqrt{n} \; \frac{ \En \left[ \left\{r_q(s_{iq}) - \eta_q \right\} n_{i,r-Q} - \lambda_q s_{iq}n_{i,r-Q}\right]}{\lambda_q - \delta_q}+ o_p(1),\; q \le Q, Q < r < T. \label{eq:lemmaeq3}
 \end{align}
\end{lemma}

\begin{proof}
At the estimates $\hatbe_q$, the Lagrangian in \eqref{eq:lagrange} enforces the constraints
 \begin{align}
 \hatbe_q^\top \hatbSigma \hatbe_{q'} &= 0, q \ne q' \label{eq:ortho} \\
  \hatbe_q^\top \hatbSigma \hatbe_q &= 1.\label{eq:unit}
  \end{align}
  Now we differentiate the Lagrangian with respect to $\bc_q$ and set the result equal to zero, and replace $\bc_q$ with the estimates $\hatbe_q$, $q=1,\dots,Q$:
 \begin{align}
  \En r_q(\hatbe_q^\top(\bx_i - \barbx))(\bx_i - \barbx) &= \theta_{qq} \hatbSigma \hatbe_q + \sum_{q' \ne q} \theta_{qq'} \hatbSigma \hatbe_{q'}.\label{eq:difflagrange}
 \end{align}
Next, write \eqref{eq:difflagrange} as
 \begin{align}
  \En r_q(\hatbe_q^\top(\bx_i - \barbx))(\bx_i - \barbx) &= \hatbSigma \sum_{q'=1}^Q \hatbe_{q'} \theta_{qq'}.\label{eq:int4}
 \end{align}
 Multiplying \eqref{eq:difflagrange} by $\hatbe_{q'}$ and applying \eqref{eq:ortho} and \eqref{eq:unit}, we get
 \begin{align*}
  \hatbe_{q'}^\top \En r_q(\hatbe_q^\top(\bx_i - \barbx))(\bx_i - \barbx) &= \theta_{qq'}.
 \end{align*}
Then substituting this expression into \eqref{eq:int4}, we write
 \begin{align}
  \En r_q(\hatbe_q^\top (\bx_i - \barbx)) (\bx_i - \barbx) &= \hatbSigma \left( \sum_{q'=1}^Q \hatbe_{q'} \hatbe_{q'}^\top \right) \left[ \En \left\{ r_q (\hatbe_q^\top (\bx_i - \barbx)) (\bx_i - \barbx) \right\} \right].\label{eq:eedef}
 \end{align}
This is the same estimating equation that appears in deflationary fastICA when $q=Q$,  see equation (4) in \cite{nordhausen2011deflation}, but here it applies to all $q\le Q$, and we replace the non-linearities with the log likelihoods. %Using equation (3) in \cite{nordhausen2011deflation}, see also \cite{miettinen2015fourth} and \cite{virta2016projection}, we have the following:
% \begin{align*}
%  \sqrt{n} \En \left\{ r(\hatbc_q^\top (\bx_i - \barbx))(\bx_i - \barbx) \right\} &= \sqrt{n} \En \left( r(\be_k^\top \bx_i) - \E \left\{ r(\be_k^\top \bx_i) \be_k^\top \bx_i \right\} %\right)
% \end{align*}
Then we apply Theorem 1 from \cite{nordhausen2011deflation}; see similar theorems in \cite{miettinen2015squared,miettinen2015fourth} and \cite{virta2016projection}, which requires Assumption \ref{assumption:rootn}: 	
\begin{align}
 \sqrt{n}\hatbe_{qq'} &= -\sqrt{n} \hatbe_{q'q} - \sqrt{n} \En (x_{iq}-\bar{x}_q)(x_{iq'}-\bar{x}_{q'}) + o_p(1),\; q  \ne q', q,q' \le Q \label{eq:int1}\\
 \sqrt{n} (\hatbe_{qq} - 1) &= -\frac{1}{2} \sqrt{n} \left\{ \En (x_{iq} - \bar{x}_q)^2 - 1 \right\} + o_p(1), q \le Q \label{eq:int2} \\
 \sqrt{n} \hatbe_{qr} &= \sqrt{n}\frac{1}{\lambda_q - \delta_q}\left[\be_r^\top \En \left\{ r_q(\be_q^\top \bx_i) - \eta_q \right\}\bx_i - \lambda_q \En (x_{iq} - \bar{x}_q)(x_{ir}-\bar{x}_r)\right] + o_p(1).  \label{eq:int3}
\end{align}
%where \eqref{eq:int1} and \eqref{eq:int2} are derived from the constraints in \eqref{eq:ortho} and \eqref{eq:unit} and \eqref{eq:int3} is derived from a Taylor series expansion around $\En g(\hatbc^\top \bx_i)(\bx_i - \barbx)$. 
Next note that 
\begin{align*}
 \sqrt{n} \left[\En ( x_{iq} - \bar{x}_q)^2 \right] &= \sqrt{n} \En x_{iq}^2 + o_p(1),
\end{align*}
since $\sqrt{n} \bar{x}_q^2 = o_p(1)$. 
%Example: $\sqrt{n} \bar{x} \Rightarrow N(0,1)$, by cont map thm, $n \bar{x}^2 \Rightarrow \chi^2_1$. Guaranteed with fourth moment assumptions. 
Similarly, $\sqrt{n}\bar{x}_q\bar{x}_r = o_p(1)$. %For this, $\sqrt{n} \bar{x}_q \Rightarrow N(0,1)$, and $\bar{x}_r \rightarrow 0$. 
Then applying $x_{iq} = s_{iq}$ and $x_{ir} = n_{i,r-Q}$, we obtain \eqref{eq:lemmaeq2} and \eqref{eq:lemmaeq3}.

To obtain \eqref{eq:lemmaeq1}, we derive a second expression for $\theta_{qq'}$ by performing the differentiation with respect to $\bc_{q'}$ and multplying by $\hatbe_q$:
 \begin{align}
  \En r_q(\hatbe_{q'}^\top(\bx_i - \barbx))\hatbe_{q}^\top(\bx_i - \barbx) &= \theta_{qq'}\label{eq:inter}.
 \end{align}
 This gives us the estimating equations:
 \begin{align}
  \En \left[ r_{q} \left\{ \hatbe_q^\top (\bx_i - \barbx)\right\} \hatbe_{q'}^\top (\bx_i - \barbx)\right] &=   \En \left[ r_{q'} \left\{ \hatbe_{q'}^\top (\bx_i - \barbx)\right\} \hatbe_{q}^\top (\bx_i - \barbx)\right],\;q,q' \le Q. \label{eq:ee1}
 \end{align}
The estimating equation in \eqref{eq:ee1} is also found in symmetric fastICA \citep{miettinen2015fourth,miettinen2015squared,wei2015convergence} but here restricted to $q,q' \le Q$, and we replace the nonlinearities by the log likelihoods.  Then \eqref{eq:lemmaeq1} is a special case of the symmetric case in Theorem 1 in \cite{miettinen2015squared} with additional details in the proof of Theorem 6 in \cite{miettinen2015fourth}, where here we revise the sign modification, $\pi_j$, in their theorem, to be equal to $-1$ when we use the log likelihood in lieu of their objective function. Again, we use the fact that the terms arising from centering converge at a faster rate and thus vanish from the asymptotic variances. Then we restate the symmetric case from Theorem 1 in \cite{miettinen2015squared} in terms of the iid non-Gaussian components. 
\end{proof}

\begin{thm}\label{thm:4}
Suppose Assumptions 1-6 and additionally let $\bM = \bI$ and $\E \bX = \bzero$. Consider a consistent estimator, $\widehat{\bE}_\bS$, of the first $Q$ rows of $\bM^{-1}$ with the rows permuted and signs specified such that $\hatbe_q \rightarrow \be_q$. Then for $q \le Q$, $\sqrt{n} (\hatbe_q - \be_q) \Rightarrow \cN(0,\bR_q)$ with
 \begin{align}
    \bR_q &= \frac{\beta_q-1}{4} \be_q \be_q^\top \label{eq:Rqwhite} \\
    \nonumber &\;\;+ \sum_{q' \ne q}^Q \frac{\xi_q + \xi_{q'} +  \delta_{q'}^2 - \lambda_q^2 - 2 \delta_{q'} \lambda_{q'}}{\left(\delta_q  - \lambda_q + \delta_{q'} - \lambda_{q'} \right)^2}\be_{q'} \be_{q'}^\top \\
  \nonumber  &\;\; + \frac{\xi_q - \lambda_q^2}{\left(\lambda_q - \delta_q \right)^2}\left(\bI - \sum_{q'=1}^Q \be_{q'} \be_{q'}^\top \right).
 \end{align}
\end{thm}
\begin{proof}
Asyptotic normality follows from the central limit theorem for the iid observations on the right-hand side of equations \eqref{eq:lemmaeq1}-\eqref{eq:lemmaeq3} together with Slutsky's theorem. The variances can be calculated directly from the previous lemma and correspond to the variances of symmetric fastICA for $q \le Q$ and the variances of deflationary fastICA for $r>Q$.
\end{proof}
Note that the asymptotic variances for symmetric fastICA are also derived in Theorem 8 in \cite{wei2015convergence} using a modified M-estimator approach. They are equivalent to \cite{miettinen2015squared} except the sign modification is replaced by the sign of the term $\E r'_q(S_q) - \E S_q r_q(S_q)$. In LCA, this is always equal to negative one due to Assumption 5(i), and then \cite{wei2015convergence} Theorem 8 is equivalent to the result presented here for $q,q' \le Q$ and $\bM = \bI$.  

It is straightforward to extend this result to arbitrary mixing matrices when the estimators are affine equivariant, and this property is used in the estimators considered in \cite{virta2016projection} and related works by \cite{nordhausen2011deflation} and \cite{miettinen2015fourth}.  Let $F_\bX$ be the cumulative distribution of $\bX$, and let $\cB(F_\bX) \in \tR^{Q \times T}$ be a functional. As defined in \cite{nordhausen2011deflation},
\begin{definition}
A functional $\cB(F_\bX)$ is affine equivariant if 
\[
\cB(F_{\bA \bX}) = \cB(F_\bX) \bA^{-1}.
\]
\end{definition}

\cite{wei2015convergence} proves that an estimator is affine equivariant if and only if it does not depend on initialization, and thus our estimators are not in general affine equivariant. In practice, we satisfy this requirement by initializing from a sufficiently large number of random orthogonal matrices, such that if we were to estimate the unmixing matrix with another set of random initial values, we would obtain the same estimate with high probability. Alternatively, one can use the two-stage estimator, $\hatbW_\bS^{LV}$, since the estimator from \cite{virta2016projection} is affine equivariant. 

For the following theorem, we additionally assume the estimator is globally consistent, for example, under finite eighth moment assumptions with $\hatbW_\bS^{LV}$, which simplifies the exposition by avoiding the dependency between the optimization space and the choice of mixing matrix. 
\begin{cor}\label{cor:av}
Suppose Assumptions 1-6. Let $\hatbB_\bS$ be a globally consistent and affine equivariant estimator of the LCA model for any full rank $\bM \in \tR^{T \times T}$ with $\bM^{-1} = \bB$, and let $\hatbB_\bS$ have rows permuted and signs chosen such that $\hatbB_\bS \rightarrow \bB_\bS$. Then for $q \le Q$, $\sqrt{n} (\hatbmb_q - \bmb_q) \Rightarrow \cN(0,\bR_q)$ with
 \begin{align}
    \bR_q &= \frac{\beta_q-1}{4} \bmb_q \bmb_q^\top  + \sum_{q' \ne q}^Q \frac{\xi_q + \xi_{q'} +  \delta_{q'}^2 - \lambda_q^2 - 2 \delta_{q'} \lambda_{q'}}{\left(\delta_q  - \lambda_q + \delta_{q'} - \lambda_{q'} \right)^2}\bmb_{q'} \bmb_{q'}^\top \label{eq:Rq} \\
  \nonumber  &\;\; + \frac{\xi_q - \lambda_q^2}{\left(\lambda_q - \delta_q \right)^2}\left(\bSigma^{-1} - \sum_{q'=1}^Q \bmb_{q'} \bmb_{q'}^\top \right).
 \end{align}
\end{cor}

\begin{proof}
Consider the trivial model: $\bz_i = \bI \bz_i$ and let $\hat{\bI}_\bS = \underset{\bO_\bS \in \cO_{Q \times T}}{\margmax} \;\; \cJ_n(\bO_\bS;\{\bz_i\})$. Define $\hatbW_\bS = \underset{\bO_\bS \in \cO_{Q \times T}}{\margmax} \;\; \cJ_n(\bO_\bS;\{\hatbL(\bx_i-\barbx)\})$. Then
\begin{align*}
 \hatbB_\bS &= \hatbW_\bS \hatbL \\
 &= \left[ \underset{\bO_\bS \in \cO_{Q \times T}}{\margmax} \;\;\; \cJ_n(\bO_\bS;\hatbL\{\bx_i - \barbx\}) \right] \hatbL \\
 &= \left[ \underset{\bO_\bS \in \cO_{Q \times T}}{\margmax} \;\;\; \cJ_n(\bO_\bS;\hatbL\bM \bz_i) \right] \hatbL \\
 &= \left[ \underset{\bO_\bS \in \cO_{Q \times T}}{\margmax} \;\;\; \cJ_n(\bO_\bS;\bz_i) \right] \bB \hatbL^{-1} \hatbL.
\end{align*}
Then $\hatbB_\bS$ is a linear transformation of the estimator in Theorem \ref{thm:4} and $\sqrt{n}$-consistency and asymptotic normality follow.

The asymptotic variance is a linear transformation of the asymptotic variance of the previous theorem. Define the $QT \times QT$ covariance matrix: $\Var \{ {\rm{vec}}(\hatbW_\bS) \} = \bR$. Using the fact $\rm{vec}(\bA \bC \bB) = (\bB^\top \otimes \bA) \rm{vec} (\bC)$, we have
\[
\Var \{ {\rm{vec}}(\bI_Q \hatbW_\bS \bB) \} = (\bB^\top \otimes \bI_Q) \bR (\bB \otimes \bI_Q)
\]
Now let $\bB_\bN$ be the rows of the full unmixing matrix corresponding to the Gaussian components. Restricting our attention to the block of this matrix corresponding to the covariance matrix for $\bmb_q$, then applying simplifications and the property that $\bB_{\bN}^\top \bB_{\bN} = \bSigma^{-1} - \bB_\bS^\top \bB_\bS$, we obtain \eqref{eq:Rq}.
\end{proof}

\cite{wei2015convergence} develop similar asymptotics for estimators using the theory of M-estimation without requiring affine equivariance; however, his approach does not readily extend to LNGCA and LCA. In particular, the identifiability issues created by the Gaussian components precludes the direct application to LNGCA. For $T=Q$, $\sum_{q'=1}^Q \bmb_{q'} \bmb_{q'}^\top = \bSigma^{-1}$, and Corollary \ref{cor:av} is equivalent to Theorem 8 in \cite{wei2015convergence} for the special case specified by our Assumption 5(i).

We validated the asymptotic approximation of the distribution of the Logis-LCA estimator on a finite sample through simulations. Here we present the results from a single random choice of $\bM$ with 10,000 simulations, $n=10,\!000$, $Q=2$, and $T=4$ in which the true densities were exponential and logistic. In Figure \ref{fig:Empirical_vs_Theoretical}, we can see that the histograms are in general agreement with the theoretical results.

\begin{figure}
 \includegraphics[width=\textwidth]{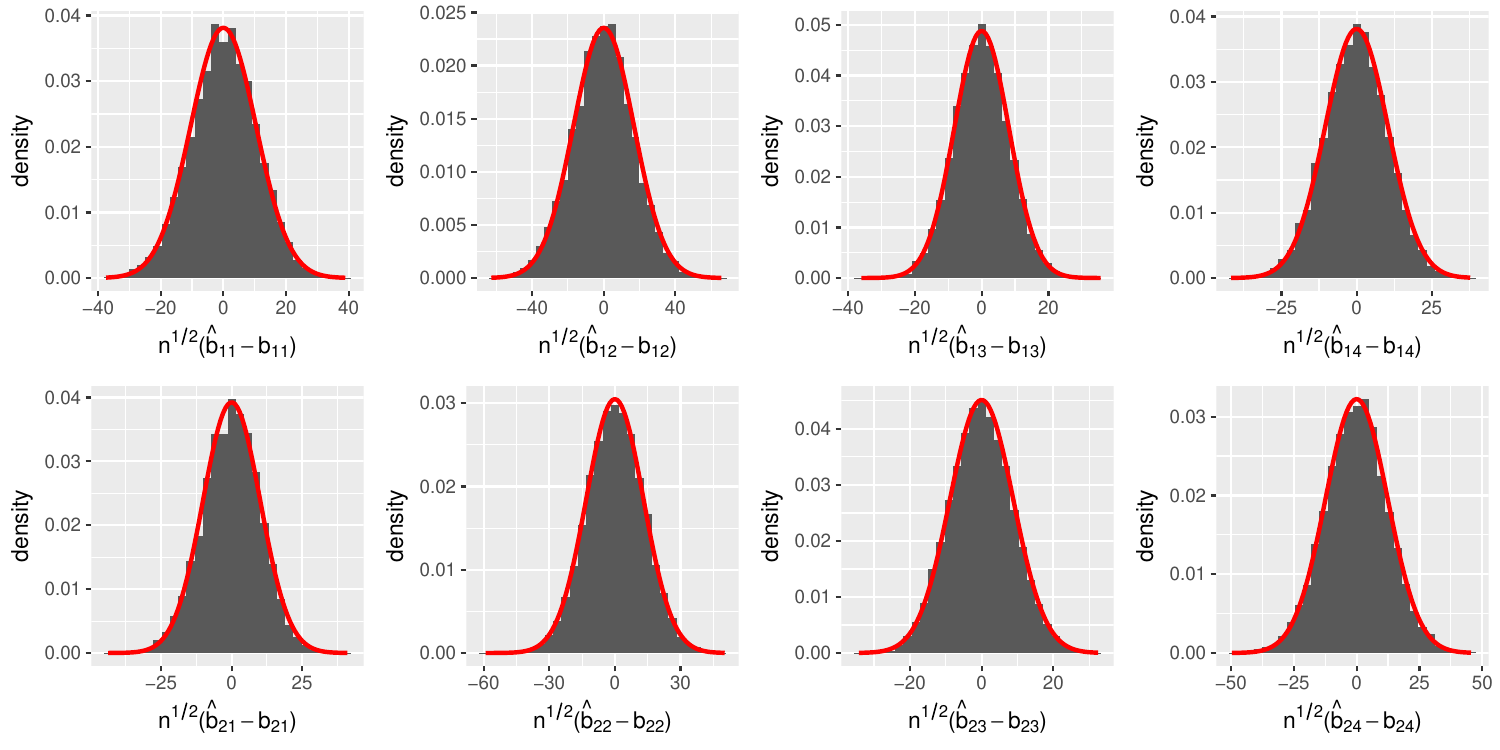}
 \caption{Theoretical densities versus histograms of $\sqrt{n}(\hatbmb^{Logis}_{qt} - \; \bmb_{qt})$ where $\hatbB_\bS^{Logis} = \hatbW_\bS^{Logis} \hatbL$ from 10,000 simulations with $n=$10,000, $Q=2$ with exponential and logistic densities, $T=4$,  and the true $\bB$ is fixed at a randomly generated matrix.}\label{fig:Empirical_vs_Theoretical}
\end{figure}

\section{Additional Background}\label{WS:addbackground}

\subsection{Projection Pursuit, D-FastICA, and Non-Gaussian Component Analysis}\label{WS:ppmethods}
Projection pursuit is an exploratory method for finding low-dimensional representations of multivariate data that reveal interesting patterns and structure \citep{huber1985projection}. Let $\{\bx_{\textrm{st},\; i}\}$, $i=1,\dots,n$ be the standardized data sample with $\bx_i \in \tR^T$, $\sum_{i=1}^n \bx_{\textrm{st},\; i} = \bzero$, where $\bzero$ is the vector of $T$ zeros,  and $\frac{1}{n}\sum_{i=1}^n \bx_{\textrm{st},\; i}^2 = \bone$, where $\bone$ is a length $T$ vector of ones. Let $Q$ be the number of projection pursuit directions that are estimated. In FastICA in deflation mode (D-FastICA), the projection pursuit index is equivalent to an approximation of negentropy \citep{hyvarinen1999fast}:
\begin{linenomath*}
\begin{align}\label{eq:D-FastICA}
 \bw_q = \underset{\bw \in \tR^T}{\margmax} \left\{ \frac{1}{n}\sum_{i=1}^n R(\bw^\top \bx_{\textrm{st},\; i}) - \int R(n) \phi(n) \, d n \right\}^2,
\end{align}
\end{linenomath*}
where $\bw$ is orthogonal to $\hatbw_1,\dots,\hatbw_{q-1}$ and $||\bw|| = 1$ with $||\cdot||$ denoting the L2-norm, $R$ is a non-linear function (in likelihood-based ICA, $R = \log f(x)$), and $\phi(n)$ is the standard normal density. A common choice for $R$ is $\log \cosh(x)$, which is used to estimate projection pursuit directions in our simulations.

NGCA uses multiple projection pursuit indices \citep{blanchard2006search} or radial basis functions \citep{kawanabe2007new} to find a non-Gaussian subspace that is assumed to contain the interesting features of the data. NGCA can be formulated using a semiparametric likelihood,
\begin{linenomath*}
\begin{align}\label{eq:NGCA}
 f_\bX(\bx) = h^*(\bB_\bS \bx) \phi_{\bzero,\bSigma}(\bx)
\end{align}
\end{linenomath*}
where $\phi_{\bzero,\bSigma}$ is multivariate normal with mean $\bzero$ and covariance $\bSigma$; $\bB_\bS$ is a ${Q \times T}$ matrix; and $h^*(\cdot)$ is a function that captures departures from Gaussianity under the constraint that $f_\bX(\bx)$ is a density. NGCA does not assume linear mixing of independent factors, and consequently the factors are not identifiable. Thus we do not consider it in our simulations.

The density in the Spline-LCA model can be considered an extension of \eqref{eq:NGCA} with the additional assumption of independence.

\begin{prop}\label{prop:NGCA}
Let $\bX$ be a random variable from the LCA model where the LCs have tilted Gaussian densities. Then the density of $\bX$ is
\[
 f_\bX (\bx) = \phi_{\bzero,\bSigma}(\bx) \prod_{q=1}^Q e^{g_q(\bw_q^\top \bL \bx)}
\]
where $\phi_{\bzero,\bSigma}$ is the mean zero multivariate distribution with covariance $\bSigma = \bL^{-2}$.
\end{prop}
\begin{proof}
Using the tilted Gaussian density, we have
\begin{linenomath*}
 \begin{align*}
  f_\bX(\bx) &= \det \bL \prod_{q=1}^Q e^{g_q(\bw_q^\top \bL \bx)} \phi(\bw_q^\top \bL \bx)\prod_{k=1}^{T-Q} \phi(\bw_{Q+k}^\top \bL \bx) \\
  &= \left\{\prod_{q=1}^Q e^{g_q(\bw_q^\top \bL \bx)} \right\} (2\pi)^{-T/2} \left( \det \bL \right) \exp\left\{-\frac{1}{2} \sum_{k=1}^T \bx^\top \bL \bw_k \bw_k^\top \bL \bx\right\} \\
  &= (\det \bSigma)^{-1/2} (2\pi)^{-T/2} \exp\left\{-\frac{1}{2} \bx^\top \bSigma^{-1} \bx \right\} \prod_{q=1}^Q e^{g_q(\bw_q^\top \bL \bx)}.
 \end{align*}
 \end{linenomath*}
\end{proof}
Writing the likelihood in this way, one notes that we are using the Gaussian density to model the covariance between components and we are using the tilt functions to model deviations from the Gaussian model.

\subsection{Noisy ICA and IFA}\label{supp:secNoisyICA}
In the noisy ICA model, $Q$ ICs are mixed and then corrupted by rank-$T$ Gaussian noise, where $Q \le T$ \citep{hyvarinen2001independent},
\begin{linenomath*}
\begin{align}\label{eq:noisyICA}
 \bX = \bM_\bS \bS + \bE
\end{align}
\end{linenomath*}
with $\bX \in \tR^T$, $\bM_\bS$ is $T \times Q$ with $Q \le T$, $\bE$ is mean-zero multivariate normal with covariance matrix $\bPsi$, and $\bE$ is independent of $\bS$.

Assume that $\bPsi = \sigma^2 \bI$. Let $d_1,\dots,d_Q$ denote the eigenvalues from the covariance matrix of $\bM_\bS \bS$ and let $d_{\epsilon_1},\dots,d_{\epsilon_T}$ denote the eigenvalues from the decomposition of $\bE$. Under the assumption of isotropic noise, we have $d_{\epsilon_i}  = \sigma^2$ for all $i,j=1,\dots,T$. Then the eigenvalue decomposition can be written as
\begin{linenomath*}
\begin{align}
 \Cov \bX = \bU \; \diag(d_1+\sigma^2,\dots,d_Q+\sigma^2,\sigma^2,\dots,\sigma^2)\; \bU^\top.
\end{align}
\end{linenomath*}
%Since the variance due to the signal is contained in the first $Q$ components, one could propose to estimate independent components by rotating the principal subspace (the span of the first $Q$ columns of $\bU$).
Let $\bX_{\rm data}$ be the $n \times T$ data matrix. In PCA+ICA, noise-free ICA is applied to the first $Q$ left singular vectors of $\bX_{\rm data}$ multiplied by $\sqrt{n}$, which is equivalent to the first $Q$ standardized principal components.

In IFA, \eqref{eq:noisyICA} is estimated under the assumption that the densities of the ICs are Gaussian mixtures \citep{attias1999independent}. In its original formulation, $\bPsi$ was an arbitrary positive definite matrix, the IC densities had $K_q$ classes, and the variance of each IC was standardized to unity after each iteration. In our presentation and estimation, we assume that the covariance of the noise is $\sigma^2 \bI$ and IC densities are mixtures of two Gaussians, which has been assumed elsewhere (e.g., \citealt{guo2013hierarchical,beckmann2004probabilistic}), and enforce the constraint that the IC densities are mean zero with unit variance. Let $\pi_{q1}$ be the probability that an observation of the $q$th IC comes from the first class, where the first class has a normal distribution with mean $\mu_{q1}$ and variance $\rho_{q1}$. Then the probability, mean, and variance for the second class are $\pi_{q2} = 1 - \pi_{q1}$, $\mu_{q2} = -\frac{\pi_{q1} \mu_{q1}}{\pi_{q2}}$, and $\rho_{q2} = \frac{1 - \pi_{q1} \rho_{q1} - \pi_{q1} \mu_{q1}^2}{\pi_{q2}} - \mu_{q2}^2$, respectively. Then the joint density of $\bX$ can be written
\begin{linenomath*}
\begin{align}\label{eq:IFA}
f_\bX (\bx \mid \bM_\bS) = \prod_{t=1}^T  \int \phi_{0,\sigma^2}\!\left( x_t - {\bf m}_t^\top \bs \right) f_{\bS}(\bs)\, d\bs,
\end{align}
\end{linenomath*}
where $\phi_{0,\sigma^2}$ is a normal density with mean zero and variance $\sigma^2$ and
\[
 f_\bS(\bs) = \prod_{q=1}^Q \left\{ \pi_{q1} \phi_{\mu_{q1},\rho_{q1}}(s_q) + \pi_{q2} \phi_{\mu_{q2},\rho_{q2}}(s_q)\right\}.
\]
Analytic integration across $\bs$ is possible. Let $k_q$ equal one if $s_q$ is in the first class and zero otherwise. Let $\mathcal{K}$ be the set of all possible states for the $Q$ components composed from the Cartesian product $Q$-times of the singletons $\{ \{0\},\{1\} \}$. Let ${\bf k}_j = \{k_1,\dots,k_Q\}$ denote an element of $\mathcal{K}$, where $j \in \{1,\dots,2^Q\}$. Let ${\bm \mu}({ \bf k}_j)$ and ${\bm \rho}({ \bf k}_j)$ denote the conditional means of $\bs$ given the states ${\bf k}_j$.  Now define
\[
 \bSigma({\bf k}_j) = \bM_\bS \; \diag\{ {\bm \rho}({ \bf k}_j)\}\; \bM_\bS^\top + \sigma^2 \bI
\]
and
\[
 {\bm \mu}^*({\bf k}_j) = \bM_\bS {\bm \mu}({ \bf k}_j).
\]
Then the density is
\begin{linenomath*}
\begin{align}\label{eq:IFAv2}
 f_\bX(\bx \mid \bM_\bS) = \sum_{{\bf k}_j \in \mathcal{K}} \Phi\{ \bx \mid {\bm \mu}^*({\bf k}_j), \bSigma({\bf k}_j)\} \prod_{q=1}^Q \pi_{q1}^{k_q}\pi_{q2}^{1-k_q}
\end{align}
\end{linenomath*}
with $\Phi\{ \bx \mid {\bm \mu}^*({\bf k}_j), \bSigma({\bm k}_j)\}$ multivariate normal with mean ${\bm \mu}^*({\bf k}_j)$ and variance $ \bSigma({\bm k}_j)$ (see (16) and (17) in \citealt{attias1999independent}). Then a likelihood can be constructed from \eqref{eq:IFAv2}, and given some $\hatbM_\bS$, the ICs can be estimated from their conditional means. Alternatively, maximum a posteriori estimates of the ICs could be obtained, though we pursue the former here.

\section{Using the fixed-point algorithm to fit the LCA model}\label{WS:fixedpoint}
Here we describe the fixed-point algorithm from \cite{hyvarinen1999fast}. Our account is equivalent to \cite{hyvarinen1999fast} except we use our novel discrepancy measure ($PMSE$) and a different orthogonalization method. Under the constraint that the noise components follow a standard normal distribution, we can ignore rows $Q^*+1 : T$ in $\hatbW$. Recall $r_q(x)$ and $r_q'(x)$ are the first and second derivatives of $\log f_q(x)$. Algorithm 1 provides details on estimating $\hatbW_\bS$.
\begin{algorithm}\label{algorithm:fastICA}
\SetKwInOut{Input}{Inputs}
\SetKwInOut{Output}{Output}
\SetKwFor{For}{for}{}{endfor}
\SetKwFor{While}{while}{}{endwhile}
\caption{The fastICA algorithm (symmetric fixed point) for LCA.}
\Input{The whitened $n \times T$ data matrix $\bX_{\textrm{st}}$; initial $\bW_\bS^0$; tolerance $\epsilon$.}
\KwResult{Estimates of the unmixing matrix, $\hatbW_\bS$, and latent components, $\hatbS = \bX_{\textrm{st}} \hatbW_{\bS}^\top$.}
\begin{enumerate}
 \item Let $\bS^{0} = \bX_{\textrm{st}} \bW_\bS^{0}{^\top}$ and let $(m)=0$, where $(m)$ denotes the number of update steps.
  \item For each row $\bw_q$, $q=1,\dots,Q$, of $\bW_\bS$, calculate
  \begin{align*}
  \bw_q^* = \frac{1}{n}\sum_{i=1}^n \left\{r_q(\bw_q^{(m)}{^\top}\bx_{\textrm{st},i})\bx_{\textrm{st},i}  - r_q'(\bw_q^{(m)}{^\top}\bx_{\textrm{st},i}) \bw_q^{(m)}\right\}
  \end{align*}
  \item Calculate the thin SVD of $\bW^*_\bS = \bU^*\:\bD^*\:\bV^{*}{^\top}$.
  \item Let $\bW^{(m+1)} = \bU^*\:\bV^{*}{^\top}$.
  \item If $PMSE(\bW_\bS^{(m+1)}{^\top},\bW_\bS^{(m)}{^\top}) < \epsilon$, stop, else increment $(m)$ and repeat (2)-(4).
\end{enumerate}
\end{algorithm}

\section{Supplemental materials for simulations examining distributional and noise-rank assumptions}

We fit D-FastICA using the `deflation' option in the fastICA R package \citep{Marchini:2010lr}. However, this popular function does not include an option to use projection pursuit for dimension reduction. If one specifies some $Q < T$ number of components, PCA is performed prior to the ICA. Consequently, one must estimate all $T$ directions and then subset to the first two.

We fit the IFA model with two-class mixtures of normals by maximizing the log likelihood using a numerical optimizer. This contrasts with methods using approximating EM algorithms, as described in the introduction. Our implementation is not scalable to large $Q$ or $T$ (nor is the exact EM algorithm) but suffices for the simulation experiments. For IFA, one must specify initial values for the unmixing matrix, the variance of the isotropic noise, and the parameters of the Gaussian mixtures. We had four strategies to find the argmax as detailed here. In our function, we constrain the latent component distributions to have zero expectation and unit norm, and as a result, the number of parameters to estimate for each latent component distribution is three. First, we estimated the parameters of the model proposed in \citet{beckmann2004probabilistic} (BS-PICA) and used this solution to initialize the IFA. We then estimated the model from six additional random matrices but with density parameters initialized from the BS-PICA solution. Secondly, when the IFA model was true, we initialized it from the true mixing matrix and true density parameters and also from six additional random matrices with density parameters initialized from their true values. When the IFA model was not true, we initialized it from the true mixing matrix but with the density parameters initialized from their BS-PICA estimates and an additional six random matrices. Thirdly, we initialized the algorithm from seven random matrices but with initial Gaussian mixture densities defined by the parameters (0.7, 0.7, $-0.5$, $-0.5$, 0.5, 0.5) (super-Gaussian distribution) for $\pi_{11}, \pi_{21}, \mu_{11}, \mu_{21}, \rho_{11}, \rho_{21}$ and $\sigma^2=1$. Finally, we initialized the algorithm from seven random matrices but with initial Gaussian mixture densities defined by the parameters (0.3, 0.3, $-1$, $-1$, 0.5, 0.5) (sub-Gaussian distribution) with $\sigma^2=1$.

The matrices $\bM_\bS$ and $\bM_\bN$ were generated by first simulating a 5$\times$5 matrix with standard normal entries, taking the singular value decomposition (SVD), then creating a diagonal matrix with five singular values from a uniform(1,10) distribution, followed by multiplying the left singular vectors from the SVD, the diagonal matrix, and the right singular vectors, which created $[\bM_\bS, \bM_\bN]$. For the noisy ICA model, we generated a random mixing matrix in the same manner, then retained the first two columns.

To generate semi-orthogonal random matrices to initiate the fixed point algorithm, matrices were generated by taking the left eigenvectors from the SVD of a $2 \times 5$ matrix with entries simulated from a standard normal. We generated random matrices constrained to the principal subspace in the following manner. Let $\hatbU_{1:Q}^\top$ denote the first $Q$ rows from $\hatbU^\top$ in the decomposition $\hatbSigma = \hatbU {\hat{\bLambda}} \hatbU^\top$. Then constraining the initial matrix, $\bW_\bS^{0}$, to the principal subspace is equivalent to $\bW_\bS^{0} = \bO \hatbU_{1:Q}^\top$ where $\bO$ is a random $Q \times Q$ orthogonal matrix.

\section{Supplemental figures for the spatio-temporal sources}
\begin{figure}
 \caption{Network recovery from the noisy-ICA scenario with $Q=3$ for $Q^*=2$, 3, or 4.}\label{fig:simFMRI123_NoisyICA}

  $\hspace{0.7in} \bf Q^* = 2 \hspace{1.5in} \bf Q^* = 3 \hspace{1.9in} \bf Q^* = 4$

\includegraphics[width=0.975\textwidth]{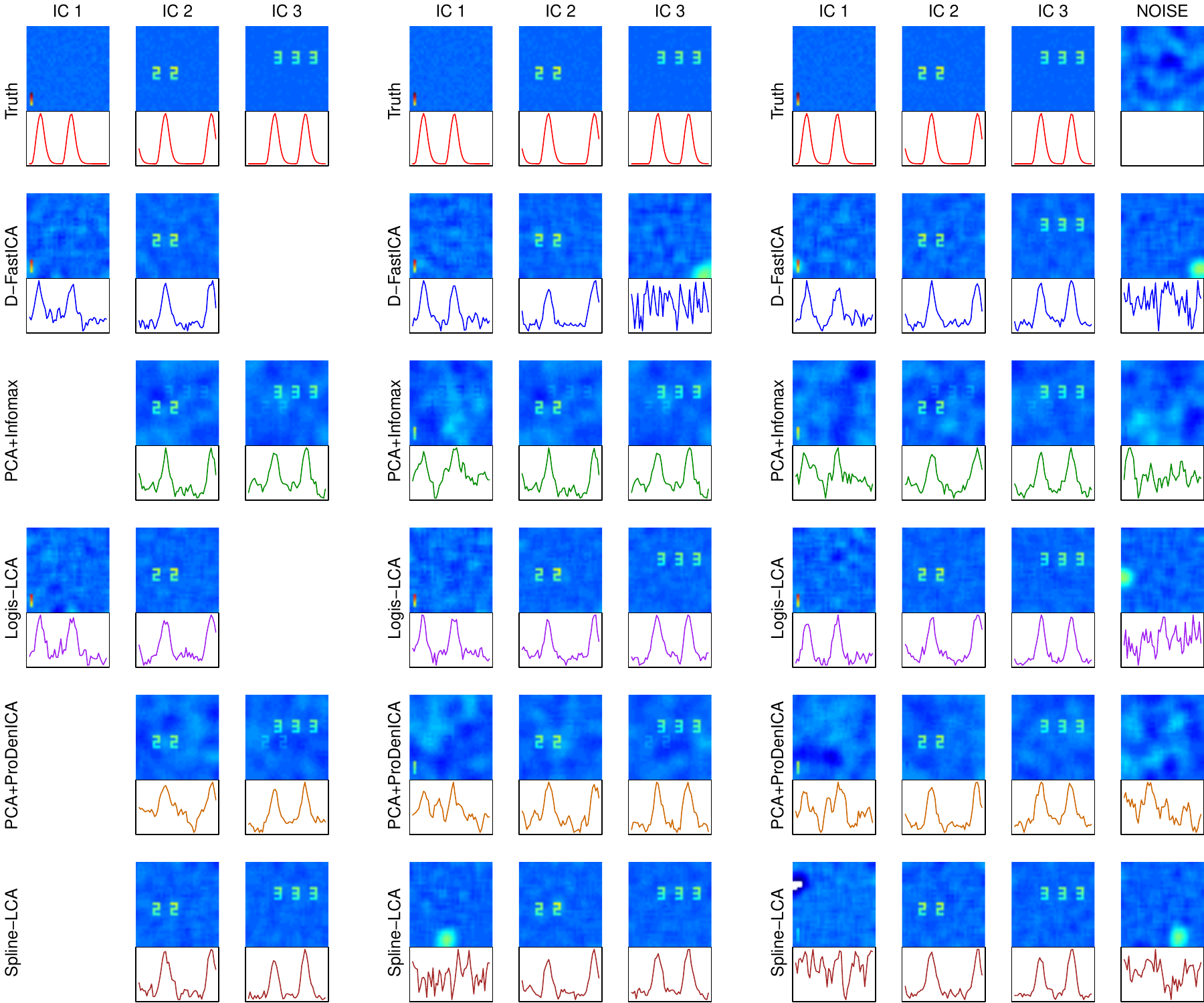}
\end{figure}

The permutation-invariant root mean squared errors for the components estimated from the spatio-temporal source simulations are much lower for Logis-LCA and Spline-LCA when the noise rank is $T-Q$  (Figure \ref{fig:MSE_simFMRI}). When the noise is rank-$T$, Logis-LCA performs best. Spline-LCA is excellent at finding two of the three components, but appears to sometimes find spurious components that were produced from the correlated noise when three or four components are estimated.
\begin{figure}
 \caption{Boxplots of $PRMSE$ for estimated columns of $\bS$ from simulations of spatial sources with temporal dependence and $Q = 3$ with $Q^* = 2,3,$ or $4$. `DF' = D-FastICA; `PI' = PCA+Infomax; `LL'= Logis-LCA; `PP' = PCA+ProDenICA; `SL' = Spline-LCA.}\label{fig:MSE_simFMRI}
   \begin{minipage}[b]{0.5\linewidth}
   \includegraphics[width=\textwidth]{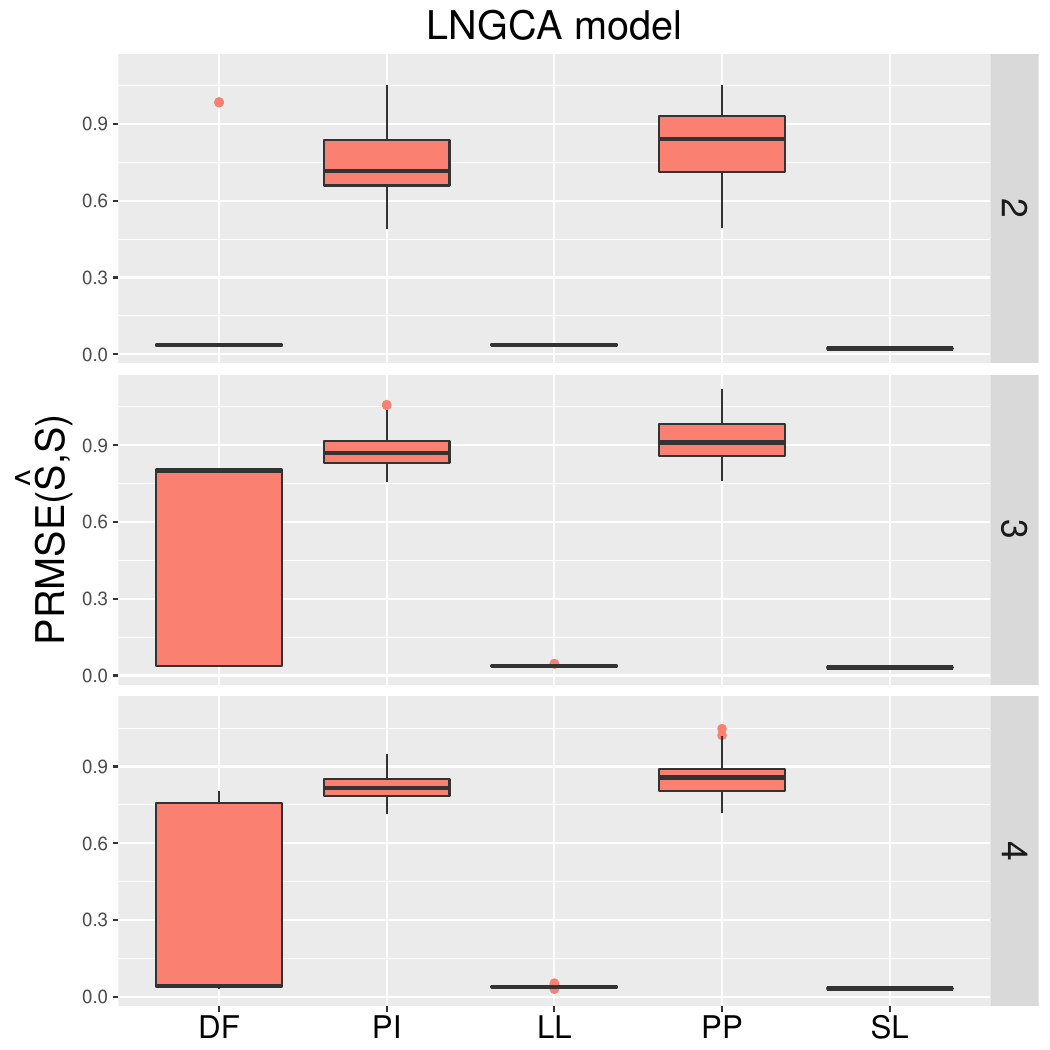}
   \end{minipage}
  \begin{minipage}[b]{0.5\linewidth}
  \includegraphics[width=\textwidth]{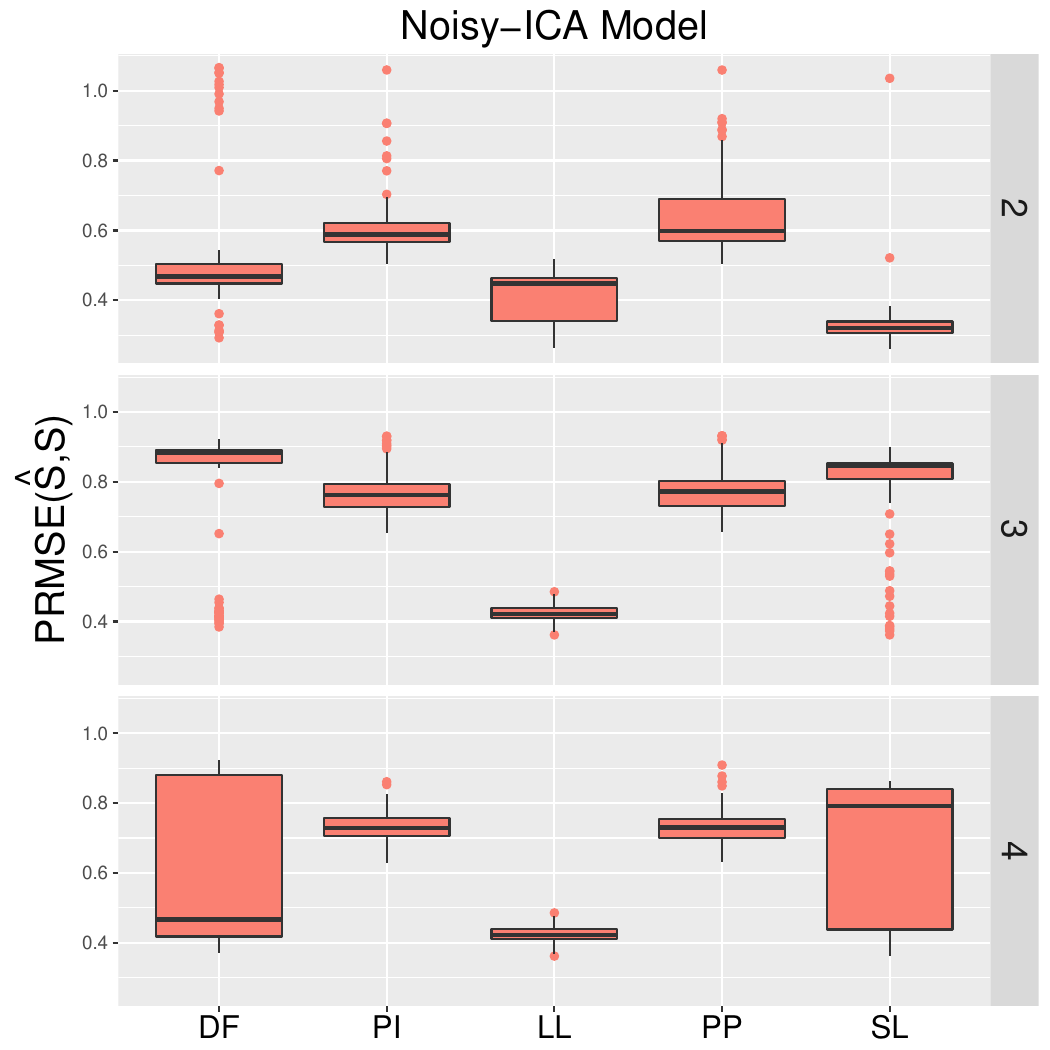}
  \end{minipage}
\end{figure}

\section{Supplemental materials for Section: Data Visualization and Dimension Reduction}\label{WS:Leaves}
 \cite{silva2013} generated covariates from photographs of leaf samples from thirty species (Figure \ref{fig:PlantSpecies}). Many of these covariates are highly correlated (Figure \ref{fig:LeafCorr}).
\begin{figure}
  \caption{Species 1-15 and 22-36 are included in the leaf dataset. Species 8 corresponds to \emph{Neurium oleander} (blue dots in Figure 4 and Supplemental Figures 4 and 5); species 31 and 34 correspond to \emph{Podocarpus sp}.\ and \emph{Pseudosasa japonica} (green dots in Figure~\ref{fig:LeafCorrPaper} and Supplemental Figures~\ref{fig:leafLogis} and~\ref{fig:leafSpline}). Figure from \cite{silva2013}.}\label{fig:PlantSpecies}
  \includegraphics[width=0.9\linewidth]{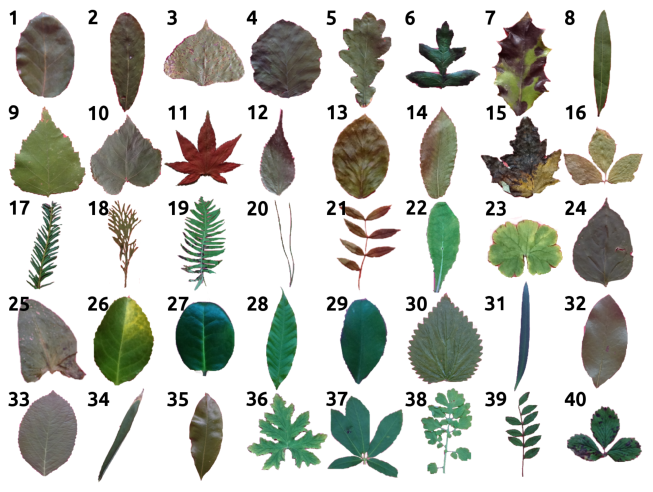}
\end{figure}

\begin{figure}
\center
 \caption{Correlation matrix of the variables in the leaf dataset: a) eccentricity, b) aspect ratio, c) elongation, d) solidity, e) stochastic convexity, f) isoperimetric factor, g) maximal indentation depth, h) lobedness, i) average intensity, j) average contrast, k) smoothness, l) third moment, m) uniformity, and n) entropy.}\label{fig:LeafCorr}
 \includegraphics[width=0.5\linewidth]{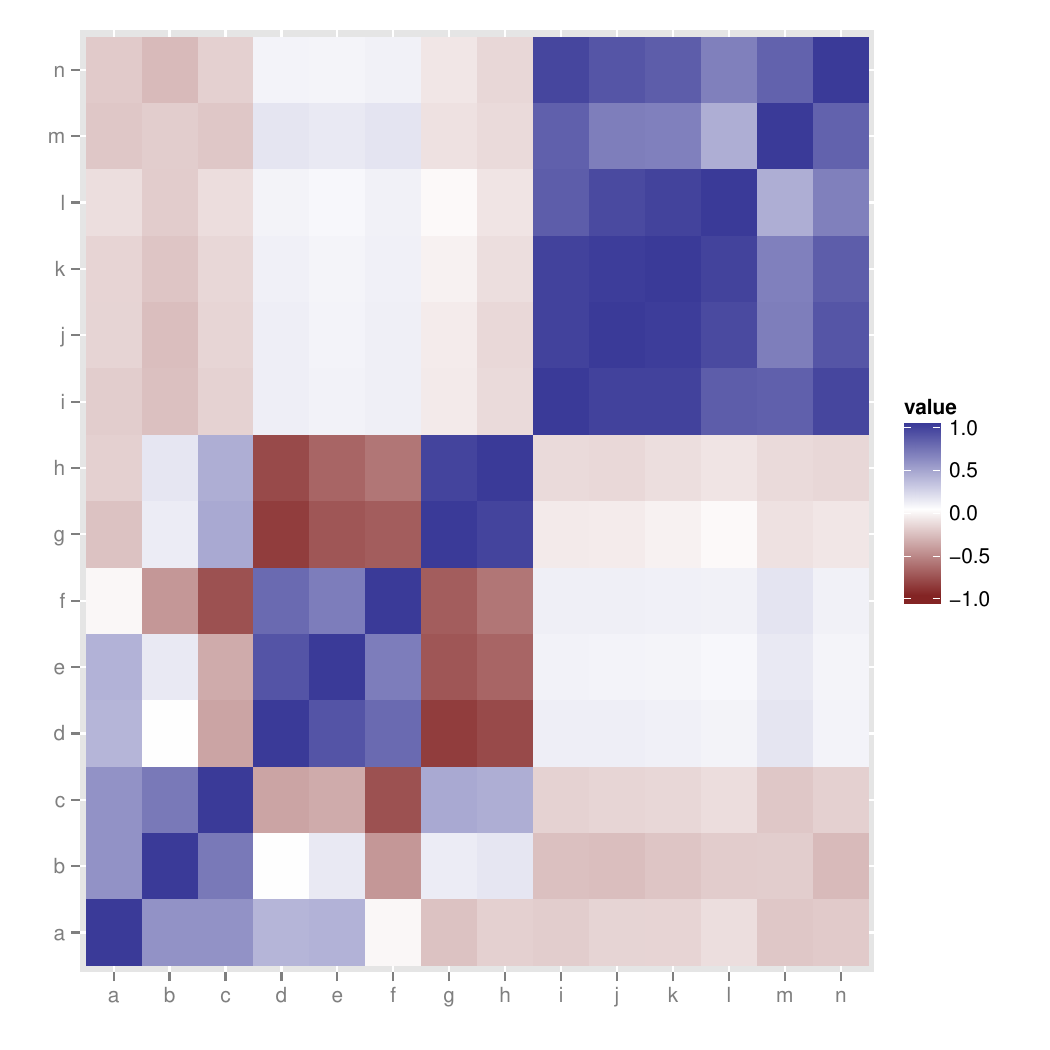}
\end{figure}

Logis-LCA and Spline-LCA reveal features in the data (Figures \ref{fig:leafLogis}, \ref{fig:leafSpline}), while PCA+Infomax and PCA+ProDenICA simply rotate the principal components. Additionally, when five components are estimated using the LCA methods, the first two components are nearly equivalent to the components obtained from $Q^*=2$. This is not the case with the PCA+ICA methods. Thus, the components in LCA appear less sensitive to the number of estimated components than the components from PCA+ICA methods.

\begin{figure}
\caption{Components in the leaf data from PCA+Infomax and Logis-LCA when two components were estimated and when five components were estimated (when five components were estimated, the two components with the highest marginal likelihood are plotted). The green dots correspond to \emph{Podocarpus sp}.\ and \emph{Pseudosasa japonica}; the blue dots correspond to \emph{Neurium oleander}; the red dots correspond to all other species.}\label{fig:leafLogis}
 \begin{minipage}[b]{0.5\linewidth}
  \includegraphics[width=\linewidth]{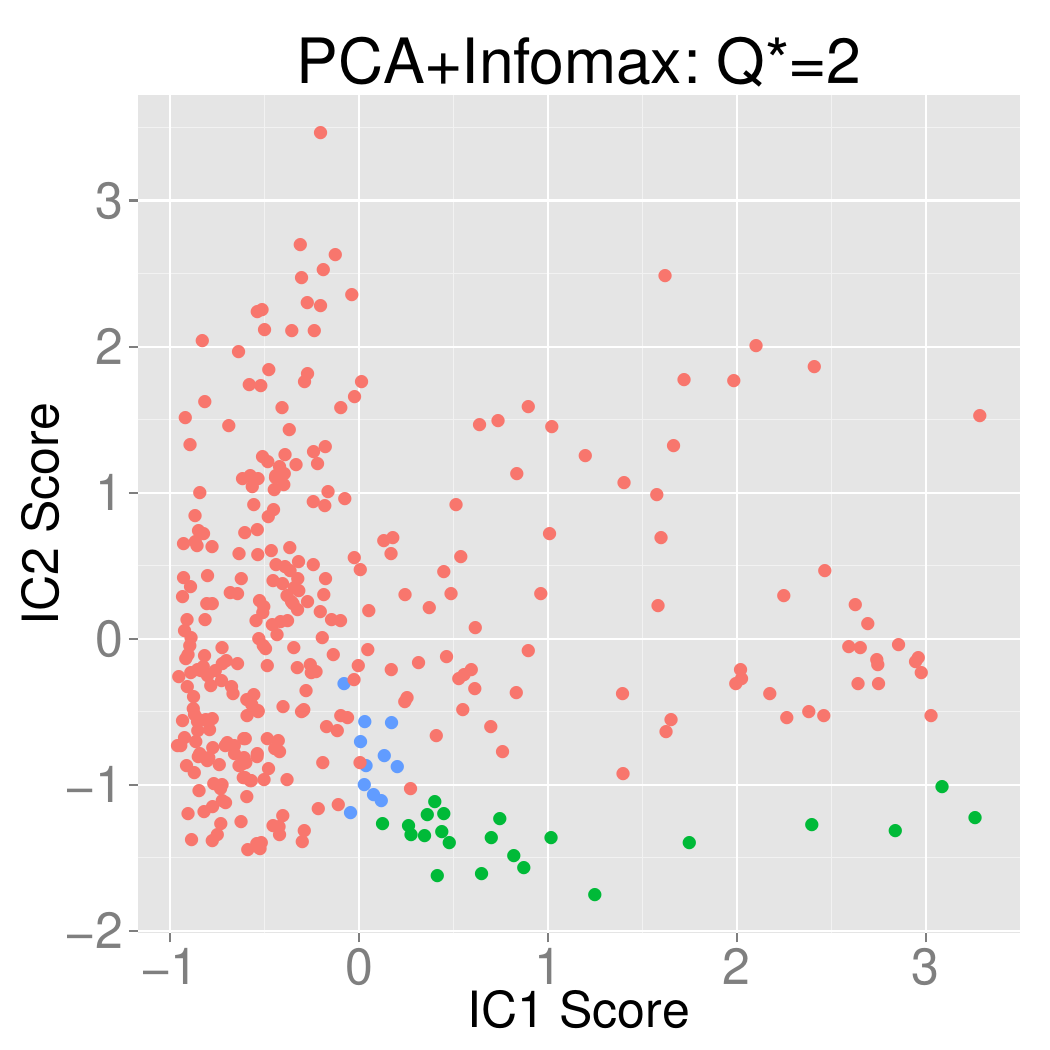}
  \includegraphics[width=\linewidth]{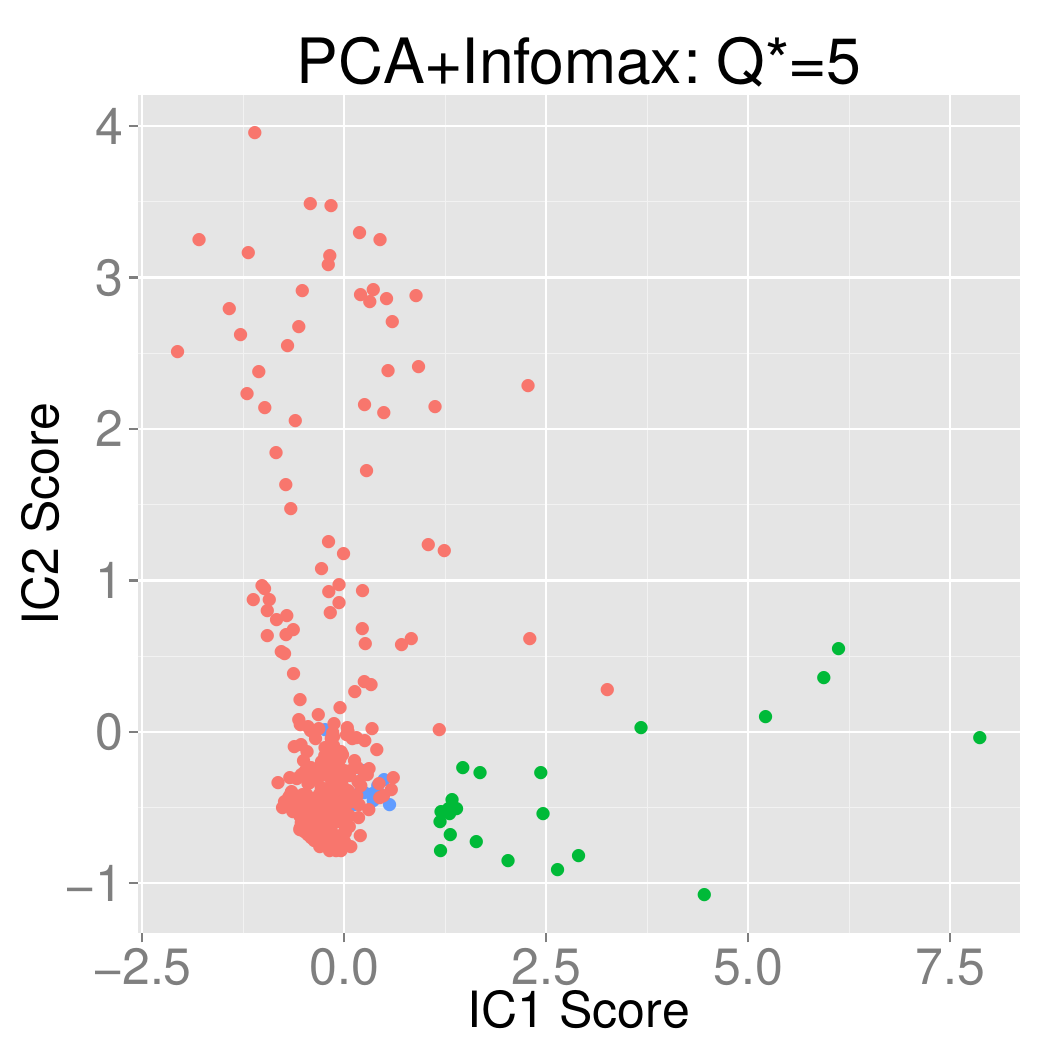}
 \end{minipage}
 \begin{minipage}[b]{0.5\linewidth}
  \includegraphics[width=\linewidth]{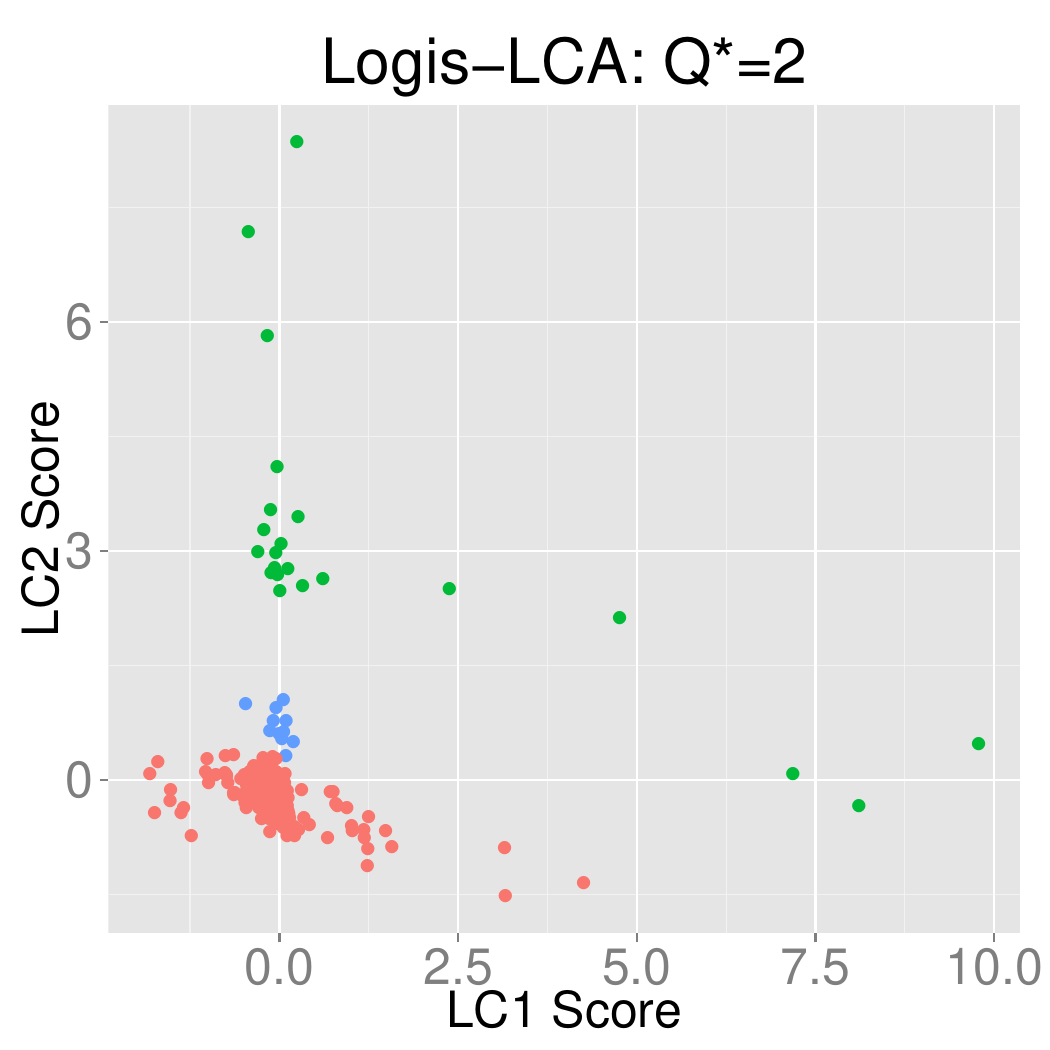}
  \includegraphics[width=\linewidth]{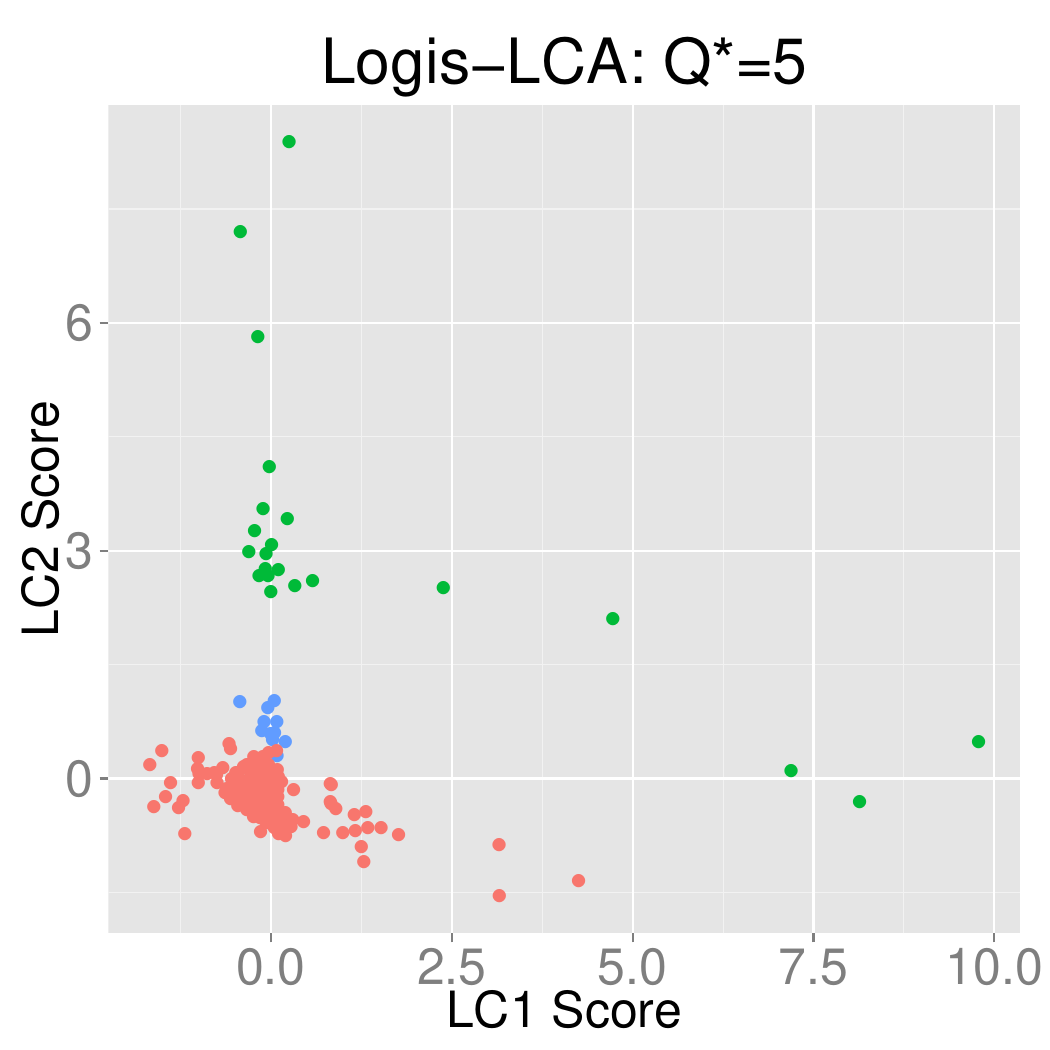}
 \end{minipage}
\end{figure}

\begin{figure}
\caption{Components in the leaf data from PCA-ProDenICA and Spline-LCA when two components were estimated and when five components were estimated (when five components were estimated, the two components with the highest marginal likelihood are plotted).  The green dots correspond to \emph{Podocarpus sp.} and \emph{Pseudosasa japonica}; the blue dots correspond to \emph{Neurium oleander}; the red dots correspond to all other species. The plots in the first row also appear in Figure~\ref{fig:LeafCorrPaper} of the main manuscript.}\label{fig:leafSpline}
 \begin{minipage}[b]{0.5\linewidth}
  \includegraphics[width=\linewidth]{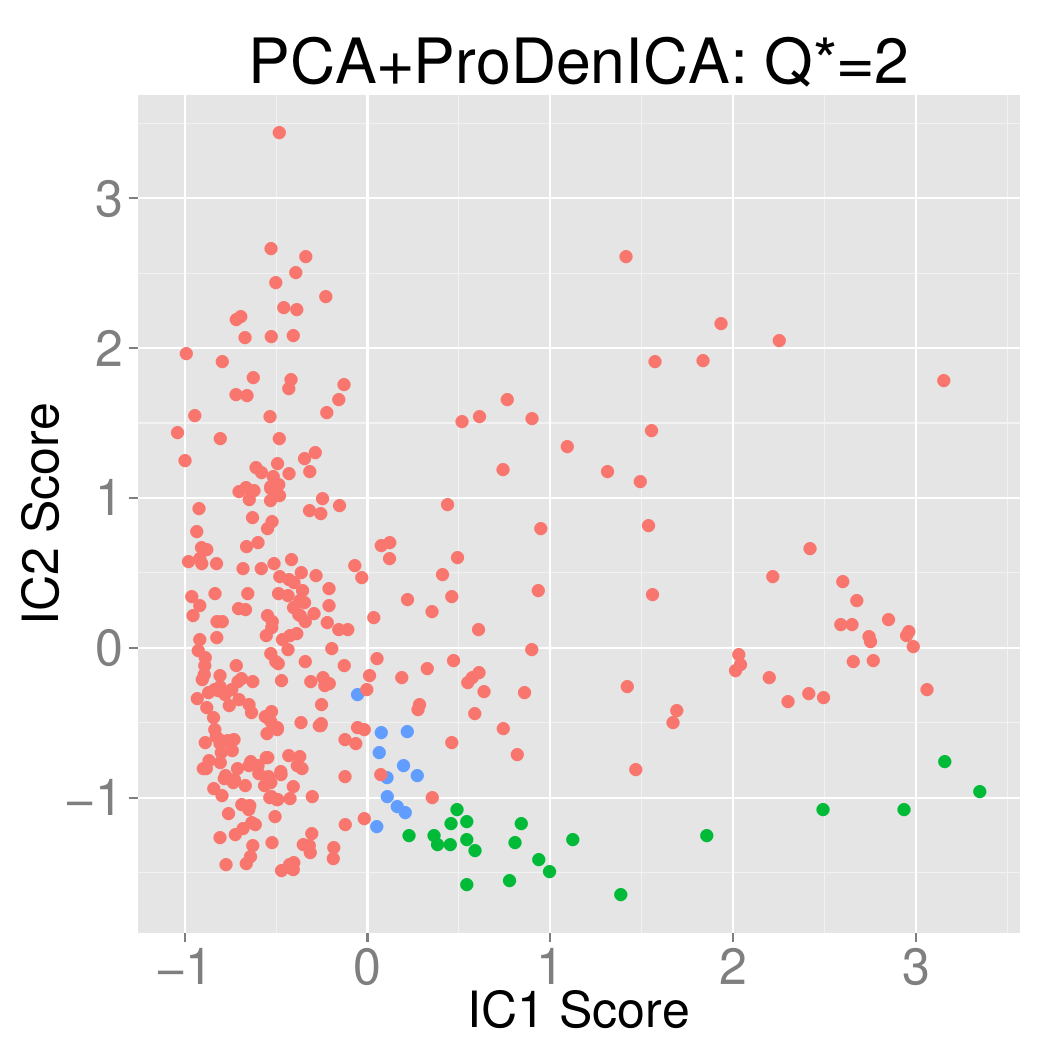}
  \includegraphics[width=\linewidth]{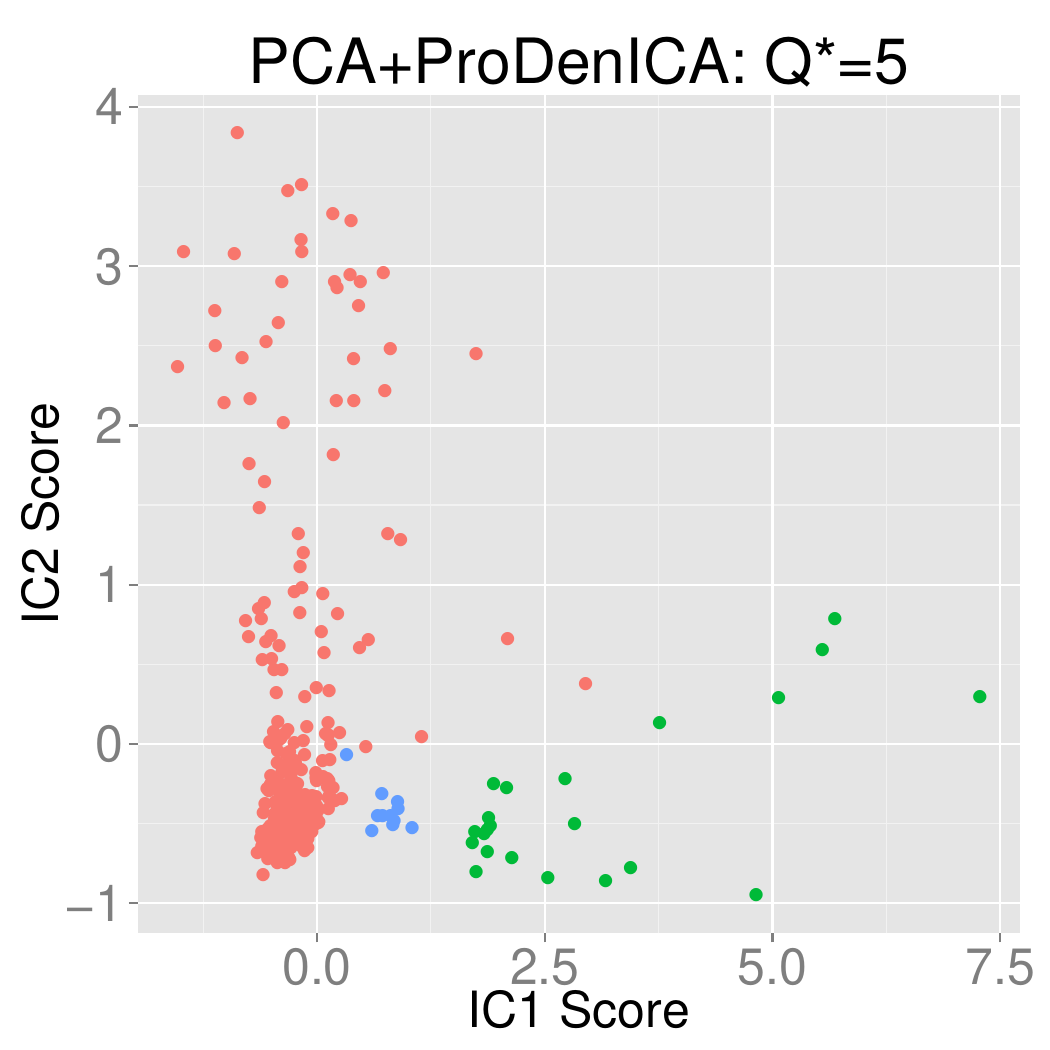}
 \end{minipage}
 \begin{minipage}[b]{0.5\linewidth}
  \includegraphics[width=\linewidth]{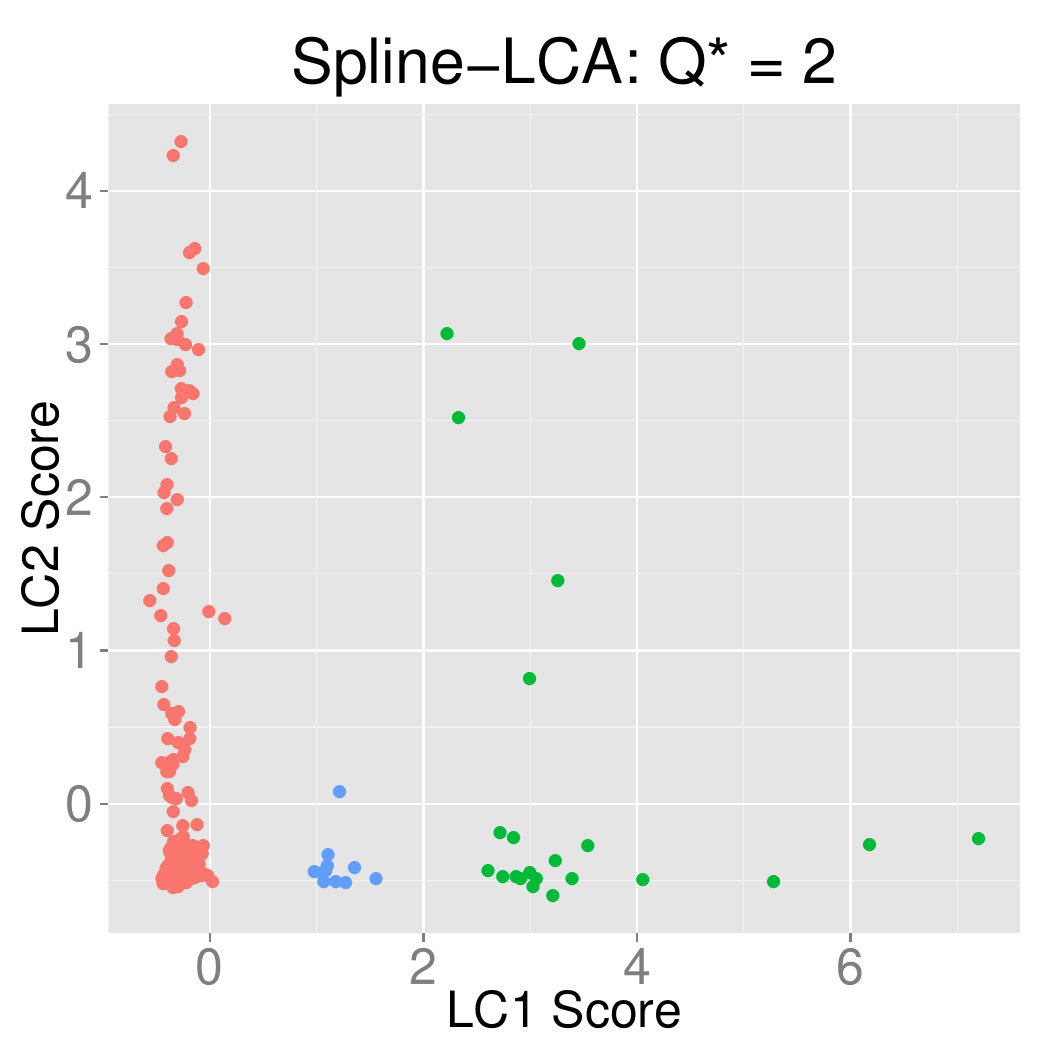}
  \includegraphics[width=\linewidth]{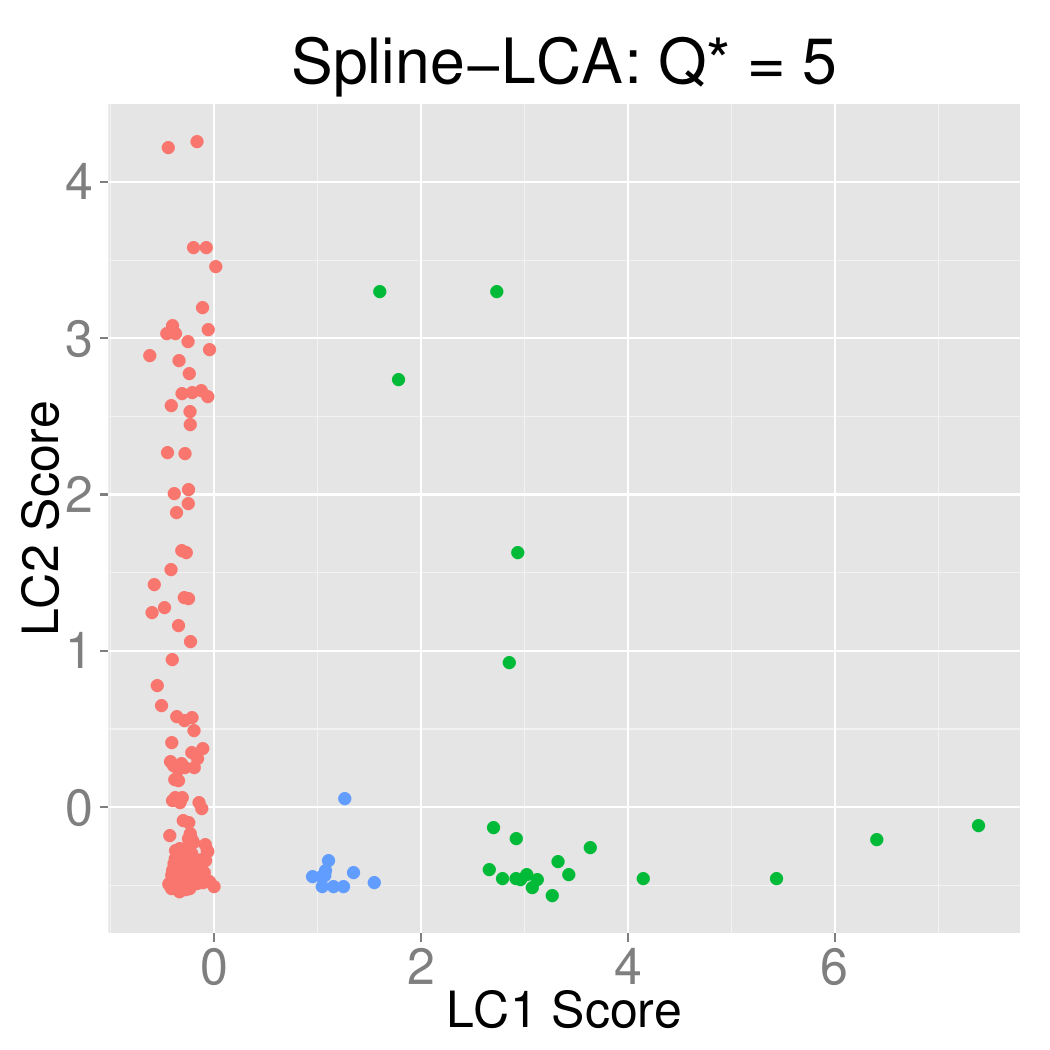}
 \end{minipage}
\end{figure}

\section{Supplemental materials for Section: Application to fMRI}\label{supp:sec:fMRI}
We analyzed task data from the theory of mind experiment in the HCP dataset. Theory of mind (ToM) refers to the ability of humans to infer the mental states of others. The experiment involved a mentalizing task in which shapes interacted in a goal-directed manner (e.g., a big triangle leading a little triangle out of a box) or according to some complex intentionality (e.g., a shape scaring another shape), and in which the random task involved shapes moving in random directions; for details see \cite{barch2013function}.

The application of ICA to fMRI usually assumes that voxels are iid (an exception for temporal ICA is \citealt{lee2011independent}). This assumption is often not made explicitly because ICA is usually derived from the perspective of maximizing non-Gaussianity. Since the objective function maximizing non-Gaussianity can also be derived from ML theory where the non-linear function is equivalent to the log likelihood (e.g., \citealt{hyvarinen2000independent}), summation of the non-linear function over voxels (e.g., Equation 12 in \citealt{beckmann2004probabilistic}) is mathematically equivalent to assuming the voxels are independent.
Despite the violation of model assumptions, ICA recovers simulated brain networks and their loadings \citep{beckmann2004probabilistic} and has proven useful in constructing models of functional connectivity that are consistent across subjects and image acquisition centers \citep{biswal2010toward}.

We analyzed the following subjects from the HCP 900-subject release dataset: 100206, 100307, 100408, 100610, 101006, 101107, 101309, 101410, 101915, 102008, and 103414.  Whole-brain data were acquired from two sessions with 274 volumes (i.e., brain images) each using gradient-echo EPI with multiband acceleration factor equal to eight and 2 x 2 x 2 mm voxels (repetition time (TR) = 720 ms; echo time (TE) = 33.1 ms; flip angle=$52^\circ$; field of view = 208 x 180 mm (readout  x phase-encoding); acquisition matrix = 104 x 90; slice thickness = 2.0 mm) in which the sessions differed in phase-encoding direction (right-left versus left-right). Only the first session was used in our analyses (the session with right-left phase encoding). Inspection revealed that the first two TRs contained BOLD signals that were higher than other time points. Consequently, we removed the first two TRs resulting in 272 time points for each voxel. After vectorization, the voxels were standardized across time to have mean zero and unit variance.

We initiated the algorithm from fifty-six matrices: from the first thirty columns of the FOBI (fourth-order blind identification) estimate of all components (an analytic solution that is fast to compute); twenty-seven semi-orthogonal matrices randomly generated in the principal subspace; and twenty-eight random semi-orthogonal matrices. We selected the estimate corresponding to the largest log likelihood as our estimate of the true argmax. The best estimate corresponded to one of the random matrices from the principal subspace for all subjects. Depending on initialization, the algorithm took between ten minutes and 3.75 hours on a 2666 MHz processor, where 3.75 hours represented initializations that reached the maximum number of iterations, which we conservatively chose to be equal to 300. We also completed an analogous PCA+ProDenICA with thirty components using the R package ProDenICA \citep{hastieR2010}, where one initialization was from the FOBI solution from the PCA-reduced dataset and fifty-five initializations were from random orthogonal matrices.  In PCA+ProDenICA, the best initialization was always from one of the fifty-five random orthogonal matrices. These results suggest that the FOBI solution was not ``close enough'' to the semiparametric solution to aid detection of the maximum in either Spline-LCA or PCA+ProDenICA.

%, possibly because temporal dependencies can be captured in the rows of the mixing matrix.
The presence of local maxima in LCA can increase computational expenses, and more initializations are required for larger values of $T$. Since the set of orthogonal matrices is non-convex, local optima are also a problem in PCA+ICA (e.g., \citealt{risk2014evaluation}). %Our approach is to initialize the LCA algorithms from random semi-orthogonal matrices generated from the eigenvectors of matrices with standard normal entries.
For fMRI data, fifty initializations appeared to be adequate when estimating thirty components with nearly three hundred time points (Figure~\ref{fig:InitValues_MDS}). In general, we found that Logis-LCA was less sensitive to initialization than Spline-LCA (results not shown). %It appears that the additional flexibility of Spline-LCA comes at the expense of increased detection of local maxima.
However, we favor Spline-LCA because it can more accurately model source densities. %However, sub-Gaussian components appear to be uncommon in fMRI (sparse images are super-Gaussian). Future research should examine whether Spline-LCA offers advantages over Logis-LCA in fMRI. Additionally, developing algorithms to more efficiently address local optima is an avenue for future research.

%Analyses were initially conducted by concatenating sessions one and two, but subsequent inspection suggested that patterns of network loadings differed greatly between the two sessions. Consequently, only the first session was included.

\begin{figure}
 \includegraphics[width=0.75\textwidth]{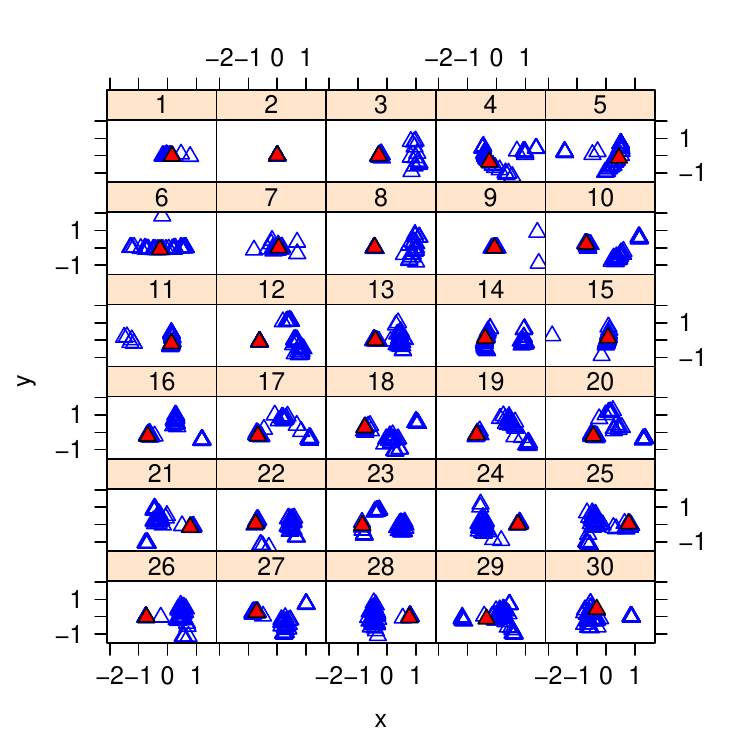}
 \caption{Multidimensional scaling of $||\hatbS_j^{(k)} - \hatbS_j^{(\ell)} ||_F$ for components $j=1,\dots,30$ and initializations $k \ne \ell \in \{1,\dots,56\}$. The coordinates corresponding to the initialization with the highest likelihood are depicted by solid red triangles. In all instances, the red triangle appears in a cluster of other triangles, indicating agreement between a subset of initializations.}\label{fig:InitValues_MDS}
\end{figure}

For subject 103414, we examined the effect of initialization in detail. Following \cite{risk2014evaluation}, we assessed the reliability of individual components by matching components from all other initializations to the components corresponding to the ${\margmax}$ using the modified Hungarian algorithm. We then created dissimilarity matrices for each component based on the MSE and visualized basins of attraction using multidimensional scaling. Generally, there were at least two basins of attraction corresponding to initializations from the principal subspace and initializations from the entire column space (Supplemental Figure \ref{fig:InitValues_MDS}). Components one, two, and nine were relatively robust to initialization and contained only one (main) basin of attraction.

We examined the correlation between the loadings (columns of $\hatbM_\bS$) and the mentalizing and random tasks. The mentalizing and random task covariates were generated by convolving each task's onsets and durations with the canonical HRF in SPM8 \citep{ashburner2004human}. In all subjects, the first component, i.e., the one with the highest likelihood, was highly correlated with the mentalizing and random tasks (e.g., Figure \ref{fig:HCP_networks}). The most positive values of this component are located in the gray matter, which indicates brain activity. Areas of Brodmann Area 19 in the visual cortex appear activated. This is an area associated with shape recognition and attention, and thus it makes sense that the movies based on moving shapes engaged this area. The same component was found using PCA+ProDenICA. For all subjects, the correlation of the matched PCA+ProDenICA component with the first Spline-LCA component was at least 0.98. Note however that this component does not distinguish between the mentalizing and random tasks. Moreover, the temporal parietal junction (TPJ) is an area often found in ToM studies \citep{castelli2000movement} (the crosshairs in Figure \ref{fig:HCP_networks} are located near the TPJ) but is not activated in this component, suggesting there exists additional signal in other components.

\begin{figure}[H]
 \caption{Selected components estimated from the HCP ToM data using Spline-LCA. The first row depicts a task-activated component that was highly correlated with the mentalizing (green) and random (blue) tasks (MNI coordinates: -50,-56,18); a similar component was found using PCA+ProDenICA (not depicted). The second row appears to be an artifact not found by PCA+ProDenICA (MNI: 0,-50,0).}\label{fig:HCP_networks}
%thresholded $|s_{v1}|\ge 2$)
\begin{minipage}[b]{0.31\linewidth}
 \includegraphics[width=\textwidth]{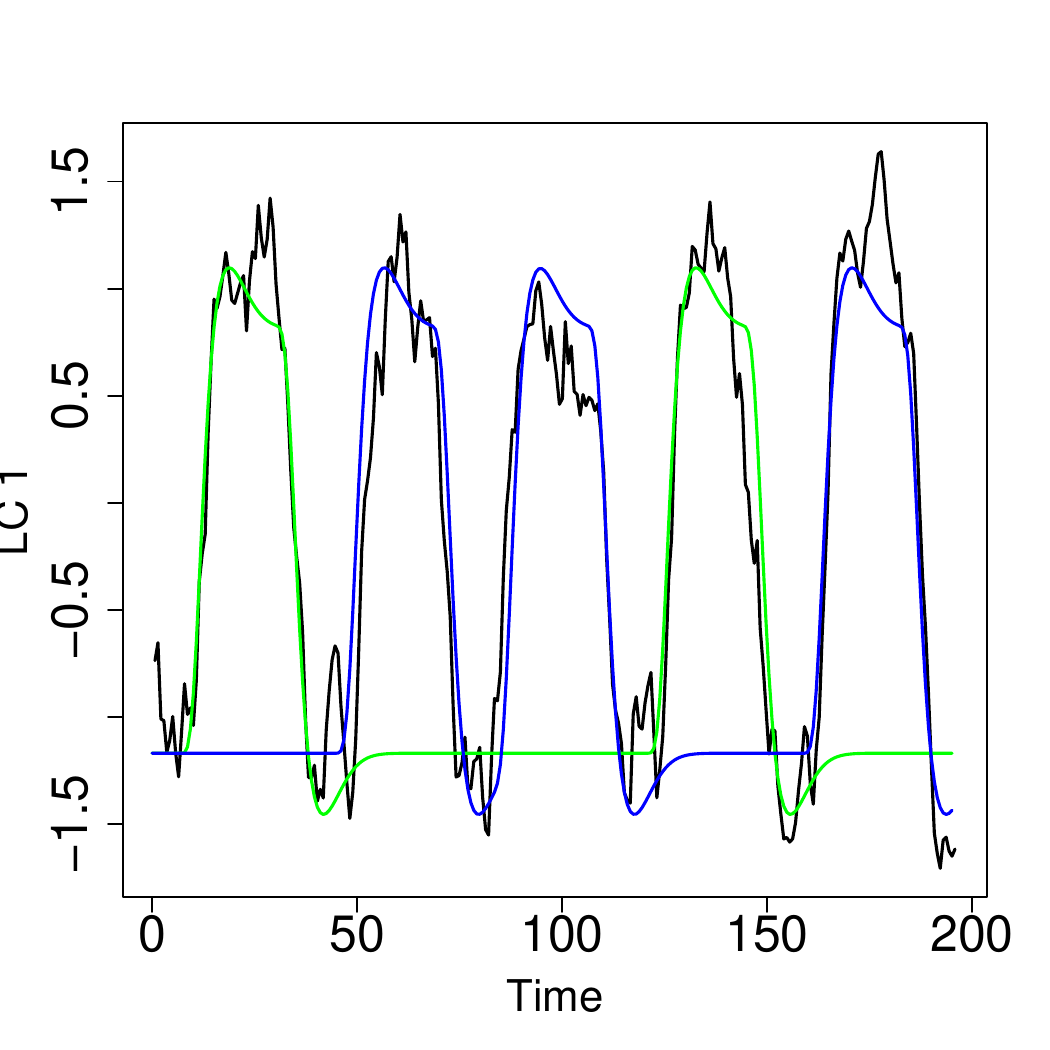}
\end{minipage}
\begin{minipage}[b]{0.62\linewidth}
 \includegraphics[width=\textwidth]{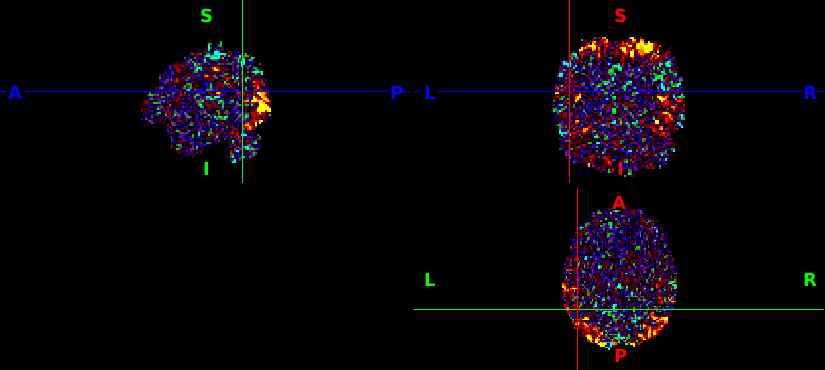}
\end{minipage}

 \begin{minipage}[b]{0.31\linewidth}
\includegraphics[width=\textwidth]{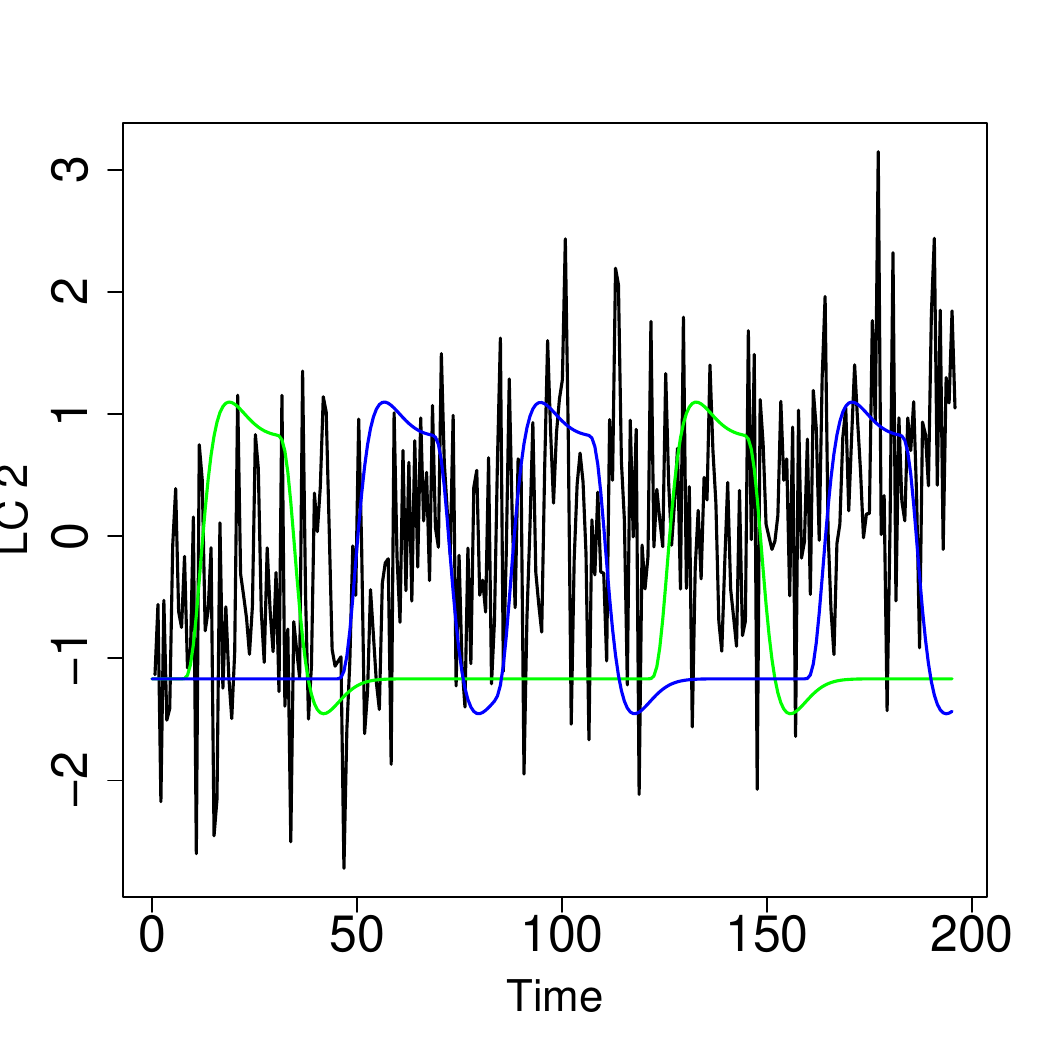}
\end{minipage}
 \begin{minipage}[b]{0.62\linewidth}
  \includegraphics[width=\textwidth]{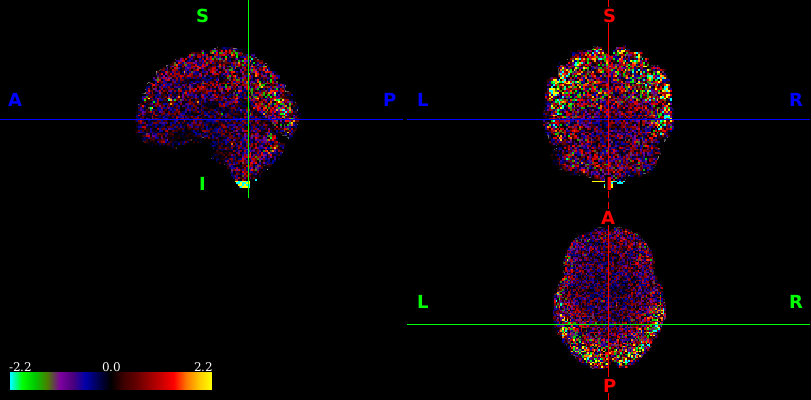}
 \end{minipage}
\end{figure}

Voxels were highly activated in the brainstem and the component's time course was correlated with three of the motion parameters from the rigid-body alignment  ($r=0.32$, 0.32, and 0.42 for the x-transformation, x-rotation, and z-rotation parameters, respectively). This may be related to a gradual relaxation of the neck or spine over the course of the subject's session. Additionally, there was a positive correlation with time ($r=0.44$), which could also be related to scanner drift.

LCA also identified a type of artifact that did not seem to be found in PCA+ProDenICA. Some components had alternating bands of positive and negative values, in particular in axial slices through orbitofrontal regions (Figure \ref{fig:LC14}). The patterns of activation ignored gray and white matter tissue boundaries, which is evidence of an artifact. This type of pattern is described as an ``MRI acquisition/reconstruction related artifact'' in \cite{salimi2014automatic}.

 \begin{figure}[H]
\caption{Artifact (component 14) identified using Spline-LCA (top) and the matched component from PCA+ProDenICA (bottom; correlation = 0.08) in subject 100307. Thresholded at $|s_{v,14}|>1.75$.}\label{fig:LC14}
 \includegraphics[width=\textwidth]{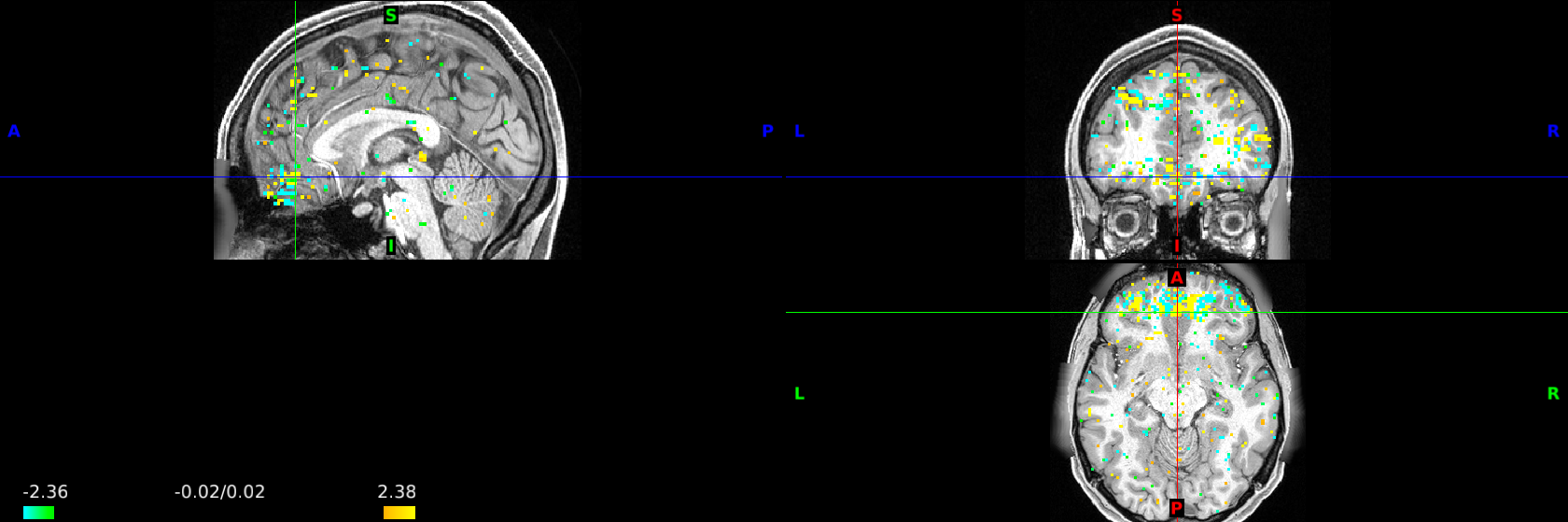}

 \includegraphics[width=\textwidth]{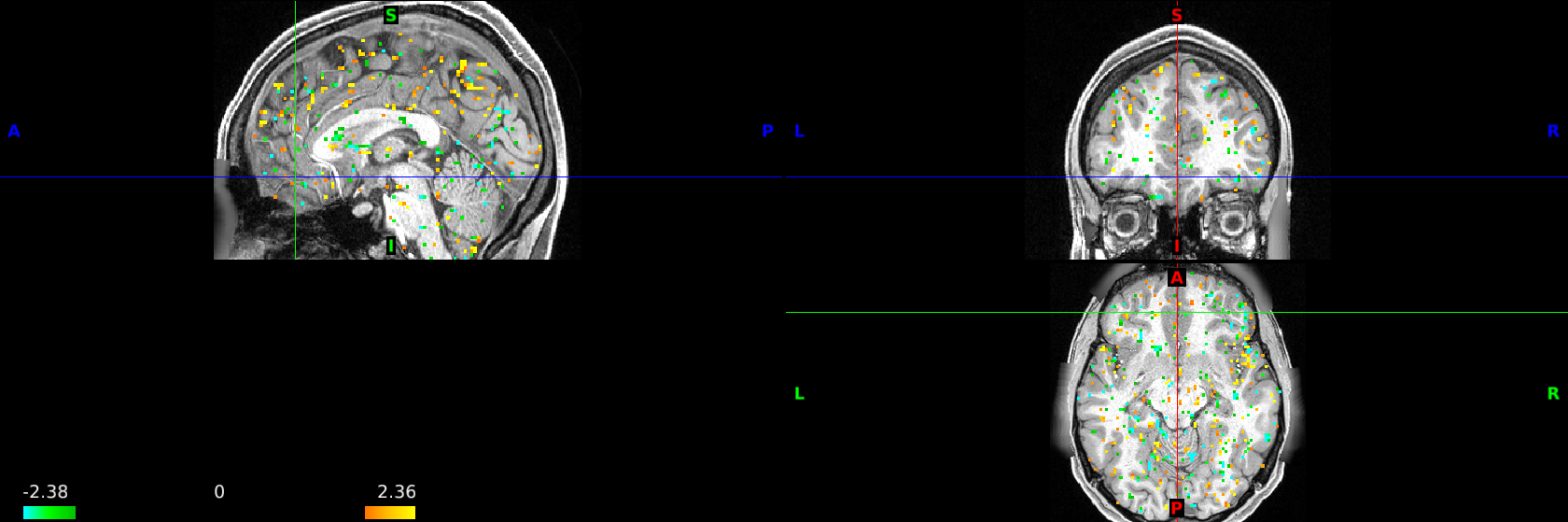}
\end{figure}

Removing artifacts from fMRI detected using PCA+ICA is a popular tool that can increase detection in subsequent mixed-modeling of voxel activation \citep{pruim2015ica}. Our results suggests that LCA may improve artifact detection.

 \singlespace

\bibliographystyle{apalike}
%\bibliography{../../MyBibtex/MasterBibliography}
\bibliography{LCAsupp_v2.bbl}

\begin{thebibliography}{}

\bibitem[Allassonniere and Younes, 2012]{allassonniere2012stochastic}
Allassonniere, S. and Younes, L. (2012).
\newblock A stochastic algorithm for probabilistic independent component
  analysis.
\newblock {\em The Annals of Applied Statistics}, 6(1):125--160.

\bibitem[Amato et~al., 2010]{amato2010noisy}
Amato, U., Antoniadis, A., Samarov, A., and Tsybakov, A. (2010).
\newblock Noisy independent factor analysis model for density estimation and
  classification.
\newblock {\em Electronic Journal of Statistics}, 4:707--736.

\bibitem[Attias, 1999]{attias1999independent}
Attias, H. (1999).
\newblock Independent factor analysis.
\newblock {\em Neural computation}, 11(4):803--851.

\bibitem[Bach and Jordan, 2003]{bach2003kernel}
Bach, F.~R. and Jordan, M.~I. (2003).
\newblock Kernel independent component analysis.
\newblock {\em The Journal of Machine Learning Research}, 3:1--48.

\bibitem[Bartlett et~al., 2002]{bartlett2002face}
Bartlett, M.~S., Movellan, J.~R., and Sejnowski, T.~J. (2002).
\newblock Face recognition by independent component analysis.
\newblock {\em IEEE Transactions on Neural Networks}, 13(6):1450--1464.

\bibitem[Beckmann, 2012]{beckmann2012modelling}
Beckmann, C.~F. (2012).
\newblock Modelling with independent components.
\newblock {\em NeuroImage}, 62(2):891--901.

\bibitem[Beckmann and Smith, 2004]{beckmann2004probabilistic}
Beckmann, C.~F. and Smith, S.~M. (2004).
\newblock Probabilistic independent component analysis for functional magnetic
  resonance imaging.
\newblock {\em IEEE Transactions on Medical Imaging}, 23(2):137--152.

\bibitem[Bell and Sejnowski, 1995]{bell1995information}
Bell, A.~J. and Sejnowski, T.~J. (1995).
\newblock An information-maximization approach to blind separation and blind
  deconvolution.
\newblock {\em Neural computation}, 7(6):1129--1159.

\bibitem[Blanchard et~al., 2006]{blanchard2006search}
Blanchard, G., Kawanabe, M., Sugiyama, M., Spokoiny, V., and M{\"u}ller, K.-R.
  (2006).
\newblock In search of non-{Gaussian} components of a high-dimensional
  distribution.
\newblock {\em The Journal of Machine Learning Research}, 7:247--282.

\bibitem[Calhoun and Adali, 2006]{calhoun2006unmixing}
Calhoun, V.~D. and Adali, T. (2006).
\newblock Unmixing {fMRI} with independent component analysis.
\newblock {\em Engineering in Medicine and Biology Magazine, IEEE},
  25(2):79--90.

\bibitem[Cardoso and Souloumiac, 1993]{cardoso1993blind}
Cardoso, J.~F. and Souloumiac, A. (1993).
\newblock Blind beamforming for non-{Gaussian} signals.
\newblock In {\em Radar and Signal Processing, IEEE Proceedings F}, volume 140,
  pages 362--370.

\bibitem[Chen and Bickel, 2006]{chen2006efficient}
Chen, A. and Bickel, P.~J. (2006).
\newblock Efficient independent component analysis.
\newblock {\em The Annals of Statistics}, 34(6):2825--2855.

\bibitem[Correa et~al., 2007]{correa2007performance}
Correa, N., Adali, T., and Calhoun, V.~D. (2007).
\newblock Performance of blind source separation algorithms for {fMRI} analysis
  using a group {ICA} method.
\newblock {\em Magnetic Resonance Imaging}, 25(5):684--694.

\bibitem[Cover and Thomas, 2006]{cover2006elements}
Cover, T.~M. and Thomas, J.~A. (2006).
\newblock {\em Elements of Information Theory}.
\newblock John Wiley \& Sons, New Jersey.

\bibitem[Eloyan and Ghosh, 2013]{Eloyan2012}
Eloyan, A. and Ghosh, S.~K. (2013).
\newblock A semiparametric approach to source separation using independent
  component analysis.
\newblock {\em Computational Statistics and Data Analysis}, 58:383 -- 396.

\bibitem[Friston et~al., 1996]{friston1996movement}
Friston, K.~J., Williams, S., Howard, R., Frackowiak, R.~S., and Turner, R.
  (1996).
\newblock Movement-related effects in {fMRI} time-series.
\newblock {\em Magnetic Resonance in Medicine}, 35(3):346--355.

\bibitem[Glasser et~al., 2013]{glasser2013minimal}
Glasser, M.~F., Sotiropoulos, S.~N., Wilson, J.~A., Coalson, T.~S., Fischl, B.,
  Andersson, J.~L., Xu, J., Jbabdi, S., Webster, M., Polimeni, J.~R., et~al.
  (2013).
\newblock The minimal preprocessing pipelines for the {Human Connectome
  Project}.
\newblock {\em Neuro{I}mage}, 80:105--124.

\bibitem[Green et~al., 2002]{green2002pca}
Green, C.~G., Nandy, R.~R., and Cordes, D. (2002).
\newblock {PCA}-preprocessing of {fMRI} data adversely affects the results of
  {ICA}.
\newblock In {\em {Proceedings of International Society of Magnetic Resonance
  in Medicine}}, page~10.

\bibitem[Griffanti et~al., 2014]{griffanti2014ica}
Griffanti, L., Salimi-Khorshidi, G., Beckmann, C.~F., Auerbach, E.~J., Douaud,
  G., Sexton, C.~E., Zsoldos, E., Ebmeier, K.~P., Filippini, N., Mackay, C.~E.,
  et~al. (2014).
\newblock {ICA}-based artefact removal and accelerated {fMRI} acquisition for
  improved resting state network imaging.
\newblock {\em NeuroImage}, 95:232--247.

\bibitem[Guo and Tang, 2013]{guo2013hierarchical}
Guo, Y. and Tang, L. (2013).
\newblock A hierarchical model for probabilistic independent component analysis
  of multi-subject {fMRI} studies.
\newblock {\em Biometrics}, 69(4):970--981.

\bibitem[Hastie, 2013]{hastiegam}
Hastie, T. (2013).
\newblock {\em {GAM}: Generalized Additive Models}.
\newblock R package version 1.08.

\bibitem[Hastie and Tibshirani, 2003]{hastie2002independent}
Hastie, T. and Tibshirani, R. (2003).
\newblock Independent components analysis through product density estimation.
\newblock {\em Advances in Neural Information Processing Systems}, 15:649--656.

\bibitem[Hastie and Tibshirani, 2010]{hastieR2010}
Hastie, T. and Tibshirani, R. (2010).
\newblock {\em ProDenICA: Product Density Estimation for ICA using tilted
  Gaussian density estimates}.
\newblock R package version 1.0.

\bibitem[Hastie et~al., 2009]{hastie2009elements}
Hastie, T., Tibshirani, R., and Friedman, J. (2009).
\newblock {\em The Elements of Statistical Learning}.
\newblock Springer.

\bibitem[Huber, 1985]{huber1985projection}
Huber, P.~J. (1985).
\newblock Projection pursuit.
\newblock {\em The Annals of Statistics}, pages 435--475.

\bibitem[Hyvarinen, 1999]{hyvarinen1999fast}
Hyvarinen, A. (1999).
\newblock Fast and robust fixed-point algorithms for independent component
  analysis.
\newblock {\em IEEE Transactions on Neural Networks}, 10(3):626--634.

\bibitem[Hyv{\"a}rinen et~al., 2001]{hyvarinen2001independent}
Hyv{\"a}rinen, A., Karhunen, J., and Oja, E. (2001).
\newblock {\em Independent component analysis}.
\newblock Wiley-Interscience.

\bibitem[Hyv{\"a}rinen and Oja, 1998]{Hyvaerinen1998}
Hyv{\"a}rinen, A. and Oja, E. (1998).
\newblock Independent component analysis by general nonlinear {Hebbian-like}
  learning rules.
\newblock {\em Signal Processing}, 64(3):301--313.

\bibitem[Hyv{\"a}rinen and Oja, 2000]{hyvarinen2000independent}
Hyv{\"a}rinen, A. and Oja, E. (2000).
\newblock Independent component analysis: algorithms and applications.
\newblock {\em Neural Networks}, 13(4-5):411--430.

\bibitem[Ilmonen et~al., 2010]{ilmonen2010new}
Ilmonen, P., Nordhausen, K., Oja, H., and Ollila, E. (2010).
\newblock {A new performance index for ICA: properties, computation and
  asymptotic analysis}.
\newblock {\em Latent Variable Analysis and Signal Separation}, pages 229--236.

\bibitem[Kagan et~al., 1973]{kagan1973characterization}
Kagan, A.~M., Rao, C.~R., and Linnik, Y.~V. (1973).
\newblock {\em Characterization Problems in Mathematical Statistics}.
\newblock Wiley.

\bibitem[Kawanabe et~al., 2007]{kawanabe2007new}
Kawanabe, M., Sugiyama, M., Blanchard, G., and M{\"u}ller, K. (2007).
\newblock A new algorithm of non-{Gaussian} component analysis with radial
  kernel functions.
\newblock {\em Annals of the Institute of Statistical Mathematics},
  59(1):57--75.

\bibitem[Lee et~al., 1999]{lee1999independent}
Lee, T.~W., Girolami, M., and Sejnowski, T.~J. (1999).
\newblock Independent component analysis using an extended infomax algorithm
  for mixed subgaussian and supergaussian sources.
\newblock {\em Neural Computation}, 11(2):417--441.

\bibitem[Marchini et~al., 2010]{Marchini:2010lr}
Marchini, J.~L., Heaton, C., and Ripley, B.~D. (2010).
\newblock {\em {FastICA}: FastICA Algorithms to perform ICA and Projection
  Pursuit}.
\newblock R package version 1.1-13.

\bibitem[Matteson and Tsay, 2016]{mattesontsay2012}
Matteson, D.~S. and Tsay, R.~S. (2016).
\newblock Independent component analysis via distance covariance.
\newblock {\em Journal of the American Statistical Association}, in press.

\bibitem[Miettinen et~al., 2014]{miettinen2014deflation}
Miettinen, J., Nordhausen, K., Oja, H., and Taskinen, S. (2014).
\newblock Deflation-based {FastICA} with adaptive choices of nonlinearities.
\newblock {\em IEEE Transactions on Signal Processing}, 62(21):5716--5724.

\bibitem[Miettinen et~al., 2017]{miettinen2015squared}
Miettinen, J., Nordhausen, K., Oja, H., Taskinen, S., and Virta, J. (2017).
\newblock The squared symmetric fastica estimator.
\newblock {\em Signal Processing}, 131:402--411.

\bibitem[Miettinen et~al., 2015]{miettinen2015fourth}
Miettinen, J., Taskinen, S., Nordhausen, K., Oja, H., et~al. (2015).
\newblock Fourth moments and independent component analysis.
\newblock {\em Statistical science}, 30(3):372--390.

\bibitem[Nordhausen et~al., 2011]{nordhausen2011deflation}
Nordhausen, K., Ilmonen, P., Mandal, A., Oja, H., and Ollila, E. (2011).
\newblock Deflation-based fastica reloaded.
\newblock In {\em Signal Processing Conference, 2011 19th European}, pages
  1854--1858. IEEE.

\bibitem[Nordhausen et~al., 2016]{nordhausen2016asymptotic}
Nordhausen, K., Oja, H., and Tyler, D.~E. (2016).
\newblock Asymptotic and bootstrap tests for subspace dimension.
\newblock {\em arXiv preprint arXiv:1611.04908}.

\bibitem[Pruim et~al., 2015]{pruim2015ica}
Pruim, R.~H., Mennes, M., van Rooij, D., Llera, A., Buitelaar, J.~K., and
  Beckmann, C.~F. (2015).
\newblock {ICA-AROMA}: a robust {ICA}-based strategy for removing motion
  artifacts from {fMRI} data.
\newblock {\em NeuroImage}, 112:267--277.

\bibitem[Risk et~al., 2014]{risk2014evaluation}
Risk, B.~B., Matteson, D.~S., Ruppert, D., Eloyan, A., and Caffo, B.~S. (2014).
\newblock An evaluation of independent component analyses with an application
  to resting-state {fMRI}.
\newblock {\em Biometrics}, 70(1):224--236.

\bibitem[Salimi-Khorshidi et~al., 2014]{salimi2014automatic}
Salimi-Khorshidi, G., Douaud, G., Beckmann, C.~F., Glasser, M.~F., Griffanti,
  L., and Smith, S.~M. (2014).
\newblock Automatic denoising of functional {MRI} data: combining independent
  component analysis and hierarchical fusion of classifiers.
\newblock {\em Neuroimage}, 90:449--468.

\bibitem[Samworth and Yuan, 2012]{samworth2012independent}
Samworth, R.~J. and Yuan, M. (2012).
\newblock Independent component analysis via nonparametric maximum likelihood
  estimation.
\newblock {\em The Annals of Statistics}, 40(6):2973--3002.

\bibitem[Shi and Guo, 2016]{shi2016modeling}
Shi, R. and Guo, Y. (2016).
\newblock Investigating differences in brain functional networks using
  hierarchical covariate-adjusted independent component analysis.
\newblock {\em Annals of Applied Statistics}, in press.

\bibitem[Silva et~al., 2013]{silva2013}
Silva, P.~F., Marcal, A.~R., and da~Silva, R. M.~A. (2013).
\newblock Evaluation of features for leaf discrimination.
\newblock {\em Springer Lecture Notes in Computer Science}, Vol. 7950(197-204).

\bibitem[St{\"o}gbauer et~al., 2004]{stogbauer2004least}
St{\"o}gbauer, H., Kraskov, A., Astakhov, S.~A., and Grassberger, P. (2004).
\newblock Least-dependent-component analysis based on mutual information.
\newblock {\em Physical Review E}, 70(6):066123.

\bibitem[Tipping and Bishop, 1999]{tipping1999probabilistic}
Tipping, M.~E. and Bishop, C.~M. (1999).
\newblock Probabilistic principal component analysis.
\newblock {\em Journal of the Royal Statistical Society: Series B (Statistical
  Methodology)}, 61(3):611--622.

\bibitem[Virta et~al., 2015]{virta2015joint}
Virta, J., Nordhausen, K., and Oja, H. (2015).
\newblock Joint use of third and fourth cumulants in independent component
  analysis.
\newblock {\em arXiv preprint arXiv:1505.02613}.

\bibitem[Virta et~al., 2016]{virta2016projection}
Virta, J., Nordhausen, K., and Oja, H. (2016).
\newblock Projection pursuit for non-gaussian independent components.
\newblock {\em arXiv preprint arXiv:1612.05445}.

\bibitem[Wei, 2015]{wei2015convergence}
Wei, T. (2015).
\newblock A convergence and asymptotic analysis of the generalized symmetric
  {FastICA} algorithm.
\newblock {\em IEEE Transactions on Signal Processing}, 63(24):6445--6458.

\bibitem[Welvaert et~al., 2011]{welvaert2011a}
Welvaert, M., Durnez, J., Moerkerke, B., Verdoolaege, G., and Rosseel, Y.
  (2011).
\newblock {neuRosim}: An {R} package for generating {fMRI} data.
\newblock {\em Journal of Statistical Software}, 44(10):1--18.

\end{thebibliography}


\begin{thebibliography}{}

\bibitem[Ashburner et~al., 2004]{ashburner2004human}
Ashburner, J., Friston, K., and Penny, W. (2004).
\newblock {Part II -- Imaging Neuroscience -- Theory and Analysis}.
\newblock In Frackowiak, R., editor, {\em Human Brain Function}. Academic
  Press, 2nd edition.

\bibitem[Attias, 1999]{attias1999independent}
Attias, H. (1999).
\newblock Independent factor analysis.
\newblock {\em Neural computation}, 11(4):803--851.

\bibitem[Barch et~al., 2013]{barch2013function}
Barch, D.~M., Burgess, G.~C., Harms, M.~P., Petersen, S.~E., Schlaggar, B.~L.,
  Corbetta, M., Glasser, M.~F., Curtiss, S., Dixit, S., Feldt, C., et~al.
  (2013).
\newblock Function in the human connectome: Task-{fMRI} and individual
  differences in behavior.
\newblock {\em NeuroImage}, 80:169--189.

\bibitem[Beckmann and Smith, 2004]{beckmann2004probabilistic}
Beckmann, C.~F. and Smith, S.~M. (2004).
\newblock Probabilistic independent component analysis for functional magnetic
  resonance imaging.
\newblock {\em IEEE Transactions on Medical Imaging}, 23(2):137--152.

\bibitem[Biswal et~al., 2010]{biswal2010toward}
Biswal, B.~B., Mennes, M., Zuo, X.~N., Gohel, S., Kelly, C., Smith, S.~M.,
  Beckmann, C.~F., Adelstein, J.~S., Buckner, R.~L., Colcombe, S., et~al.
  (2010).
\newblock Toward discovery science of human brain function.
\newblock {\em Proceedings of the National Academy of Sciences},
  107(10):4734--4739.

\bibitem[Blanchard et~al., 2006]{blanchard2006search}
Blanchard, G., Kawanabe, M., Sugiyama, M., Spokoiny, V., and M{\"u}ller, K.-R.
  (2006).
\newblock In search of non-{Gaussian} components of a high-dimensional
  distribution.
\newblock {\em The Journal of Machine Learning Research}, 7:247--282.

\bibitem[Castelli et~al., 2000]{castelli2000movement}
Castelli, F., Happ{\'e}, F., Frith, U., and Frith, C. (2000).
\newblock Movement and mind: a functional imaging study of perception and
  interpretation of complex intentional movement patterns.
\newblock {\em Neuroimage}, 12(3):314--325.

\bibitem[Guo and Tang, 2013]{guo2013hierarchical}
Guo, Y. and Tang, L. (2013).
\newblock A hierarchical model for probabilistic independent component analysis
  of multi-subject {fMRI} studies.
\newblock {\em Biometrics}, 69(4):970--981.

\bibitem[Hastie and Tibshirani, 2010]{hastieR2010}
Hastie, T. and Tibshirani, R. (2010).
\newblock {\em ProDenICA: Product Density Estimation for ICA using tilted
  Gaussian density estimates}.
\newblock R package version 1.0.

\bibitem[Huber, 1985]{huber1985projection}
Huber, P.~J. (1985).
\newblock Projection pursuit.
\newblock {\em The Annals of Statistics}, pages 435--475.

\bibitem[Hyvarinen, 1999]{hyvarinen1999fast}
Hyvarinen, A. (1999).
\newblock Fast and robust fixed-point algorithms for independent component
  analysis.
\newblock {\em IEEE Transactions on Neural Networks}, 10(3):626--634.

\bibitem[Hyv{\"a}rinen et~al., 2001]{hyvarinen2001independent}
Hyv{\"a}rinen, A., Karhunen, J., and Oja, E. (2001).
\newblock {\em Independent component analysis}.
\newblock Wiley-Interscience.

\bibitem[Hyv{\"a}rinen and Oja, 1998]{Hyvaerinen1998}
Hyv{\"a}rinen, A. and Oja, E. (1998).
\newblock Independent component analysis by general nonlinear {Hebbian-like}
  learning rules.
\newblock {\em Signal Processing}, 64(3):301--313.

\bibitem[Hyv{\"a}rinen and Oja, 2000]{hyvarinen2000independent}
Hyv{\"a}rinen, A. and Oja, E. (2000).
\newblock Independent component analysis: algorithms and applications.
\newblock {\em Neural Networks}, 13(4-5):411--430.

\bibitem[Kagan et~al., 1973]{kagan1973characterization}
Kagan, A.~M., Rao, C.~R., and Linnik, Y.~V. (1973).
\newblock {\em Characterization Problems in Mathematical Statistics}.
\newblock Wiley.

\bibitem[Kawanabe et~al., 2007]{kawanabe2007new}
Kawanabe, M., Sugiyama, M., Blanchard, G., and M{\"u}ller, K. (2007).
\newblock A new algorithm of non-{Gaussian} component analysis with radial
  kernel functions.
\newblock {\em Annals of the Institute of Statistical Mathematics},
  59(1):57--75.

\bibitem[Lee et~al., 2011]{lee2011independent}
Lee, S., Shen, H., Truong, Y., Lewis, M., and Huang, X. (2011).
\newblock Independent component analysis involving autocorrelated sources with
  an application to functional magnetic resonance imaging.
\newblock {\em Journal of the American Statistical Association},
  106(495):1009--1024.

\bibitem[Marchini et~al., 2010]{Marchini:2010lr}
Marchini, J.~L., Heaton, C., and Ripley, B.~D. (2010).
\newblock {\em {FastICA}: FastICA Algorithms to perform ICA and Projection
  Pursuit}.
\newblock R package version 1.1-13.

\bibitem[Miettinen et~al., 2017]{miettinen2015squared}
Miettinen, J., Nordhausen, K., Oja, H., Taskinen, S., and Virta, J. (2017).
\newblock The squared symmetric fastica estimator.
\newblock {\em Signal Processing}, 131:402--411.

\bibitem[Miettinen et~al., 2015]{miettinen2015fourth}
Miettinen, J., Taskinen, S., Nordhausen, K., Oja, H., et~al. (2015).
\newblock Fourth moments and independent component analysis.
\newblock {\em Statistical science}, 30(3):372--390.

\bibitem[Nordhausen et~al., 2011]{nordhausen2011deflation}
Nordhausen, K., Ilmonen, P., Mandal, A., Oja, H., and Ollila, E. (2011).
\newblock Deflation-based fastica reloaded.
\newblock In {\em Signal Processing Conference, 2011 19th European}, pages
  1854--1858. IEEE.

\bibitem[Pollard, 2001]{asymptopia}
Pollard, D. (2001).
\newblock Chapter 13 from {Asymptopia} work-in-progress.

\bibitem[Pruim et~al., 2015]{pruim2015ica}
Pruim, R.~H., Mennes, M., van Rooij, D., Llera, A., Buitelaar, J.~K., and
  Beckmann, C.~F. (2015).
\newblock {ICA-AROMA}: a robust {ICA}-based strategy for removing motion
  artifacts from {fMRI} data.
\newblock {\em NeuroImage}, 112:267--277.

\bibitem[Risk et~al., 2014]{risk2014evaluation}
Risk, B.~B., Matteson, D.~S., Ruppert, D., Eloyan, A., and Caffo, B.~S. (2014).
\newblock An evaluation of independent component analyses with an application
  to resting-state {fMRI}.
\newblock {\em Biometrics}, 70(1):224--236.

\bibitem[Salimi-Khorshidi et~al., 2014]{salimi2014automatic}
Salimi-Khorshidi, G., Douaud, G., Beckmann, C.~F., Glasser, M.~F., Griffanti,
  L., and Smith, S.~M. (2014).
\newblock Automatic denoising of functional {MRI} data: combining independent
  component analysis and hierarchical fusion of classifiers.
\newblock {\em Neuroimage}, 90:449--468.

\bibitem[Silva et~al., 2013]{silva2013}
Silva, P.~F., Marcal, A.~R., and da~Silva, R. M.~A. (2013).
\newblock Evaluation of features for leaf discrimination.
\newblock {\em Springer Lecture Notes in Computer Science}, Vol. 7950(197-204).

\bibitem[{v}an~der Vaart, 2000]{van2000asymptotic}
{v}an~der Vaart, A.~W. (2000).
\newblock {\em Asymptotic Statistics}, volume~3.
\newblock Cambridge University Press.

\bibitem[Virta et~al., 2016]{virta2016projection}
Virta, J., Nordhausen, K., and Oja, H. (2016).
\newblock Projection pursuit for non-gaussian independent components.
\newblock {\em arXiv preprint arXiv:1612.05445}.

\bibitem[Wald, 1949]{wald1949note}
Wald, A. (1949).
\newblock Note on the consistency of the maximum likelihood estimate.
\newblock {\em The Annals of Mathematical Statistics}, 20(4):595--601.

\bibitem[Wei, 2015]{wei2015convergence}
Wei, T. (2015).
\newblock A convergence and asymptotic analysis of the generalized symmetric
  {FastICA} algorithm.
\newblock {\em IEEE Transactions on Signal Processing}, 63(24):6445--6458.

\end{thebibliography}

% note: Run pdflatex on Main.tex; then run bibtex on each <include> file; then run pdflatex; then run pdflatex;
\end{document}